\documentclass[aps,prx,superscriptaddress,reprint,longbibliography]{revtex4-2}

\usepackage{graphicx}
\usepackage{dcolumn}
\usepackage{bm, bbm}
\usepackage{amsmath, amsfonts, amssymb}
\usepackage{subfigure}
\usepackage{color, xcolor}
\usepackage{mathdots}
\usepackage{amsthm}
\usepackage{hyperref}
\usepackage{enumerate}
\usepackage{kbordermatrix}

\usepackage{braket}
\usepackage{quantikz}
\usetikzlibrary{quantikz2}

\hypersetup{hypertex=true, colorlinks=true, linkcolor=blue, anchorcolor=blue, citecolor=blue,urlcolor=blue}
\newtheorem{theorem}{Theorem} 
\newtheorem{corollary}{Corollary} 
\newtheorem{lemma}{Lemma}

\newcommand{\tr}{\operatorname{tr}}
\renewcommand{\Re}{\operatorname{Re}}

\newcommand{\id}{\mathbbm{1}}
\newcommand{\Exp}{\mathop{\mathbb{E}}}
\newcommand{\Var}{\mathop{\mathrm{Var}}}
\newcommand{\EPR}{\ket{\mathrm{EPR}}}
\renewcommand{\Pr}{\mathrm{Pr}}

\begin{document}

\title{Relieving the post-selection problem by quantum singular value transformation}

\author{Hong-Yi Wang}
\email{hywang@princeton.edu}
\affiliation{Princeton Quantum Initiative, Princeton University, Princeton, NJ 08544, USA}


\date{\today}

\begin{abstract}
    Quantum measurement is a fundamental yet experimentally challenging ingredient of quantum information processing. Many recent studies on quantum dynamics focus on expectation values of nonlinear observables; however, their experimental measurement is hindered by the post-selection problem---namely, the substantial overhead caused by uncontrollable measurement outcomes. In this work, we propose a post-selection--free experimental strategy based on a fully quantum approach. The key idea is to deterministically simulate the post-selected quantum states by applying quantum singular value transformation (QSVT) algorithms. For pure initial state post-selection, our method is a generalization of fixed-point amplitude amplification to arbitrary projective measurements, achieving an optimal quadratic speedup. We further extend this framework to mixed initial state post-selection by applying linear amplitude amplification via QSVT, which significantly enhances the measurement success probability. However, a deterministic quantum algorithm for preparing the post-selected mixed state is generally impossible because of information-theoretic constraints imposed by quantum coding theory. Additionally, we introduce a pseudoinverse decoder for measurement-induced quantum teleportation. This decoder possesses the novel property that, when conditioned on a successful flag measurement, the decoding is nearly perfect even in cases where channel decoders are information-theoretically impossible. Overall, our work establishes a powerful approach for measuring novel quantum dynamical phenomena and presents quantum algorithms as a new perspective for understanding quantum dynamics and quantum chaos.
\end{abstract}

\maketitle

\section{Introduction}
\label{sec:intro}

Understanding the capability of measurements is one of the cornerstones of the framework of quantum mechanics. Traditionally, we regard the expectation values of Hermitian linear observables as measurable quantities. Hence, the density matrix encapsulates all information about linear observables. However, a more general description of an ``undetermined'' state is to model it as a random ensemble of state vectors, each occurring with some probability. Recent developments in quantum simulators---with the ability to perform site-resolved and trajectory-resolved manipulations---have enabled efficient access to certain non-linear observables (see, e.g., Refs.~\cite{Brydges2019probing,PhysRevLett.124.010504,PhysRevX.12.011018,Bluvstein2024}). This advancement in quantum technologies motivates a deeper understanding of the properties of state ensembles.

State ensembles have emerged as a powerful tool for studying the probabilistic nature of quantum physics. In particular, randomly distributed pure states play important roles in quantum complexity theory \cite{PRXQuantum.2.030316} and black hole physics \cite{HaydenPreskill2007,PhysRevD.90.126007}. They can also be used to demonstrate quantum advantage \cite{Arute2019,PhysRevLett.127.180501} and to benchmark quantum devices \cite{Choi2023}. A natural way to generate a state ensemble is through measurement: the state randomly collapses into different trajectories corresponding to the measurement outcomes. Many interesting phenomena have been identified as unique properties of such measurement-induced state ensembles, also known as \emph{projected ensembles}. For example, starting from a many-body wave function, measuring a subsystem of qubits can induce a state design (i.e., an approximate Haar-random state ensemble) on the remaining subsystem \cite{Choi2023,PRXQuantum.4.010311,PhysRevLett.128.060601,Ippoliti2022solvablemodelofdeep}. Another significant example is the \emph{measurement-induced phase transition} (MIPT) \cite{PhysRevX.9.031009, PhysRevB.99.224307, PhysRevLett.125.030505, PhysRevX.10.041020, PhysRevB.101.104302, PhysRevB.101.104301}, where the transition is reflected in nonanalytic behavior of average entanglement entropy, decodability, learnability, etc. This stands in sharp contrast to traditional phase transitions, which are physically manifested in the discontinuity or nonanalyticity of linear observables, known as order parameters. In contrast, MIPT can only be detected through non-linear observables.

This work addresses a ubiquitous challenge in measuring non-linear observables of a projected ensemble: the post-selection problem. The core issue is that to characterize a state, one must prepare and measure it multiple times. In a projected ensemble, the probability of obtaining the same state even twice is low, necessitating many repetitions of the experiment. Consequently, post-selection leads to an exponential overhead in sample complexity (in the number of measured qubits), posing a severe bottleneck in experimentally detecting MIPTs and other measurement-induced phenomena. Previously, the mainstream approach to circumvent the post-selection problem has been to replace the non-linear observable with a quantum-classical correlation \cite{PhysRevLett.125.070606, PhysRevLett.130.220404, PRXQuantum.5.030311}. This strategy has been successfully employed in experiments to demonstrate MIPTs \cite{Noel2022,Hoke2023,kamakari2024scalablecrossentropy}, but fundamentally depends on the ability to classically simulate the quantum dynamics.

In this work, we propose a genuinely quantum solution to the post-selection problem, which does not rely on classical simulation at all. The key insight is that, whenever we have a state preparation protocol with low probability (i.e., preparation via post-selection), we can transform it into a deterministic preparation protocol, albeit at the cost of a deeper circuit. The simplest case involves post-selection on a pure state, which can be achieved using Grover's search algorithm \cite{grover1996fast}. For the best stability, we implement fixed-point amplitude amplification (FPAA) \cite{PhysRevLett.113.210501} within the framework of quantum singular value transformation (QSVT) \cite{GilyenQSVT,PRXQuantum.2.040203}. At a high level, QSVT serves as a unifying quantum algorithm for many tasks with known quantum advantage, and thus underpins all algorithms proposed in this work. In our context, the post-selected state is naturally embedded as a block of the unitary evolution matrix. This structure, known as block encoding, is central to many quantum algorithms including QSVT. Notably, unlike other uses of QSVT that require deliberate construction of a block encoding for the target matrix, here the block encoding arises naturally from the evolution combined with the measurement projector. This FPAA algorithm provides an optimal square-root speedup over naive post-selection, akin to Grover's speedup argument \cite{BBBV1997}. With FPAA, we can thus perform experiments without incurring sample-complexity overhead from post-selection, with the tradeoff being the need for longer state preparation circuits.

Our approach becomes even more interesting when extended to mixed initial states. The question is: can we transform a mixed state into its corresponding post-selected state deterministically? Unfortunately, there is generally no deterministic quantum algorithm to do so. However, QSVT remains valuable because it allows for significant enhancement of the state preparation success probability. While an information-theoretic bound limits the maximum achievable success probability, linear amplitude amplification (LAA) by QSVT can yield an $O(1)$ success probability for generic mixed states. By tuning a single parameter in LAA, we can continuously trade off between high success probability and high fidelity of state transformation. The performance of our LAA algorithm depends on how ``uninformative'' the measurement is about the pure states hidden within the mixed state. Roughly, the less information the measurement reveals about the pure state, the higher the success probability. In the extreme case of a completely uninformative measurement, our algorithm becomes deterministic.

The above-mentioned task of mixed state post-selection bears strong structural similarity to measurement-induced quantum state teleportation \cite{PhysRevLett.70.1895,PhysRevLett.132.030401,Antonini2022}. It follows that the QSVT-based algorithm also inspires a novel and intuitive decoder for quantum teleportation. We identify the mapping from the initially encoded state to the post-selected state as a non-unitary linear map that comes naturally with a block encoding description. In this framework, our decoder effectively amounts to applying the Moore--Penrose pseudoinverse, which is known to be implementable via QSVT. The decoder's performance is qualitatively similar to that of the LAA algorithm in mixed state post-selection. First, a deterministic (quantum channel) decoder does not exist unless the decoupling condition is satisfied \cite{hayden2008decoupling}. However, a unique feature of our pseudoinverse decoder is that, even when the decoupling condition is violated, the decoding fidelity can still approach one. This is because our algorithm sacrifices success probability in favor of higher decoding fidelity. This decoder favors uninformative measurements as well. Furthermore, this decoder can also be directly applied to the decoding task in approximate quantum error correction codes. In this application, the violation of the decoupling condition is guaranteed to be small, and it follows that our pseudoinverse decoder can succeed with $O(1)$ probability. 

Finally, all the quantum algorithms introduced above lead us to a conjectured tradeoff relation in terms of post-selected state complexity. This arises from a comparison between the quantum-classical correlation approach and our amplitude amplification approach. We conjecture that, to overcome the post-selection problem, one must either classically simulate all quantum trajectories or implement complex quantum state preparation circuits. Furthermore, such a tradeoff may be a general feature across a wide class of state transformation problems, as the decoding task similarly exhibits a dichotomy between classical simulation approaches and quantum-intensive approaches like our pseudoinverse decoder.

The remainder of this paper is organized as follows. In Sec.~\ref{sec:previous}, we introduce the post-selection problem, review quantum-classical correlation approaches, and discuss their inherent limitations. Section~\ref{sec:FPAA} presents our QSVT-based amplitude amplification algorithm, which deterministically simulates post-selected pure states with optimal quadratic speedup. In Sec.~\ref{sec:protocol}, we propose a post-selection--free experimental protocol that uses fixed-point amplitude amplification to efficiently estimate nonlinear observables. Section~\ref{sec:sequential} extends our framework to sequential measurements by recasting mid-circuit projective measurements into a final post-selection scheme. Section~\ref{sec:mixed} generalizes our approach to mixed state post-selection via linear amplitude amplification. In Sec.~\ref{sec:teleportation}, we develop a pseudoinverse decoder for measurement-induced teleportation that achieves near-perfect decoding fidelity. Finally, Sec.~\ref{sec:tradeoff} proposes a tradeoff relation between classical simulability and quantum complexity, and Sec.~\ref{sec:summary} summarizes our findings while outlining promising directions for future research.

\section{The post-selection problem: previous attempts}
\label{sec:previous}

Let us begin with a minimalistic version of our problem. On a quantum device, we prepare a state $\ket{\psi}$ and subsequently perform a projective measurement described by the projectors $\{\Pi_m\}$. The post-measurement state corresponding to outcome $m$ is given by
\begin{equation}
    \ket{\psi_m} = \frac{\Pi_m\ket{\psi}}{\sqrt{p_m}} = \frac{\Pi_m\ket{\psi}}{\sqrt{ \braket{\psi|\Pi_m|\psi} }},
\end{equation}
where $p_m=\braket{\psi|\Pi_m|\psi}$. This is illustrated in Fig.~\ref{fig:post_to_AA}. Although we depict the measurement as a computational basis measurement on a subsystem in Fig.~\ref{fig:post_to_AA}(b), as we will continue to do throughout, this is not essential: all approaches and reasoning apply equally to general projectors of arbitrary rank acting on the full system. The measurement converts the initial state into a random ensemble of pure states, called the projected ensemble, which we denote as a collection of probability--state pairs:
\begin{equation}    \label{eq:proj_ensemble}
    \mathcal{E} = \{ p_m, \ket{\psi_m} \}_m   .
\end{equation}
The setup of measurement-induced phase transitions is different in that measurements occur mid-circuit, but it is always possible to push all measurements to the final time (see Sec.~\ref{sec:protocol}), thereby interpreting such ensembles as projected ensembles described above. The procedure used to generate this ensemble is inherently probabilistic: while one can sample from the ensemble, the specific outcome obtained in each experimental run cannot be controlled.

A state ensemble contains more information than the density matrix. If we study the ensemble average of linear observables, these are merely properties of the density matrix. In contrast, we are particularly interested in non-linear observables, which are exclusive properties of a state ensemble. For example, a frequently studied non-linear observable is the subsystem purity:
\begin{equation}    \label{eq:purity}
    P_{Am} = \tr_{A} \left( \rho_{Am} \right)^2   ,
\end{equation}
where $\rho_{Am} = \tr_{\bar{A}} \ket{\psi_m} \bra{\psi_m}$. We are ultimately interested in the ensemble average $P_{A} = \sum_m p_m P_{Am}$. To obtain an unbiased estimate of $P_{Am}$, one must prepare the state $\psi_m$ at least twice---for instance, by a SWAP test between two copies of $\rho_{Am}$. In contrast, any quantity like $\braket{\psi_m | O | \psi_m}$ cannot be an unbiased estimate of purity. However, if many qubits are measured, each probability $p_m$ becomes exponentially small (in the number of qubits measured), so once we obtain $\psi_m$, we cannot expect to encounter it again in a feasible timeframe. This experimental conundrum in measuring non-linear observables is referred to as the \emph{post-selection problem}.

\begin{figure*}[t]
    \centering
    \includegraphics[width=\linewidth]{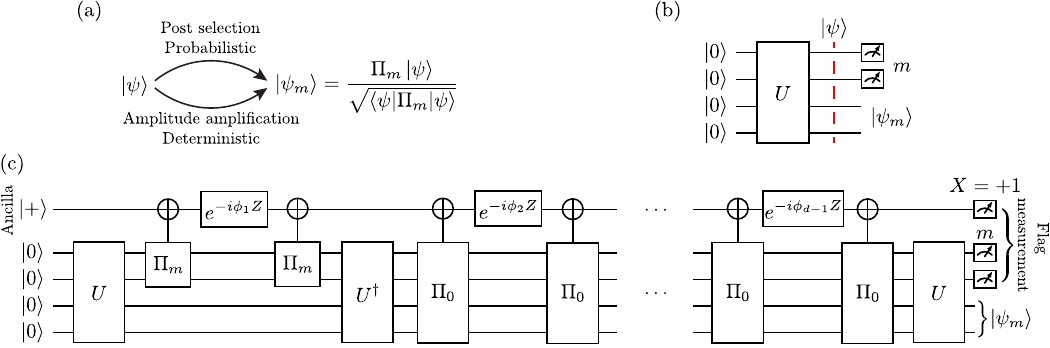}
    \caption{Illustration of the main idea. (a) When there exists a measure-and-post-select protocol to map $\ket{\psi}$ to $\ket{\psi_m}$, it can be replaced by a deterministic protocol via amplitude amplification. (b) Circuit diagram for $\ket{\psi}$ state preparation and post-selection. (c) The fixed-point amplitude amplification circuit that equivalently produces $\ket{\psi_m}$.}
    \label{fig:post_to_AA}
\end{figure*}

To overcome the post-selection problem, the currently prevailing approach is to replace the non-linear quantity with a \emph{quantum-classical correlation}. This idea is inspired by the cross-entropy benchmark (XEB), used as a proxy for state-preparation fidelity in quantum supremacy experiments \cite{Boixo2018,Arute2019,RevModPhys.95.035001}. Again, we illustrate this idea using the example of measuring the ensemble-averaged purity. Crudely describing the strategy, each time we obtain a measurement outcome $m$, instead of recording $\tr(\rho_{Am}^2)$, we record the value
\begin{equation}    \label{eq:QCcorrelation}
    \tr(\rho_{Am}^Q \rho_{Am}^C)    ,
\end{equation}
where $\rho_{Am}^Q$ is from the quantum experiment, and $\rho_{Am}^C$ is obtained from classical simulation of the same dynamics \cite{PhysRevLett.125.070606,Hoke2023,PRXQuantum.5.030311}. As long as the classical simulation is accurate, each record provides an unbiased estimate of the ensemble-averaged purity. Hence, all measurement records are useful, and the post-selection problem is circumvented. Classical simulation can be performed using exact matrix evaluation \cite{Hoke2023}, the stabilizer formalism \cite{PhysRevLett.130.220404,kamakari2024scalablecrossentropy}, or approximate methods such as matrix-product states (MPS) \cite{PRXQuantum.5.030311} or machine learning \cite{Kim2025talk, You2025toappear}. A practical method to measure Eq.~\eqref{eq:QCcorrelation} is to further replace $\rho_{Am}^Q$ with a classical shadow \cite{Huang2020}, which serves as an unbiased estimator of $\rho_{Am}^Q$ and is based on single-shot measurements. Other non-linear quantities beside the purity can be estimated in a similar manner.

A serious disadvantage of the quantum-classical correlation approach is its heavy reliance on the ability to classically simulate quantum dynamics. However, if the dynamics can be classically simulated, the necessity of conducting the quantum experiment is diminished---except perhaps as a consistency check between quantum and classical descriptions. Therefore, to perform experiments with broader applicability, one must be able to quantum-simulate the post-measurement state itself: upon obtaining a measurement outcome $m$, we should be able to reproduce the state $\psi_m$ multiple times. This motivates the use of the quantum algorithm of amplitude amplification, which we introduce in the next section.

In addition to quantum-classical correlation, another existing approach to mitigate post-selection overhead is to study non-unitary dynamics that are space--time dual to unitary dynamics \cite{PhysRevLett.126.060501,Hoke2023}. In such models, by realizing the dual unitary evolution, one can access subsystem purities without any explicit post-selection in the experiment. However, the class of non-unitary circuits that are dual to unitary ones is quite limited. Moreover, even within this class, many observables beyond purity are not dual to post-selection--free observables. Therefore, it remains highly desirable to develop a general method for simulating post-selected states.

\section{Amplitude amplification by QSVT}
\label{sec:FPAA}

As motivated above, our goal is to replace the probabilistic measurement protocol that maps $\ket{\psi} \to \ket{\psi_m}$ with a deterministic one. In this section, we restrict the discussion to pure states subjected to projective measurements. Mixed-state post-selection will be addressed in Secs.~\ref{sec:mixed} and \ref{sec:teleportation}, where we demonstrate that a similar approach applies, albeit with certain caveats.

This computational task is known as amplitude amplification. Initially, the state $\ket{\psi}$ has a (typically small) overlap with the target subspace defined by $\Pi_m$,
\begin{equation}
    p_m = \braket{\psi | \Pi_m | \psi} \ll 1.
\end{equation}
Our objective is to apply a sequence of unitary operations to $\ket{\psi}$ such that the resulting state attains an overlap close to unity with the post-measurement state $\ket{\psi_m}$, which resides in the subspace defined by $\Pi_m$. In this way, amplitude amplification deterministically simulates the post-measurement state.

In fact, this goal can be achieved using Grover's search algorithm---famous for its quadratic speedup over classical algorithms \cite{grover1996fast,BBBV1997}. We briefly review Grover's algorithm to motivate the subsequent approaches. The algorithm repeatedly applies two unitary reflection operators, $R_m$ and $R_i$, to the initial state $\ket{\psi}$:
\begin{equation}
    R_i = 2\ket{\psi}\bra{\psi} - \id = 2U \ket{0}\bra{0} U^\dagger - \id,
\end{equation}
\begin{equation}
    R_m = 2\Pi_m - \id   .
\end{equation}
The net effect of $R_i R_m$ is to rotate the state $\ket{\psi}$ towards $\ket{\psi_m}$. The rotation angle is approximately $2\sqrt{p_m}$ in the small-$p_m$ limit. Note that the rotation always occurs within the plane spanned by $\ket{\psi}$ and $\ket{\psi_m}$. Consequently, $O(1/\sqrt{p_m})$ iterations are required to reach $\ket{\psi_m}$. In Appendix~\ref{app:lowerbound}, we provide a proof that this quadratic speedup is optimal for amplitude amplification tasks with general multi-rank target projectors.

Although it is tempting to directly apply the above strategy to overcome the post-selection problem, Grover's algorithm has a few disadvantages. First, it requires prior knowledge of the rotation angle---equivalent to knowing $p_m$ in advance---to determine the appropriate number of iterations. Over-rotating or under-rotating due to incorrect iteration counts significantly reduces the fidelity with $\ket{\psi_m}$, a problem previously referred to as the ``soufflé problem'' \cite{Brassard1997}. Second, the final infidelity cannot be reduced arbitrarily, as the iteration count must be an integer, and there may not exist an integer that corresponds to an exact $\pi/2$ rotation. Both of these issues are resolved by the FPAA algorithm \cite{PhysRevLett.113.210501}, which can achieve any desired fidelity as long as the initial overlap exceeds a certain threshold. Furthermore, to integrate this approach with mixed-state post-selection in subsequent sections, we implement FPAA using quantum singular value transformation (QSVT) \cite{PRXQuantum.2.040203}. This implementation is as efficient as the original FPAA and offers greater interpretability and generalizability.

QSVT aims to apply matrix functions on possibly non-unitary or even non-square matrices, which we denote as $M$. How is such an oracle $M$ implemented in quantum circuits? This is achieved through a strategy called block encoding. The matrix $M$ is embedded as a block within a unitary matrix $U$:
\begin{equation}
    U = 
    \kbordermatrix{
    & \Pi &  \\
    \tilde{\Pi} & M & \cdot \\
    & \cdot & \cdot
  }
\end{equation}
Equivalently, we can write $M = \tilde{\Pi} U \Pi$, where $\Pi$ and $\tilde{\Pi}$ are projectors. We will frequently use the relative phase gate between the linear subspace $\operatorname{span}(\Pi)$ and its complement, defined as 
\begin{equation}    \label{eq:rel_rot}
    \Pi_{\phi} = e^{i\phi} \Pi + e^{-i\phi} (\id-\Pi)
    = e^{i\phi (2\Pi - \id)}   ,
\end{equation}
and similarly for $\tilde{\Pi}$. As a remark, with one ancilla qubit, $\Pi_\phi$ can be synthesized using two $\mathrm{C_{\Pi}NOT}$ gates and a single-qubit ancilla rotation $e^{-i \phi Z}$ between them (see Fig.~\ref{fig:post_to_AA}(c) and Appendix~\ref{app:QSVT} for details).

Singular value transformations are defined as follows. Let the singular value decomposition (SVD) of $M$ be $\sum_i s_i \ket{u_i} \bra{v_i}$. For an odd polynomial $P_d(x)$, its action on a matrix is defined as
\begin{equation}
    P_d(M) = \sum_i P_d(s_i) \ket{u_i} \bra{v_i}   .
\end{equation}
For example, if $P_d(x) = x^3+x$, then $P_d(M) = M M^\dagger M + M$. Actions of even polynomials can also be defined with slight modifications, but we will not use them in this work. One of the main results of QSVT is that any real odd polynomial satisfying $|P_d(x)| \leq 1$ for $x \in [-1,1]$ can be implemented via the so-called \emph{alternating phase modulation sequence}:
\begin{equation}    \label{eq:QSVT_odd}
    \begin{aligned}
        U_{\vec{\phi}} &\equiv \tilde{\Pi}_{\phi_1} U 
        \prod_{k=1}^{(d-1)/2} \left( \Pi_{\phi_{2k}} U^\dagger \tilde{\Pi}_{\phi_{2k+1}} U \right) \\
        &= \kbordermatrix{
        & \Pi &  \\
        \tilde{\Pi} & \tilde{P}_d(M) & \cdot \\
        & \cdot & \cdot
    },
    \end{aligned}
\end{equation}
where $\tilde{P}_d$ is a complex polynomial such that $\Re \tilde{P}_d = P_d$. The phase sequence $\vec{\phi}$ can be determined entirely from the polynomial $P_d$, independent of the specific matrix $M$. There exist efficient and numerically stable classical algorithms for computing $\vec{\phi}$ \cite{PRXQuantum.2.040203,Dong_2021}. Although QSVT only implements polynomial transformations, it is important to note that many useful functions can be approximated by polynomials.

Within the QSVT framework, Grover's search algorithm emerges as a special case, where the reflections $R_i$ and $R_m$ are both of the form $\Pi_{\pi/2}$ up to a global phase. In fact, Grover's algorithm corresponds to realizing Chebyshev polynomials. However, we can design improved polynomial transformations by carefully choosing the phase sequences \footnote{As a remark, in Fig.~\ref{fig:post_to_AA}(c) and all subsequent QSVT implementations, we realize an odd degree-$d$ polynomial but use only $(d-1)$ applications of $(\Pi_0)_\phi$ and $(\Pi_m)_\phi$. This is because the final $(\Pi_m)_\phi$ acts like a global phase within the $\Pi_m$ subspace and can thus be omitted.}. 

We now interpret amplitude amplification as the search for a particular transformation on an encoded matrix. The encoded matrix is
\begin{equation}   \label{eq:post_block_enc}
    M = \Pi_m U \Pi_0 = \Pi_m \ket{\psi}\bra{0} = \sqrt{p_m} \ket{\psi_m} \bra{0}   ,
\end{equation}
where $\Pi_0 = \ket{0}\bra{0}$ is the projector onto the initial state. $M$ has only one nonzero singular value, $\sqrt{p_m}$, and the goal is to amplify it to a value close to $1$, since if it becomes $1$, we can apply $\ket{\psi_m}\bra{0}$ on $\ket{0}$ to get $\ket{\psi_m}$. Figure~\ref{fig:post_to_AA}(c) shows the circuit diagram for QSVT acting on $M$ as defined in Eq.~\eqref{eq:QSVT_odd}. Whereas Grover's algorithm implements the transformation $x \to T_d(x)$ (the Chebyshev polynomial of the first kind, which is highly oscillatory), a more suitable alternative is an approximate sign function. Specifically, we seek an odd polynomial satisfying:
\begin{enumerate}[(1)]
    \item $\left| P_d(x) \right| \leq 1$ for $x \in [-1,1]$;
    \item $1 - P_d(x) \leq \delta$ for $x \in [\sqrt{p^*}, 1]$, where $p^*$ is a (typically small) positive number.
\end{enumerate}
As long as $p_m \geq p^*$, the FPAA succeeds with high probability. Crucially, this method does not require prior knowledge of the exact value of $p_m$. As a classic result in polynomial approximation theory, there exists a polynomial of degree $d = O\left(\frac{1}{\sqrt{p^*}} \log \frac{1}{\delta}\right)$ satisfying the above conditions \cite{GilyenQSVT,Low2017QSP}. Hence, the query complexity (i.e., the number of queries to the block encoding $U$ and $U^\dagger$) is also $O\left(\frac{1}{\sqrt{p^*}} \log \frac{1}{\delta}\right)$. Compared to the query complexity of Grover's algorithm, FPAA via QSVT achieves optimal scaling with respect to $p^*$ as well, albeit with a slightly larger constant factor, since Grover's iteration achieves the fastest possible convergence (see Appendix~\ref{app:lowerbound}). On the other hand, FPAA has the advantage of offering precise control over the error $\delta$ through moderately higher-degree polynomials.

Finally, we address countermeasures for the residual error $\delta$ in FPAA. For general QSVT algorithms, as indicated by Eq.~\eqref{eq:QSVT_odd}, the outcome is $P_d(M)\ket{\varphi}$ only if $\ket{\varphi} \in \operatorname{span}(\Pi)$ and the final state is post-selected onto $\tilde{\Pi}$. We refer to the measurements that indicate success in QSVT as \emph{flag measurements}. In Fig.~\ref{fig:post_to_AA}(c), in addition to the $\Pi_m$ measurement, there is a flag measurement on the ancilla qubit requiring the $\ket{+} = (\ket{0} + \ket{1})/\sqrt{2}$ state. This arises from the procedure of extracting the real part, similar to the Hadamard test (see Appendix~\ref{app:QSVT} for details). However, FPAA differs in that the flag measurement is, in principle, unnecessary: the probability of flag failure is small, bounded as $1 - P_d(\sqrt{p_m})^2 < 2\delta$. Nevertheless, one may still perform the flag measurement to ensure that the post-selected state is precisely $\ket{\psi_m}$.

\section{Post-selection-free experiment}
\label{sec:protocol}

With the FPAA algorithm in hand, we now present a strategy to measure non-linear observables of a projected ensemble.

As before, we denote the projected ensemble of interest as in Eq.~\eqref{eq:proj_ensemble}. Let us consider the task of measuring an observable that depends on up to the $k$th moment:
\begin{equation}    \label{eq:non_linear_O}
    \braket{\mathcal{O}} = \sum_m p_m \tr \left( \left( \ket{\psi_m}\bra{\psi_m} \right)^{\otimes k} \mathcal{O} \right)    ,
\end{equation}
where $\mathcal{O}$ is a Hermitian operator acting on $k$ copies of the system. The goal is to estimate $\braket{\mathcal{O}}$ with high accuracy. Ignoring the post-selection problem for a moment, if we want to estimate 
\begin{equation}
    \braket{\mathcal{O}}_m = \tr \left( \left( \ket{\psi_m}\bra{\psi_m} \right)^{\otimes k} \mathcal{O} \right)
\end{equation}
for a single state $\ket{\psi_m}$, there exist unbiased estimators that use $k$ copies of the state. Let $\hat{\mathcal{O}}$ denote an experimental protocol that takes as input $k$ copies of $\ket{\psi_m}$ and outputs an unbiased estimate of $\braket{\mathcal{O}}_m$ via measurement and classical post-processing. For example, if the observable is subsystem purity, commonly used approaches include the SWAP test \cite{PhysRevLett.87.167902} and randomized measurement protocols \cite{PhysRevLett.108.110503,Huang2020,Elben2023toolbox}---each with its own advantages and limitations. Any of these methods can be incorporated into the post-selection--free experiment described below.

We begin by fixing the parameters $p^*$ and $\delta$ in the FPAA algorithm and calculating the corresponding QSVT phase sequence. Roughly, $p^*$ should be smaller than most of the $p_m$ values, and $\delta$ should be chosen to be smaller than but of the same order as the allowed error for estimating $\braket{\mathcal{O}}$. In particular, since decreasing $\delta$ is computationally inexpensive while reducing $p^*$ incurs a significantly higher query complexity, it is crucial to obtain a theoretical estimate of the distribution of $p_m$ whenever possible. The post-selection--free experimental protocol proceeds as follows:
\begin{enumerate}[(i). ]
    \item Apply the original probabilistic protocol (i.e., measurement) to randomly sample a post-selected state $\ket{\psi_m}$ and record its classical outcome $m$.
    \item Prepare another copy of $\ket{\psi_m}$ using FPAA \footnote{In general, if the experimental protocol $\hat{\mathcal{O}}$ requires coherent measurements across all copies, then we must maintain $k$ copies in $k$ identical quantum registers. However, if no coherent measurement is required---as in the case of randomized or adaptive measurements---then one quantum register can be reused to prepare $\ket{\psi_m}$ multiple times.}. Perform the flag measurement; if it fails, return to step (i).
    \item Repeat step (ii) $(k-1)$ times to obtain $k$ copies of the state $\ket{\psi_m}$.
    \item Apply the protocol $\hat{\mathcal{O}}(\ket{\psi_m}^{(1)}, \dots , \ket{\psi_m}^{(k)})$ to obtain an unbiased estimate of $\braket{\mathcal{O}}_m$.
    \item Repeat from step (i), and finally average over many $m$'s to obtain an estimate of $\braket{\mathcal{O}}$.
\end{enumerate}

In this strategy, the flag measurement serves as a sanity check for the chosen parameters $p^*$ and $\delta$ in the FPAA algorithm. To explain this, we focus on the case of quadratic observables, i.e., $k=2$ in Eq.~\eqref{eq:non_linear_O}, and assume that $\|\mathcal{O}\|_\infty = 1$, which is true for the purity operator. Due to the flag measurements, each state $\ket{\psi_m}$ is sampled with probability $p_m \cdot P_d(x)^2$ rather than $p_m$. This introduces a sampling bias in the estimate of $\braket{\mathcal{O}}$, as states with low $p_m$ are underrepresented. It is not difficult to show that this bias is upper bounded by $2\delta \Pr(p_m \geq p^*) + \Pr(p_m < p^*)$. The failure probability of the flag measurement is also bounded by the same expression. This justifies our principle for selecting appropriate values of $p^*$ and $\delta$.

Thus, by performing flag measurements, we can monitor and adjust $p^*$ and $\delta$ to control the total estimation error. Concretely, suppose the target accuracy for estimating $\braket{\mathcal{O}}$ is $\varepsilon$. One source of error arises from statistical fluctuations across different $m$ values, which, by the central limit theorem, scale as $\Var_m \braket{\mathcal{O}}_m/\sqrt{N}$, where $N$ is the number of $\braket{\mathcal{O}}_m$ samples. This form of error can only be suppressed by increasing the number of samples, which we do not further address here. Instead, we focus on the error due to under-sampling low-probability states. By monitoring the frequency of failed flag measurements, we obtain an empirical failure rate $f_{\mathrm{fail}}$. If $f_{\mathrm{fail}} \gtrsim \varepsilon$, it indicates that either $p^*$ or $\delta$ is too large. In that case, one must reduce $p^*$ or $\delta$ and repeat the experiment.


\section{Sequential measurements}
\label{sec:sequential}

\begin{figure}
    \centering
    \includegraphics[width=\linewidth]{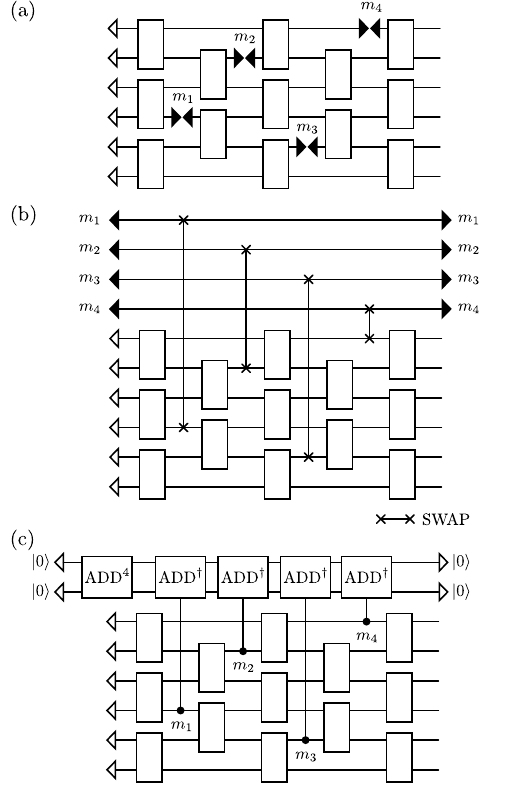}
    \caption{(a) Example of a quantum circuit with two-qubit gates and mid-circuit projective measurements. White triangles denote $\ket{0}$ states, and black triangles represent measurement outcome states $\ket{m_j}$, where $m_j = 0$ or $1$, $j = 1,2,\dots, N_{\mathrm{meas}}$. In the figure, $N_{\mathrm{meas}} = 4$. (b) and (c) depict different block encodings of (a). (b) The equivalent circuit with all post-selections postponed to the final time by introducing $N_{\mathrm{meas}}$ ancilla qubits. (c) The compression gadget circuit, which requires only $\left\lceil \log_2 N_{\mathrm{meas}} \right\rceil$ ancilla qubits acting as a counter.}
    \label{fig:mipt_block_enc}
\end{figure}

MIPT (in terms of entanglement, teleportation, learnability, etc.) is among the most well-studied phenomena of projected ensembles. However, in many MIPT setups, the circuit is hybrid in nature, meaning that measurements occur at different times interleaved by unitary gates. This structure renders the amplitude amplification approach from Sec.~\ref{sec:FPAA} not directly applicable. Thus, we seek methods to convert such sequential mid-circuit measurements into measurements at the final time.

To be concrete, consider a hybrid circuit with post-selected outcomes as shown in Fig.~\ref{fig:mipt_block_enc}(a). Let the number of single-qubit projective measurements be $N_{\mathrm{meas}}$. For each measurement, the outcome $m_j = 0$ or $1$ corresponds to the Kraus operator $\ket{m_j}\bra{m_j}$. The resulting tensor network, such as that depicted in Fig.~\ref{fig:mipt_block_enc}(a), yields the unnormalized post-selected state $\sqrt{p_m} \ket{\psi_m}$, where $m$ is a collective label for the bitstring $(m_1, m_2, \dots, m_{N_{\mathrm{meas}}})$. As in Sec.~\ref{sec:FPAA}, we view this state as a $1 \times D$ matrix, where $D$ is the Hilbert space dimension. In the language of quantum algorithms, we seek block encodings for this post-selected state transformation.

One method of block encoding is shown in Fig.~\ref{fig:mipt_block_enc}(b), where $N_{\mathrm{meas}}$ ancilla qubits are introduced, and all post-selections are postponed to the final time. In fact, Figs.~\ref{fig:mipt_block_enc}(a) and (b) are entirely equivalent from a tensor network perspective. The key identity is:
\begin{equation}
    \vcenter{\hbox{\includegraphics{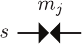}}} = \vcenter{\hbox{\includegraphics{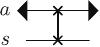}}},
\end{equation}
where the labels $a$ and $s$ denote the ancilla and system qubits, respectively. Explicitly, this equality reads:
\begin{equation}
    \ket{m_j} \bra{m_j}_s
    = \left( \bra{m_j}_a \otimes \id_s \right) \mathrm{SWAP} \left( \ket{m_j}_a \otimes \id_s \right).
\end{equation}

A less straightforward yet more resource-efficient block encoding uses the so-called ``compression gadget'' method \cite{low2019interactionpicture,Fang_2023}. Instead of requiring $N_{\mathrm{meas}}$ ancilla qubits, this method uses only $\left\lceil \log_2 N_{\mathrm{meas}} \right\rceil$ ancilla qubits. It is particularly advantageous for deep hybrid circuits with large $N_{\mathrm{meas}}$. The compression gadget circuit is shown in Fig.~\ref{fig:mipt_block_enc}(c). Here, we introduce an $N_{\mathrm{meas}}$-dimensional ancilla system (or equivalently, $\left\lceil \log_2 N_{\mathrm{meas}} \right\rceil$ qubits) acting as a coherent counter. On the counter system, we define the ADD gate:
\begin{equation}
    \mathrm{ADD} \ket{i} = \ket{i+1} \mod N_{\mathrm{meas}},
\end{equation}
and its controlled version:
\begin{equation}
    \begin{aligned}
         &\phantom{=} \mathrm{C}_{m_j} \mathrm{ADD}
            = \vcenter{\hbox{\includegraphics{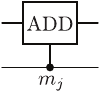}}}   \\
            &= \mathrm{ADD} \otimes \ket{m_j}\bra{m_j} + \id \otimes \ket{\overline{m_j}} \bra{\overline{m_j}}.
    \end{aligned}
\end{equation}
The exponential saving in ancilla qubits is possible because we do not need to record the full bitstring $(m_1, m_2, \dots, m_{N_{\mathrm{meas}}})$; it suffices to verify whether all measurement outcomes are successful, i.e., whether the string matches a specific target string. It is not hard to show that the counter being in the $\ket{0}$ state at the end is equivalent to the success of all measurements.

The idea of the compression gadget is not limited to sequential measurements; it also applies to ``compressing'' simultaneous measurements. As a starting point, observe that the gadget still functions if the unitary gates between measurements are removed. Specifically, consider the case where many disconnected qubits are measured at the same time. Even though the single-qubit measurements are local, the corresponding controlled-phase gate $\Pi_\phi$, defined in Eq.~\eqref{eq:rel_rot}, becomes a non-local operation over all the measured qubits, which can be difficult to implement. However, by using the compression gadget circuit together with QSVT, each measured qubit interacts only with the counter system, which then indirectly mediates the entangling operations among the measured qubits. This renders the implementation of controlled-phase gates much more tractable.

In the subsequent sections, we will consider mixed initial states, for which the block encoding becomes a many-to-many dimensional mapping rather than one-to-many. Nonetheless, both the SWAP-based block encoding and the compression gadget can be straightforwardly generalized to the mixed-state case.




\section{Simulation of mixed state measurement}
\label{sec:mixed}

It is natural to further investigate whether our amplitude amplification approach also applies to mixed-state post-selection. Specifically, the task in this section is: given a mixed state $\rho$, find a deterministic quantum algorithm (i.e., a unitary or quantum channel) that maps it to
\begin{equation}
    \rho_m = \frac{1}{p_m} \Pi_m \rho \Pi_m
\end{equation}
with high fidelity, where $p_m = \tr(\rho \Pi_m)$. Unfortunately, it turns out that in general, such a transformation cannot be realized by a unitary operator or channel. Nevertheless, using QSVT, we develop an algorithm that can significantly amplify the success probability of post-selection.

\begin{figure}
    \centering
    \includegraphics[width=\linewidth]{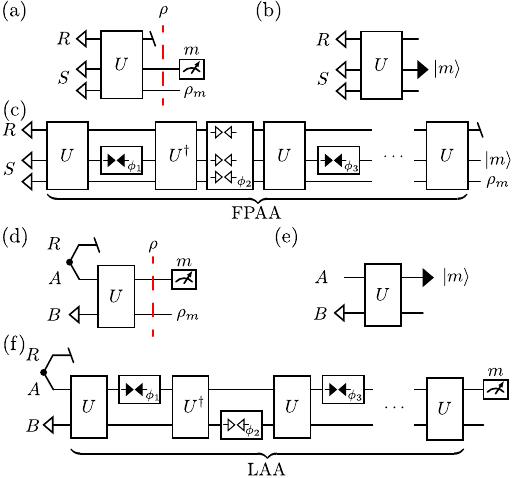}
    \caption{(a--c) illustrate the case where the mixed state $\rho$ is prepared via purification. (a) The state preparation protocol and the post-selected state $\rho_m$. (b) Tensor diagram of the corresponding block-encoded matrix. (c) FPAA circuit simulating the post-selected state. (d--f) illustrate the case where the mixed state is prepared without access to a purification. (d) The state preparation protocol and the resulting post-selected state $\rho_m$. (e) Tensor diagram of the corresponding block-encoded matrix. (f) LAA circuit simulating the post-selected state, conditioned on the flag measurement $\Pi_m$ being successful. In (c) and (f), the $\Pi_\phi$ gates are composed of $\mathrm{C_{\Pi}NOT}$ gates and ancilla qubit rotations. For clarity, the internal structure of $\Pi_\phi$ and the ancilla qubit are omitted.}
    \label{fig:mixed}
\end{figure}

We discuss two different setups of mixed state post-selection. Consider a mixed initial state $\rho$ with eigendecomposition $\rho = \sum_a \lambda_a \ket{\psi_a}\bra{\psi_a}$ on system $S$. The first setup assumes access to a purification. That is, a pure state $\ket{\Psi}_{RS}$ such that $\rho = \tr_R \ket{\Psi}\bra{\Psi}_{RS}$, and a unitary $U$ such that $U \ket{0}_R \ket{0}_S = \ket{\Psi}_{RS}$. The mixed state is then prepared by evolving $\ket{0}_R \ket{0}_S$ under $U$ and tracing out $R$. In this scenario, the post-measurement state can be simulated in the same way as in the pure-state case (see Fig.~\ref{fig:mixed}(a--c)): apply the FPAA sequence using $\ket{\Psi}\bra{\Psi}_{RS}$-controlled and $\Pi_m$-controlled phase gates, then trace out $R$ at the end.

However, this setup is often unrealistic. In practice, the reference system $R$ (e.g., an environment) may be too complicated to manipulate, or is intentionally hidden from the experimenter. So for the rest of this section, we consider a scenario where only the system $S$ can be manipulated. We adopt the following specific preparation protocol:
\begin{equation}    \label{eq:mix_init_state}
    \rho = U \left( \frac{\id_A}{d_A} \otimes \ket{0}\bra{0}_B \right) U^\dagger,
\end{equation}
where the system $S$ is bipartitioned into $A$ and $B$. With a fictitious reference system $R$ such that $d_R = d_A$, the state can be purified as $\ket{\Psi}_{RS} = \frac{1}{\sqrt{d_R}} \sum_{a=1}^{d_R} \ket{a}_R \otimes \ket{\psi_a}_S$. This construction fixes the eigenvalues of the density matrix, making them uniform and unadjustable. The corresponding circuit is shown in Fig.~\ref{fig:mixed}(d).

The state preparation together with the projector defines a block-encoded matrix to which we can apply QSVT, as in the pure-state case. Let $\Pi_B = \id_A \otimes \ket{0}\bra{0}_B$ be the projector onto the initial state. Then the block-encoded matrix is (see Fig.~\ref{fig:mixed}(e)):
\begin{equation}    \label{eq:mix_block_enc}
    \begin{aligned}
        M &= \Pi_m U \Pi_B \\
          &= \sum_a \Pi_m U \ket{a_A 0_B} \bra{a_A 0_B} \\
          &= \sum_a \Pi_m \ket{\psi_a} \bra{a_A 0_B} \\
          &= \sum_a \sqrt{p_{am}} \ket{\psi_{am}} \bra{a_A 0_B},
    \end{aligned}
\end{equation}
where we defined $\ket{\psi_a} = U \ket{a_A 0_B}$, $\sqrt{p_{am}} \ket{\psi_{am}} = \Pi_m \ket{\psi_a}$, and $p_{am} = \braket{\psi_a | \Pi_m | \psi_a}$. If any $p_{am} = 0$, we arbitrarily define $\ket{\psi_{am}}$ to be orthogonal to the rest; this has no physical consequence.

Unlike the pure-state case, $M$ now has multiple nonzero singular values, and Eq.~\eqref{eq:mix_block_enc} is not generally a singular value decomposition (SVD), since the $\ket{\psi_{am}}$'s need not be orthogonal. The latter problem can be fixed by a unitary change of basis on subsystem $A$. Let $\ket{a'} = \sum_a w_{aa'} \ket{a}$ with $w$ unitary. Then
\begin{equation}
    \braket{\psi_{a'} | \Pi_m | \psi_{b'}} = \sum_{a,b} w^*_{aa'} w_{bb'} \braket{\psi_a | \Pi_m | \psi_b}.
\end{equation}
Since the matrix $p_{ab} = \braket{\psi_a | \Pi_m | \psi_b}$ is Hermitian, we can diagonalize it and thus choose a basis where the $\ket{\psi_{am}}$'s are orthogonal. Henceforth, we work under the assumption that Eq.~\eqref{eq:mix_block_enc} is an SVD. So on one hand, $p_{am}$ is the probability of getting outcome $m$ when the initial state is $\ket{\psi_a}$. On the other hand, $\sqrt{p_{am}}$'s are the singular values of $M$.

The post-measurement state, written in the basis $\{ \ket{a}_R \ket{\psi_{am}}_S \}$, is
\begin{equation}    \label{eq:mix_post_state}
    \ket{\Psi_m} = \frac{1}{\sqrt{d_R p_m}} \sum_{a=1}^{d_R} \sqrt{p_{am}} \ket{a}_R \ket{\psi_{am}}_S,
\end{equation}
where $p_m = \sum_a p_{am}/d_R$. Meanwhile, let us look at the final state when the initial state undergoes a QSVT circuit associated with $x\to f(x)$, where $f(x)$ is an odd polynomial. Note that if the initial state is $\ket{a_A 0_B}$, it is mapped to $f(\sqrt{p_{am}}) \ket{\psi_{am}}$. Thus, the QSVT output state is
\begin{equation}    \label{eq:mix_qsvt_state}
    \ket{\Psi_{\mathrm{QSVT}}} = \frac{1}{\sqrt{d_R p_{\mathrm{QSVT}}}} \sum_{a=1}^{d_R} f(\sqrt{p_{am}}) \ket{a}_R \ket{\psi_{am}}_S,
\end{equation}
with normalization factor
\begin{equation}    \label{eq:pQSVT}
    p_{\mathrm{QSVT}} = \frac{1}{d_R} \sum_{a=1}^{d_R} f(\sqrt{p_{am}})^2.
\end{equation}
Notably, the QSVT circuit only outputs the desired transformation when the final state lies in the subspace defined by $\Pi_m$. So unlike pure state amplitude amplification which succeeds almost for sure, for a general QSVT the flag measurement is necessary (see Fig.~\ref{fig:mixed}(f)). The quantity $p_{\mathrm{QSVT}}$ represents the post-selection success probability.

Comparing Eqs.~\eqref{eq:mix_post_state} and \eqref{eq:mix_qsvt_state}, we see that QSVT simulates the measurement effect when $f(x) = x/\sqrt{p^*}$, with $p^*$ a constant. We can use a small $p^*$ to amplify the post-selection probability---this is known as linear amplitude amplification (LAA) \cite{PRXQuantum.2.040203}, or uniform singular value amplification \cite{GilyenQSVT}. To simulate $\ket{\Psi_m}$ faithfully, $p^*$ must be at least $p_{\max} = \max_a p_{am}$. However, larger $p^*$ leads to smaller success probability, so it is also acceptable to use $p^*<p_{\max}$, thereby a bit of state fidelity is sacrificed in exchange for higher $p_{\mathrm{QSVT}}$. The circuit diagram for LAA by QSVT is shown in Fig.~\ref{fig:mixed}(f). 

What is the query complexity of LAA by QSVT? Based on an efficient polynomial approximating the linear transformation given in Refs.~\cite{low2017hamiltoniansimulationuniformspectral,GilyenQSVT}, we deduce that LAA can be realized with query complexity $O(\frac{1}{\sqrt{p^*}} \log\frac{1}{\sqrt{p^*} \delta})$, where $\delta$ is the allowed \emph{multiplicative} error in realizing $x/\sqrt{p^*}$. The reason why we need to bound the multiplicative error but not just the additive error is that only the former can lead to a state-independent lower-bound on the fidelity (Theorem \ref{thm:infidelity_mix} in Appendix \ref{app:polyerror}). Therefore, the LAA algorithm also has a square root advantage over naive post-selection.

The rest of this section analyzes the performance of the LAA algorithm. In addition to aiming for high fidelity between $\ket{\Psi_{\mathrm{QSVT}}}$ and $\ket{\Psi_m}$, it is also desirable to achieve a high success probability for the flag measurement. The fidelity of the state produced by the QSVT algorithm is given by
\begin{equation}    \label{eq:FQSVT_mix}
    F_{\mathrm{QSVT}} 
    = \left| \braket{\Psi_{\mathrm{QSVT}} | \Psi_m} \right|^2
    = \frac{ \left( \sum_a \sqrt{p_{am}} f(\sqrt{p_{am}}) \right)^2 }{ d_R^2 p_m p_{\mathrm{QSVT}} },
\end{equation}
where the success probability $p_{\mathrm{QSVT}}$ is defined in Eq.~\eqref{eq:pQSVT}.

Alternatively, one may consider the overall fidelity by incorporating the success probability. That is, instead of post-selecting on $\Pi_m$ after the $U_{\vec{\phi}}$ circuit, we treat it as a quantum channel that outputs $\ket{\Psi_{\mathrm{QSVT}}}$ with probability $p_{\mathrm{QSVT}}$ and some other state with probability $1 - p_{\mathrm{QSVT}}$. Although the latter may still have a positive contribution to the fidelity, we ignore this contribution to obtain a lower bound. Accordingly, we define
\begin{equation}    \label{eq:Foverall_mix}
    F_{\mathrm{overall}} = p_{\mathrm{QSVT}} F_{\mathrm{QSVT}}.
\end{equation}

If we want an algorithm with unit overall fidelity, i.e., $p_{\mathrm{QSVT}} = F_{\mathrm{QSVT}} = 1$, a necessary condition is that $p_{am}$ be independent of $a$. We refer to such a measurement operator $\Pi_m$ as an \emph{uninformative measurement}, since it prevents the measurer from gaining any information about the purifying system $R$. This condition has rich physical implications and can be understood from multiple perspectives.

First, the uninformative condition coincides with the decoupling condition in quantum coding theory, which will be discussed in detail in the context of the decoding task in Sec.~\ref{sec:teleportation}. Another perspective is obtained by comparing $F_{\mathrm{overall}}$ with an information-theoretic upper bound on the fidelity achievable by any quantum channel. The following statement can be proved using Uhlmann's theorem: there exists a quantum channel $\mathcal{D}_S$ such that $F\left( \mathcal{D}_S(\Psi_{RS}), \Psi_{RS|m} \right) \geq 1 - \varepsilon$ if and only if $F(\Psi_R, \Psi_{R|m}) \geq 1 - \varepsilon$ (Theorem~\ref{thm:decoupling_sim} in Appendix~\ref{app:Uhlmann}). Using the fact that the eigenvalues of $\Psi_{R|m}$ are $p_{am}/(d_R p_m)$, we can express this upper bound, which we refer to as the Uhlmann fidelity, as
\begin{equation}    \label{eq:FUhlmann}
    \begin{aligned}
        F_{\mathrm{Uhlmann}}
        &= F\left( \frac{\id_R}{d_R}, \Psi_{R|m} \right) \\
        &= \frac{1}{d_R} \left( \tr \sqrt{\Psi_{R|m}} \right)^2 \\
        &= \frac{1}{d_R^2 p_m} \left( \sum_a \sqrt{p_{am}} \right)^2.
    \end{aligned}
\end{equation}

From this, we conclude that a deterministic algorithm exists to simulate the post-selected state if and only if $p_{am}$ is constant for all $a$, at fixed $m$. LAA with $p^* = p_m$ is an example of such an algorithm. Moreover, when all $p_{am} = p_m$, the value of $f(x)$ at other points becomes irrelevant. In this case, using the FPAA polynomial is preferable, as it requires less knowledge about the value of $p_m$.

\begin{figure}
    \centering
    \includegraphics[width=\linewidth]{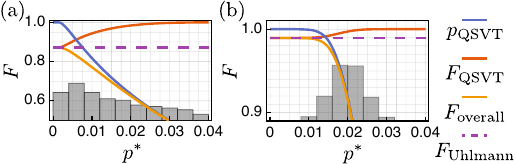}
    \caption{Fidelity and success probability of the LAA-by-QSVT algorithm in simulating mixed state post-selection. The histograms show the frequency distribution of $p_{am}$ values. (a) A 14-qubit state is initialized with seven maximally mixed qubits, followed by a Haar-random unitary. Eight qubits are then measured in the computational basis to obtain $p_{am}$ values, which are known to follow the Jacobi distribution. (b) A hypothetical distribution of 2048 independent and identically distributed (i.i.d.) $p_{am}$ values sampled from a normal distribution.}
    \label{fig:mix_fidelity}
\end{figure}

To further support the fidelity analysis above, we ``benchmark'' the LAA algorithm through numerical simulations. In computing the fidelity and success probability shown in Fig.~\ref{fig:mix_fidelity}, we assume the singular value transformation $f(x)$ takes the following ideal form:
\begin{equation}    \label{eq:LAAfunc}
    f(x) = 
    \begin{cases}
        -1 & \text{for } -1 \leq x < -\sqrt{p^*}, \\
        x/\sqrt{p^*} & \text{for } -\sqrt{p^*} \leq x < \sqrt{p^*}, \\
        1 & \text{for } \sqrt{p^*} \leq x \leq 1.
    \end{cases}
\end{equation}
Since reducing $\delta$ is much cheaper than reducing $p^*$, we take the limit $\delta \to 0$ and work with this simplified function, leaving $p^*$ as the sole tunable parameter. We then examine how the various fidelity measures discussed earlier depend on $p^*$.

From Fig.~\ref{fig:mix_fidelity}, we make the following observations.

First and foremost, when $p^* \geq p_{\max}$, the fidelity of the output state $F_{\mathrm{QSVT}}$ reaches 1, as expected. Interestingly, even when $p^*$ is set such that only a small fraction of $p_{am}$ values exceed it, the fidelity can still remain close to 1.

Second, the algorithm generally performs better for less informative measurements, as illustrated in Fig.~\ref{fig:mix_fidelity}(b), compared to more informative ones in Fig.~\ref{fig:mix_fidelity}(a). In the former case, the $p_{am}$ values provide little information about which $\ket{a}_R$ is present in the initial state, which aligns with the notion of an uninformative measurement. This behavior is a direct manifestation of Theorem~\ref{thm:decoupling_sim} in Appendix~\ref{app:Uhlmann}. Although the achievable fidelity and success probability are lower in the more informative case of Fig.~\ref{fig:mix_fidelity}(a), the success probability $p_{\mathrm{QSVT}}$ can still remain $O(1)$ while $F_{\mathrm{QSVT}}$ stays close to 1 (e.g., by choosing $p^* \in [0.02, 0.04]$). In general, the LAA algorithm is practically useful only if $p_{\mathrm{QSVT}} = O(1)$. Whether this holds depends on the distribution of $p_{am}$, and must be assessed on a case-by-case basis. A general lower bound on $p_{\mathrm{QSVT}}$ is provided in Theorem~\ref{thm:tailbound_mix} in Appendix~\ref{app:pQSVTbound}. Roughly speaking, it supports the claim that $p_{\mathrm{QSVT}}$ is non-vanishing as long as no individual $p_{am}$ is significantly larger than the average $p_m$.

Lastly, although the primary goal is to keep $F_{\mathrm{QSVT}}$ close to 1, we may also ask: what is the optimal function $f(x)$ that maximizes the overall fidelity $F_{\mathrm{overall}}$? As shown in both Fig.~\ref{fig:mix_fidelity}(a) and (b), the overall fidelity is maximized as $p^* \to 0$, corresponding to $f(x) = \operatorname{sgn}(x)$---that is, the FPAA algorithm with a sufficiently small threshold. Actually, this can be proven analytically: from Eq.~\eqref{eq:Foverall_mix}, it is clear that $F_{\mathrm{overall}}$ reaches $F_{\mathrm{Uhlmann}}$ if and only if $|f(\sqrt{p_{am}})| = 1$ for all $p_{am}$.

\section{Measurement-induced teleportation}
\label{sec:teleportation}

\begin{figure}
    \centering
    \includegraphics[width=\linewidth]{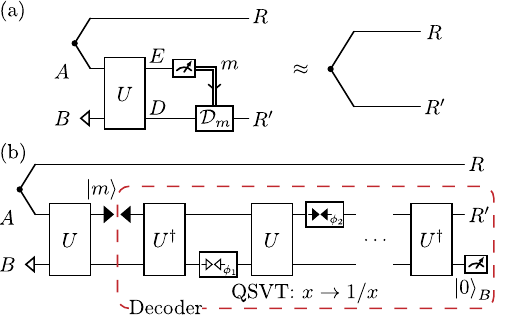}
    \caption{(a) Setup for measurement-induced teleportation and decoding. Here, $\mathcal{D}_m$ is a quantum channel that maps subsystem $D$ to $R'$, where $R'$ is isomorphic to $R$. (c) The pseudoinverse decoder, where the internal structure of the $\Pi_{\phi}$ gates is omitted for clarity.}
    \label{fig:tel}
\end{figure}

Closely related to the mixed-state post-selection setup discussed above, we study another interesting application of QSVT: decoding protocols for quantum teleportation. 

Our setup is shown in Fig.~\ref{fig:tel}. Suppose Alice encodes $d_R$-dimensional quantum information into a larger Hilbert space and retains a record of this information in a reference Hilbert space $R$. The pure state on $RS$ is
\begin{equation}
    \ket{\Psi} = \sum_{a=1}^{d_R} \ket{a}_R \otimes \ket{\psi_a}_R    ,
\end{equation}
where $\ket{\psi_a} = U\ket{a_A 0_B}$. The system $S$ is then measured by a projective measurement $\Pi_m$, and Bob must decode Alice's encoded information from the post-measurement state and the classical outcome $m$. 

In the spirit of block encoding, the encoding defines a linear transformation (up to a normalization factor due to state collapse) from the encoded state $\ket{a}$ to $M\ket{a} = \sqrt{p_{am}} \ket{\psi_{am}}$, where the matrix $M$ is the same as that used in Sec.~\ref{sec:mixed}; see Fig.~\ref{fig:mixed}(e). A mathematically straightforward decoder is to apply the Moore--Penrose pseudoinverse of $M$. Suppose for now that $M$ is injective. The pseudoinverse can then be defined as
\begin{equation}
    M^+ = \sum_{a=1}^{d_R} \frac{1}{\sqrt{p_{am}}} \ket{a} \bra{\psi_{am}}   .
\end{equation}
Ideally, using $M^+$ as a decoder leads to the overall state transformation $M^+ M = \sum_a \ket{a}\bra{a} = \id_R$, thus achieving perfect decoding. Taking practical considerations into account, we only require that the singular values $x$ greater than a small threshold (denoted as $\sqrt{p^*}$) be mapped to $1/x$. Therefore, if $M$ is injective, setting $p^* \leq p_{\min} = \min_{a} p_{am}$ yields a perfect decoder. To maximize the algorithm's success probability and minimize complexity, the ideal choice for $p^*$ is simply $p_{\min}$. However, if some $p_{am}$ values are zero (or extremely small), then inversion becomes impossible (or computationally hard) for the corresponding dimensions. For instance, if one $p_{am}$ is extremely small, then in encoding $\ket{\varphi}_A$ into $M\ket{\varphi}_{AB}$, the corresponding amplitude of $\ket{\psi_{am}}$ is nearly zero, making it difficult to accurately recover the amplitude of $\ket{a}$ in the original state. 

How can this pseudoinverse be implemented in quantum circuits? Once again, the answer is QSVT: we view $M^\dagger = \Pi_B U^\dagger \Pi_m$ as the encoded matrix, and then design a QSVT algorithm to realize the functional transformation $x \to \sqrt{p^*}/x$ (with the numerator imposed by the constraint $|f(x)| \leq 1$). In practice, this function is replaced by its polynomial approximation. Regarding query complexity, it is known that there exists a polynomial approximation of the inverse function with additive error $\delta$, and the required degree is $O(\frac{1}{\sqrt{p^*}} \log \frac{1}{\delta})$ \cite{GilyenQSVT}. Crucially, however, we must upper-bound the \emph{multiplicative} error to ensure that the decoding fidelity $F_{\mathrm{QSVT}}$ in Eq.~\eqref{eq:FQSVT_tel} remains close to $1$ (Theorem~\ref{thm:infidelity_tel} in Appendix~\ref{app:polyerror}). Altogether, this implies that there exists a polynomial of degree $O\left( \frac{1}{\sqrt{p^*}} \log\left( \sqrt{p_{\max}/p^*} / \delta \right) \right)$ that agrees with $\sqrt{p^*}/x$ on the interval $[\sqrt{p^*}, \sqrt{p_{\max}}]$ and has a multiplicative error less than $\delta$.

We highlight the advantages of our pseudoinverse decoder. The most striking feature is that the decoding fidelity can approach $1$ even when the decoupling condition is violated, i.e., when decoding via quantum channels is information-theoretically impossible. The key insight is that we sacrifice determinism in exchange for high decoding fidelity. Another appealing feature is that the decoder requires very few ancilla qubits. In fact, it only requires the ancillae needed for the block encoding and one additional qubit for the $\Pi_{\phi}$ gates.

To further justify these advantages, we compare our decoder with other teleportation decoders. Ref.~\cite{utsumi2024explicit} proposed explicit decoders for decoding after a noise channel. Their general strategy is to first construct a Yoshida--Kitaev decoder \cite{yoshida2017efficient}, which succeeds upon post-selection onto a Bell state, and then replace the post-selection with FPAA using QSVT. In Appendix~\ref{app:decoder}, we adapt the decoders from Ref.~\cite{utsumi2024explicit} to our teleportation decoding setting. These FPAA decoders and our pseudoinverse decoder each have their own advantages. Specifically, the FPAA decoder is a deterministic decoder and achieves the optimal overall fidelity, as can be seen by explicitly expanding Eq.~\eqref{eq:Foverall_tel}. However, when the decoupling condition is not satisfied, the resulting state does not faithfully represent the initially encoded state. In contrast, although our pseudoinverse decoder succeeds with lower probability, it is designed to ensure that the decoding fidelity remains close to $1$ even in scenarios where the decoupling condition is violated.

Now, we analyze the performance of the pseudoinverse decoder in terms of fidelity and success probability. The state on the full system ($RS = RAB = RED$) before the decoder is
\begin{equation}
    \ket{\Psi_m} = \frac{1}{\sqrt{d_R p_m}} \sum_{a=1}^{d_R} \sqrt{p_{am}} \ket{a}_R \ket{\psi_{am}}_S .
\end{equation}
The QSVT circuit (see the dashed box in Fig.~\ref{fig:tel}(c)) implements the matrix $\sum_{a} f(\sqrt{p_{am}}) \ket{a}_{R'} \bra{\psi_{am}}_{S}$ followed by state collapse. Hence, the state after decoding becomes
\begin{equation}
    \begin{aligned}
        &\phantom{=} \ket{\Psi_{\mathrm{QSVT}}}_{RR'}   \\
        &= \frac{1}{\sqrt{p_{\mathrm{QSVT}}}} \sum_{a} f(\sqrt{p_{am}}) \ket{a}_{R'} \bra{\psi_{am}}_{S} \ket{\Psi_m}   \\
        &= \frac{1}{\sqrt{d_R p_m p_{\mathrm{QSVT}}}} \sum_{a=1}^{d_R} \sqrt{p_{am}} f(\sqrt{p_{am}}) \ket{a}_R \ket{a}_{R'}    ,
    \end{aligned}
\end{equation}
where the algorithm success probability appears as the normalization factor:
\begin{equation}    \label{eq:pQSVT_tel}
    p_{\mathrm{QSVT}} = \frac{1}{d_R p_m} \sum_{a=1}^{d_R} p_{am} f(\sqrt{p_{am}})^2    .
\end{equation}
The decoding fidelity is defined as the fidelity between $\ket{\Psi_{\mathrm{QSVT}}}_{RR'}$ and $\EPR_{RR'}$, and is given by
\begin{equation}    \label{eq:FQSVT_tel}
    F_{\mathrm{QSVT}}
    = \left| \braket{ \mathrm{EPR} | \Psi_{\mathrm{QSVT}}}_{RR'} \right|^2
    = \frac{ \left( \sum_a \sqrt{p_{am}} f(\sqrt{p_{am}}) \right)^2 }{ d_R^2 p_m p_{\mathrm{QSVT}} }    .
\end{equation}
We remark that this equation appears identical to Eq.~\eqref{eq:FQSVT_mix} in Sec.~\ref{sec:mixed}, but the expression for $p_{\mathrm{QSVT}}$ differs.

In addition, we may also consider the fidelity of our decoder as a channel. Similar to the discussion in Sec.~\ref{sec:mixed}, the fidelity of the corresponding channel decoder is lower-bounded by
\begin{equation}    \label{eq:Foverall_tel}
    F_{\mathrm{overall}} = p_{\mathrm{QSVT}} F_{\mathrm{QSVT}}   .
\end{equation}
On the other hand, from an information-theoretic perspective, the fidelity of all possible channel decoders is bounded by $F(\id_R/d_R , \Psi_{R|m})$ (see Theorem~\ref{thm:decoupling} in Appendix~\ref{app:Uhlmann} for the precise statement). This bound, denoted as $F_{\mathrm{Uhlmann}}$, has the same expression as in Eq.~\eqref{eq:FUhlmann}. Especially, the Uhlmann bound implies that the decoding fidelity can be $1$ if and only if all $p_{am}$'s are constant for fixed $m$---that is, the condition for uninformative measurement. 

\begin{figure}
    \centering
    \includegraphics[width=\linewidth]{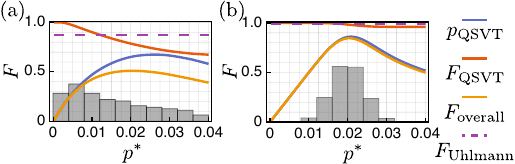}
    \caption{Decoding fidelity and success probability of the pseudoinverse decoder. The distribution of probabilities $p_{am}$ is generated in the same way as in Fig.~\ref{fig:mix_fidelity}.}
    \label{fig:decoding_fidelity}
\end{figure}

We numerically calculate $p_{\mathrm{QSVT}}$, $F_{\mathrm{QSVT}}$, $F_{\mathrm{overall}}$, and $F_{\mathrm{Uhlmann}}$ for the same representative examples as those studied in Fig.~\ref{fig:mix_fidelity}. The results for the pseudoinverse decoder are shown in Fig.~\ref{fig:decoding_fidelity}. For simplicity, in generating the data in Fig.~\ref{fig:decoding_fidelity}, we assume that the polynomial approximation error $\delta$ is suppressed to zero, and we adopt the following ideal form of $f(x)$:
\begin{equation}    \label{eq:inv_func}
    f(x) = 
    \begin{cases}
        \sqrt{p^*}/x & \text{if } \sqrt{p^*} < |x| \leq 1, \\
        x/\sqrt{p^*} & \text{if } |x| \leq \sqrt{p^*}.
    \end{cases}
\end{equation}
This function now has only one tunable parameter, $p^*$. We plot the fidelity and the success probability of the QSVT decoder as functions of $p^*$. The plots confirm that when $p^*$ is smaller than the majority of $p_{am}$ values, the decoding fidelity becomes sufficiently close to $1$. However, as $p^*$ decreases, the success probability drops linearly, indicating the need to balance these two competing factors. The optimal choice of $p^*$ depends on the acceptable level of infidelity and on the distribution of $p_{am}$, which generally requires a case-by-case analysis. Nevertheless, we provide a lower bound for $p_{\mathrm{QSVT}}$ in terms of the $p_{am}$ distribution in Theorem~\ref{thm:tailbound_tel} in Appendix~\ref{app:pQSVTbound}. The comparison between Figs.~\ref{fig:decoding_fidelity}(a) and (b) shows that our algorithm performs better when the distribution of $p_{am}$ is less informative about the initial state $\ket{a}_R$. Indeed, Fig.~\ref{fig:decoding_fidelity}(b) achieves close-to-unity $F_{\mathrm{QSVT}}$ with a significantly higher success probability $p_{\mathrm{QSVT}}$.

Finally, we point out that our pseudoinverse decoder can be applied to the decoding task of approximate quantum error correction codes as well. Specifically, we interpret the linear mapping $M$ from $A$ to $D$ as Kraus operator in the entire encoding-plus-noise channel, and a subsequent syndrome measurement $\Pi_m$ indicates which error (i.e., Kraus operator) has happened. Thus, the task of decoding $A$ from $D$ appears identical to the above-mentioned teleportation decoding task. We formulate the condition for approximate quantum codes as an approximately satisfied decoupling condition:
\begin{equation}
    \frac{1}{2}\left\| \Psi_{R|m} - \frac{\id_R}{d_R} \right\|_1 \leq \varepsilon'. 
\end{equation}
In terms of the probabilities $p_{am}$, this reads
\begin{equation}
    \frac{1}{2d_R} \sum_{a=1}^{d_R} \left| \frac{p_{am}}{p_m} - 1 \right| \leq \varepsilon'  .
\end{equation}
This upper-bounds the deviation of all $p_{am}$'s from their average. From this, we can prove that the success probability $p_{\mathrm{QSVT}} = O(1)$ as long as the allowed infidelity in $F_{\mathrm{QSVT}}$ is larger than $\epsilon'$ (Corollary~\ref{cor:AQEC_bound_pQSVT} in Appendix \ref{app:pQSVTbound}). Therefore, the pseudoinverse decoder is guaranteed to be useful for approximate quantum codes.

\section{Tradeoff between classical and quantum complexity}
\label{sec:tradeoff}

In this section, we propose a general tradeoff relation between classical and quantum complexity arising from post-selection.

Let us take a step back to Sec.~\ref{sec:previous} and consider quantum--classical correlation approaches. These approaches eliminate the sample-complexity overhead, at the expense of having to classically simulate each quantum trajectory that appears in the experiment. In contrast, this work introduces amplitude amplification to remove the sample-complexity overhead without relying on classical simulation. The price we pay is that the state preparation circuit becomes $O(1/\sqrt{p_m})$ deeper. This leads us to the observation that, among sample complexity, extent of classical simulability, and quantum query complexity, at least one must be large. 

To formulate a fair comparison between classical and quantum approaches, one may consider the task of estimating the ensemble average of some non-linear observable under fixed error tolerance and a fixed number of samples. In this setting, it is possible to define a ``combined complexity'' that includes a term reflecting the extent to which a classical algorithm can simulate the quantum dynamics, and a term representing quantum query complexity. The tradeoff relation would then be manifested as a lower bound on this combined complexity.

A similar tradeoff is also observed in decoding problems. Let us consider decoders for measurement-induced teleportation. On one hand, we proposed QSVT-based decoders with complexity $\tilde{O}(1/\sqrt{p^*})$. On the other hand, if the encoding circuit is Clifford, a simple correction operation dependent on the measurement outcome suffices to recover the injected state—this is a standard technique in measurement-based quantum computation \cite{Briegel2009,Wei_2021}. Although sample complexity does not explicitly enter here, we still observe a tradeoff between classical simulability and quantum query complexity.

While not the main focus of this paper, such tradeoff also arise in the decoding problem of the Hayden--Preskill protocol \cite{HaydenPreskill2007}. Specifically, the Yoshida--Kitaev decoder \cite{yoshida2017efficient} is effective when the black hole dynamics is sufficiently scrambling, so that deterministic decoding becomes possible via the random decoupling mechanism. This construction is Grover-based, thus requiring many simulations of the black hole evolution. But later, Ref.~\cite{yoshida2022recovery} considered the Hayden--Preskill protocol with a Clifford black hole, where decoding reduces to an error correction operation conditioned on syndrome measurements. Thus, the Hayden--Preskill protocol also exemplifies a potential quantum--classical complexity tradeoff.

Summarizing the above examples, they all point to the principle that the less classical understanding we have, the more quantum operations we must perform; and vice versa. This principle is more broadly applicable than just measurement- and post-selection--related problems; it pertains to a wide class of state transformation tasks, resembling the framework of the Uhlmann transformation problem \cite{bostanci2023unitary}, although it does not contain classical simulations. It is tempting to formalize this tradeoff in terms of a quantum--classical combined complexity, which could be lower-bounded for general state transformation tasks.

\section{Summary and Outlook}
\label{sec:summary}

Overall, we have developed a quantum algorithmic framework that relieves the post-selection problem in measuring ensemble averages of non-linear observables. For pure initial states, our approach demonstrates that probabilistic state preparation can be deterministically simulated using QSVT-based amplitude amplification, thereby yielding post-selected pure states with an optimal quadratic speedup. By leveraging the FPAA algorithm, our method eliminates the need for exponential sampling when estimating expectation values of non-linear observables. Our method can be useful in many mainstream topics involving projected ensembles, including the estimation of entropy-like quantities as proxies for MIPT, and the testing of deep thermalization \cite{PhysRevLett.128.060601,Ippoliti2022solvablemodelofdeep}.

Our method also extends to mixed-state post-selection. In this case, instead of fixed-point amplitude amplification, linear amplitude amplification via QSVT must be used. Due to the lack of access to a purification of the mixed state, our LAA algorithm carries the caveat of being unable to deterministically simulate the effect of post-selection. Nevertheless, LAA can significantly boost the success probability of state preparation.

We further introduced a QSVT-based decoder for measurement-induced quantum teleportation. Mathematically, this decoder corresponds to applying the pseudoinverse of the encoding matrix, and is therefore referred to as the pseudoinverse decoder. It achieves near-perfect recovery of encoded information, even in scenarios where conventional quantum channel decoders fail due to the violated decoupling condition. The success probability of our pseudoinverse decoder is model-dependent and generally favors measurements that are less informative about the initial encoded state. Moreover, our pseudoinverse decoder can be directly applied to approximate quantum error correction code, in which case the deviation from the decoupling condition is small and thus the decoder succeeds with finite probability. 


Our work sheds light on various topics in quantum complexity theory and quantum gravity. In the realm of quantum complexity theory, our quantum algorithmic solutions to post-selection problems suggest a tradeoff between classical simulability and quantum query complexity. Regarding quantum gravity, several models of the universe or black holes involve post-selection---for instance, the black hole final state model \cite{Horowitz_2004}, the ``python's lunch'' geometry \cite{brown2019pythonslunch}, and quantum theories of de Sitter space. Our QSVT-based algorithms may offer a new perspective for analyzing the capability and complexity of measurements in such space-times with final state(s).

Finally, we list several future directions worth exploring.

\begin{enumerate}
    \item One of the most intriguing open questions is to rigorously formulate or provide further evidence for the tradeoff between quantum and classical complexity, as conjectured in Sec.~\ref{sec:tradeoff}. Regarding practical applications, since achieving Grover's speedup is challenging on near-term noisy quantum devices \cite{Preskill2018qc,Chen2023}, a promising direction is to develop hybrid quantum--classical strategies to address the post-selection problem. Specifically, one could consider quantum dynamics where partial information is classically simulable, while Grover-like quantum iterations are still necessary to fully eliminate the post-selection sampling overhead. Such models could drastically reduce the required circuit depth, enabling implementation on near-term devices.
    \item Another direction is to investigate the performance of our mixed-state simulation and teleportation decoding algorithms under typical random dynamics. As the circuit depth of a random quantum circuit increases, the probability distribution is known to transition from concentration to anti-concentration \cite{Hangleiter2018anticoncentration,PRXQuantum.3.010333}. Our results (especially Theorem~\ref{thm:tailbound_tel}) on the condition for $p_{\mathrm{QSVT}}=O(1)$ is analogous to some definitions of anti-concentration, but it would be interesting to quantitatively study how our algorithms behave across this crossover.
    \item Our quantum algorithms for mixed-state post-selection and teleportation decoding do not necessarily saturate the Uhlmann bound. It is an appealing to design new algorithms that achieve $F \geq 1-\varepsilon$ with $p \approx F_{\mathrm{Uhlmann}}$. Additionally, the quantum complexity-theoretic perspective of the decodable channel problem \cite{bostanci2023unitary} may shed light on the fundamental limits of computationally tractable decoding protocols.
    \item With respect to the pseudoinverse decoder, many quantum algorithms for solving linear systems beyond QSVT could, in principle, be adapted into similar decoders (see, e.g., \cite{HHL2009,costa2021optimalscalingquantumlinear,low2024quantumlinear}). However, such algorithms must operate solely with samples of the initial state, without access to the state preparation protocol. Exploring whether improved decoding protocols could be derived from these algorithms remains an open question.
\end{enumerate}

\begin{acknowledgments}
I thank Ehud Altman, Yulong Dong, Samuel J. Garratt, Matteo Ippoliti, Vedika Khemani, Yaodong Li, Yuan Su, Jinzhao Wang and Yi-Zhuang You for discussion. I am particularly grateful to Sarang Gopalakrishnan, Xiao-Liang Qi and Haifeng Tang for discussion and collaboration on related projects. 
This work was conceived during my visit at Stanford Institute for Theoretical Physics (SITP). I am grateful for the hospitality of SITP. 
\end{acknowledgments}

\appendix

\section{Quantum circuit notations}
\label{sec:notation}

For intuitive presentation of certain quantum circuit equations, we use circuit diagrams in place of symbolic expressions when appropriate. In this section, we fix notations and introduce relevant definitions. 

By convention, time flows either from bottom to top or from left to right in our circuit diagrams. Lines represent the transmission of quantum states, while double lines denote classical signals. Square boxes indicate unitary operations. Pure states are represented by triangles; in particular, a white triangle denotes the computational basis state $\ket{0}$, which is assumed to be easily preparable:
\begin{equation}    \label{eq:computational_basis}
    \ket{0}_A = 
    \vcenter{\hbox{\includegraphics{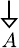}}}  .
\end{equation}

Uppercase Latin letters denote subsystems of a larger system (e.g., $A$), and the same letter is used for the corresponding Hilbert space when the meaning is clear from context. The Hilbert space dimension of $A$ is denoted by $d_A$. Expressions such as $AB$, $ABC$, etc., represent the union of several subsystems, whose Hilbert space is the tensor product of the constituent subspaces. We use $\id$ to denote the identity matrix, or $\id_A$ and $\id_{d_A}$ to clarify its underlying Hilbert space or dimension when necessary. 

The following notation is used for the Einstein--Podolsky--Rosen (EPR) state, also known as the maximally entangled state between isometric subsystems $A$ and $A'$:
\begin{equation}    \label{eq:EPR}
    \EPR_{AA'}
    = \vcenter{\hbox{\includegraphics{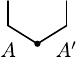}}}
    = \frac{1}{\sqrt{d_A}} \sum_{a=1}^{d_A} \ket{a}_A \ket{a}_{A'}   ,
\end{equation}
where, if the lines are interpreted as tensor contractions, the black dot can be understood as a constant factor $1/\sqrt{d_A}$. 

Next, we introduce our convention for mixed states. A general quantum state is represented mathematically by a density matrix. In circuit diagrams, pure states may still be denoted using vector-like notations, such as in Eq.~\eqref{eq:computational_basis} or Eq.~\eqref{eq:EPR}, when no confusion arises. It is understood that the actual density matrix is the projector onto the corresponding pure state.

Starting from a pure state $\ket{\psi}_{AB}$, we use $\rho_A$ to denote the reduced density matrix of subsystem $A$:
\begin{equation}
    \rho_A = \tr_{B} \rho_{AB} .
\end{equation}
The partial trace operation is depicted by a slash terminating the traced line. For example,
\begin{equation}
    \tr_B \left( \ket{\psi}_{AB} \bra{\psi}_{AB} \right) = \vcenter{\hbox{\includegraphics{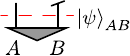}}} ,
\end{equation}
which, if explicitly drawing the two copies, corresponds to the following tensor diagram:
\begin{equation}
    \vcenter{\hbox{\includegraphics{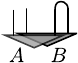}}}
\end{equation}

Regarding measurements and state overlaps, we use \includegraphics{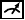} to represent a projective measurement apparatus, usually in the computational basis. When the measurement outcome is $m$, its effect on a state is to apply the corresponding projector $\Pi_m$, followed by normalization according to Born's rule.

The purity of a density matrix $\rho$ is defined as $P(\rho) = \tr \rho^2$. The second Rényi entropy is related to the purity by $S^{(2)}(\rho) = -\log P(\rho)$. 

Finally, the fidelity between two states is defined as
\begin{equation}
    F(\rho, \sigma) = \left( \tr \sqrt{\sqrt{\rho} \sigma \sqrt{\rho}} \right)^2 .
\end{equation}
For pure states, we also use the notation $F(\ket{\psi}, \ket{\varphi})$, which should be understood as $F(\ket{\psi}\bra{\psi}, \ket{\varphi}\bra{\varphi})$. In particular, the fidelity between two pure states admits a simple expression:
\begin{equation}
    F(\ket{\psi}, \ket{\varphi}) = \left| \braket{\psi|\varphi} \right|^2 .
\end{equation}

\section{Additional information on QSVT}
\label{app:QSVT}

In this Appendix, we provide additional information on QSVT, focusing mainly on the implementation details and the approach to determine the phase sequence. 

\subsection{The $\Pi_\phi$ gate}

A crucial operation in QSVT algorithms is the $\Pi$-controlled phase gate:
\begin{equation}
    \Pi_{\phi} = e^{i\phi} \Pi + e^{-i\phi} (\id-\Pi)   .
\end{equation}
To realize this gate with an arbitrary angle, a standard method is to introduce one ancilla qubit and the $\Pi$-controlled NOT gate, which is defined as
\begin{equation}
    \mathrm{C_{\Pi}NOT} = X\otimes \Pi + \id_2 \otimes (\id-\Pi)  . 
\end{equation}
Now, observe the following circuit:
\begin{equation}    \label{eq:CNOT_phi_CNOT}
    \vcenter{\hbox{\includegraphics{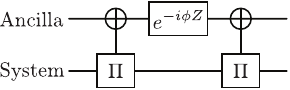}}}
\end{equation}
It equals
\begin{equation}
    \begin{aligned}
        &\phantom{=} \mathrm{C_\Pi NOT} \left( e^{-i\phi Z} \otimes \id \right) \mathrm{C_\Pi NOT} \\
        &= e^{i\phi Z} \otimes \Pi + e^{-i\phi Z} \otimes (\id-\Pi)   \\
        &= \ket{0}\bra{0} \otimes \Pi_{\phi} + \ket{1}\bra{1} \otimes \Pi_{-\phi}   .
    \end{aligned}
\end{equation}
Thus, by initializing the ancilla qubit in the $\ket{0}$ state, Eq.~\eqref{eq:CNOT_phi_CNOT} leaves the ancilla unchanged, and the effect on the system is exactly $\Pi_{\phi}$.

\subsection{Real and complex polynomials}

For a complex odd polynomial satisfying more restrictive conditions \cite{GilyenQSVT} than those preceding Eq.~\eqref{eq:QSVT_odd}, we have
\begin{equation}
    P_d(M) = \tilde{\Pi} U_{\vec{\phi}} \Pi,
\end{equation} 
where
\begin{equation}
    U_{\vec{\phi}} = \tilde{\Pi}_{\phi_1} U 
    \prod_{k=1}^{(d-1)/2} \left( \Pi_{\phi_{2k}} U^\dagger \tilde{\Pi}_{\phi_{2k+1}} U \right)  .
\end{equation}

If our goal is to realize a real polynomial transformation, for instance $\Re P(x)$, we can achieve this by superposing the transformations $P(x)$ and $P^*(x)$. In fact, this superposition can be implemented with almost no additional effort. Specifically, each gate $\Pi_\phi$ and $\tilde{\Pi}_\phi$ can be realized using Eq.~\eqref{eq:CNOT_phi_CNOT}. Depending on whether the ancilla qubit is initialized in the $\ket{0}$ or $\ket{1}$ state, the system evolves according to $U{\vec{\phi}}$ or $U_{-\vec{\phi}}$, respectively.

Suppose the initial state is $\ket{+}\otimes \ket{\psi}$, with $\ket{+} = (\ket{0}+\ket{1})/\sqrt{2}$. Then, after the alternating phase modulation sequence, the state becomes
\begin{equation}
\frac{1}{\sqrt{2}} \left( \ket{0}\otimes U_{\vec{\phi}}\ket{\psi} + \ket{1}\otimes U_{-\vec{\phi}}\ket{\psi} \right).
\end{equation}

Finally, post-selecting by the projector $\ket{+}\bra{+} \otimes \tilde{\Pi}$, the final state becomes
\begin{equation}
    \frac{\Re P(M) \ket{\psi}}{\| \Re P(M) \ket{\psi} \|}   ,
\end{equation}
with the success probability given by the square of the denominator.

In summary, a real polynomial transformation can be implemented by initializing an ancilla qubit in the $\ket{+}$ state, applying the phase modulation circuit, and finally post-selecting the ancilla qubit on $\ket{+}$ and the system on $\tilde{\Pi}$.





\subsection{From polynomial to phase sequence}

Given a desired functional transformation, the next step is to obtain an appropriate polynomial approximation. To realize the polynomial in QSVT, we need to map valid (as defined before Eq.~\eqref{eq:QSVT_odd}) polynomials to corresponding phase sequences.

The algorithm we use to determine phase sequences is based on Ref.~\cite{Dong_2021}. The code for this algorithm is available in the \texttt{QSPPACK} (MATLAB) and \texttt{pyqsp} (Python) repositories on GitHub.

We briefly describe the algorithm as follows. The key insight is that the phase sequence for QSVT is identical to that for quantum signal processing (QSP) in the reflection convention. For QSP, we formulate the problem of determining phase factors $\vec{\phi} = (\phi_0, \phi_1, \dots, \phi_d)$ as an optimization task aimed at minimizing the distance between $\Re U_{\vec{\phi}}$, the function realized by QSP, and the desired polynomial $P_d(x)$. Compared to other numerical algorithms that primarily rely on iterative methods, this optimization-based approach is more efficient and numerically stable, since it avoids the pitfalls associated with high-precision computations. The computational cost of the algorithm scales as $d^2$, where $d$ is the degree of the polynomial, and it can determine phase sequences for polynomials of degree up to $d=10^4$ using standard double-precision arithmetic.

\section{Complexity lower bound for general amplitude amplification}
\label{app:lowerbound}

In this Appendix, we prove that amplitude amplification from a pure initial state to the projector space $\Pi_m$ requires at least $\Omega(1/\sqrt{p_{m}})$ queries, treating the unitary operators $\exp(i \phi \Pi_m)$ as oracles (equivalent to $\Pi_{\phi/2}$ defined in Eq.~\eqref{eq:rel_rot} up to a global phase). Therefore, the best quantum algorithm achieves a quadratic speedup over classical algorithms, which require $\Omega(1/p_m)$ queries---essentially repeated sampling until the outcome $m$ is observed. This conclusion, as well as the proof, shares similarities with the optimality proof for Grover's search algorithm \cite{BBBV1997}. However, to the best of our knowledge, no previous work has explicitly demonstrated the optimality of this square-root speedup for amplitude amplification involving general multi-rank projectors.

\begin{theorem} \label{thm:AAoptimality}
Let $\ket{\psi}$ be a pure state and $\Pi_m$ a projector acting on the same Hilbert space of dimension $D$, with $\operatorname{rank}\Pi_m = d_m$. Suppose that $\braket{\psi | \Pi_m | \psi} = p_m$, where $0 < p_m < 1$, and that one can implement the unitaries $O_\phi = e^{i\phi \Pi_m}$ as oracles. If a quantum algorithm uses the oracle $k$ times to generate a state $\ket{\psi_m^{(k)}}$ satisfying $\braket{\psi_m^{(k)} | \Pi_m | \psi_m^{(k)}} \geq 1-\varepsilon$ for $\varepsilon \ll 1$, then 
\begin{equation}
    k = \Omega\Bigl(\min\{ p_m^{-1/2}, (d_m/D)^{-1/2} \}\Bigr).
\end{equation}
\end{theorem}

\begin{proof}
Without loss of generality, let the state $\ket{\psi_m^{(k)}}$ be generated from the following unitary operator:
\begin{equation}
    \ket{\psi_m^{(k)}} = U_k O_{\phi_k} \dots U_2 O_{\phi_2} U_1 O_{\phi_1} \ket{\psi^{(0)}}   .
\end{equation}
We compare the deviation of the true algorithm from a ``fake algorithm''. The fake algorithm erases all oracles in the above formula and obtains the state
\begin{equation}
    \ket{\psi^{(k)}} = U_k \dots U_2 U_1 \ket{\psi^{(0)}}   .
\end{equation}
We consider the deviation
\begin{equation}
    D_k = \Exp_{\Pi_m} \left| \psi_m^{(k)} - \psi^{(k)} \right|^2   ,
\end{equation}
where the average is over some projector ensemble to be specified below. We prove the theorem in two steps. First, to meet the fidelity requirement the value of $D_k$ must be finite. The second step is to bound the increment of $D_k$ with $k$. Regarding the ensemble average, because the theorem holds for any $\Pi_m$ satisfying $\braket{\psi | \Pi_m | \psi} = p_m$, the statements in terms of $D_k$ must hold for any choice of ensemble satisfying the same condition. 

The gist of the first step is that $\ket{\psi_m^{(k)}}$ is almost in the linear space spanned by $\Pi_m$, but $\ket{\psi^{(k)}}$ is a fixed state, so their distance is finite on average. To rigorously prove it, consider the following triangle inequality:
\begin{equation}    \label{eq:DtoEF}
\begin{aligned}
    D_k &= \Exp_{\Pi_m} \left| \psi^{(k)}_m - \psi^{(k)} \right|^2    \\
    &= \Exp_{\Pi_m} \left| \left( \psi^{(k)} - \Pi_m\psi^{(k)}_m \right) - \left( \psi^{(k)}_m - \Pi_m\psi^{(k)}_m \right) \right|^2  \\
    &\geq \left( \sqrt{F_k} - \sqrt{E_k} \right)^2  ,
\end{aligned}
\end{equation}
where
\begin{align}
    E_k &= \Exp_{\Pi_m} \left| \psi^{(k)}_m - \Pi_m\psi^{(k)}_m \right|^2  ,  \\
    F_k &= \Exp_{\Pi_m} \left| \psi^{(k)} - \Pi_m\psi^{(k)}_m \right|^2  .
\end{align}
$E_k$ is at most the allowed infidelity:
\begin{equation}
    E_k = \Exp_m \braket{ \psi^{(k)}_m | \id-\Pi_m | \psi^{(k)}_m } \leq \varepsilon    ,
\end{equation}
whereas
\begin{equation}    \label{eq:Fk_mid}
    F_k \geq \Exp_{\Pi_m} \left| (\id-\Pi_m) \psi^{(k)} \right|^2
    = 1 - \braket{\psi^{(k)} | \Exp_m \Pi_m | \psi^{(k)}}   .
\end{equation}
The above inequality holds because the right-hand side is the shortest possible distance between the off-plane vector $\ket{\psi^{(k)}}$ and any vector in the plane $\Pi_m$. 

We are left with a calculation of $\Exp_{\Pi_m} \Pi_m$. Generally, any projector satisfying $\braket{\psi | \Pi | \psi} = p_m$ reads
\begin{equation}
    \Pi_m =  
    \begin{bmatrix}
        p_m & v^\dagger \\
        v & \Pi_{\perp}
    \end{bmatrix},
\end{equation}
where the first dimension corresponds to the basis $\ket{\psi}$ by convention. This matrix is written in a block form: $v$ is $1\times (D-1)$, and $\Pi_\perp$ is $(D-1)\times (D-1)$. Also, $\tr \Pi_\perp = d_m-p_m$. To capture the properties of the average that does not depend on any specific projector, we conjugate the projector by $\mathrm{U}(1) \oplus \mathrm{U}(D-1)$ and average over the Haar measure, since this conjugation leaves the first entry invariant. The average gives
\begin{equation}
    \begin{aligned}
        \Exp_{\Pi_m} \Pi_m 
        &= \Exp_{\Pi_m} \Exp_{\substack{g_1 \in \mathrm{U(1)}\\ g_\perp \in \mathrm{U}(D-1)}}
        \begin{bmatrix}
            g_1 & 0 \\ 0 & g_\perp
        \end{bmatrix}
        \begin{bmatrix}
            p_m & v^\dagger \\
            v & \Pi_{\perp} 
        \end{bmatrix}
        \begin{bmatrix}
            g_1^\dagger & 0 \\ 0 & g_\perp^\dagger
        \end{bmatrix}   \\
        &= \Exp_{\Pi_m} \Exp_{\substack{g_1 \in \mathrm{U(1)}\\ g_\perp \in \mathrm{U}(D-1)}}
        \begin{bmatrix}
            p_m & g_1 v^\dagger g_\perp^\dagger \\
            g_\perp v g_1^\dagger & g_\perp \Pi_{\perp} g_\perp^\dagger
        \end{bmatrix}   \\
        &= \Exp_{\Pi_m}
        \begin{bmatrix}
            p_m & 0 \\
            0 & \frac{\tr\Pi_\perp}{D-1} \id_\perp
        \end{bmatrix}   \\
        &= \begin{bmatrix}
            p_m & 0 \\
            0 & \frac{d_m-p_m}{D-1} \id_\perp
        \end{bmatrix}.
    \end{aligned}
\end{equation}
Therefore, 
\begin{equation}    \label{eq:avg_proj}
    \Exp_{\Pi_m} \Pi_m = \frac{d_m-p_m}{D-1} \id + \frac{Dp_m-d_m}{D-1} \ket{\psi}\bra{\psi}   .
\end{equation}

Inserting Eq.~\eqref{eq:avg_proj} back into Eq.~\eqref{eq:Fk_mid}, we have
\begin{equation}
    F_k = 1 - \frac{d_m-p_m}{D-1} - \frac{Dp_m-d_m}{D-1} \left| \braket{\psi^{(k)} | \psi} \right|^2   .
\end{equation}
We see distinct behaviors depending on the sign of $p_m-d_m/D$: if $p_m \geq d_m/D$, $F_k$ is lower-bounded by $1-p_m$ by taking $\left| \braket{\psi^{(k)} | \psi} \right|^2 = 1$; if $p_m < d_m/D$, $F_k$ is lower-bounded by $1 - \frac{d_m-p_m}{D-1}$ by taking $\left| \braket{\psi^{(k)} | \psi} \right|^2 = 0$. We will discuss their implications on the actual algorithm at the end, but in both cases $F_k$ is finite. Therefore, $D_k$ must be finite.

In the second step, we show $\sqrt{D_{k+1}}-\sqrt{D_k} \leq 2\sqrt{\max\{ p_m, d_m/D \}}$, such that by induction $D_k \leq 4 k^2 \max\{p_m, d_m/D\}$. We start from expanding $D_{k+1}$:
\begin{equation}
    \begin{aligned}
        D_{k+1} &= \Exp_{\Pi_m} \left| \psi^{(k+1)}_m - \psi^{(k+1)} \right|^2    \\
        &= \Exp_{\Pi_m} \left| \psi^{(k)}_m - O_{-\phi_{k+1}} \psi^{(k)} \right|^2  .
    \end{aligned}
\end{equation}
Using the triangle inequality, we find
\begin{equation}    \label{eq:Dk+1}
    \begin{aligned}
        \sqrt{D_{k+1}} - \sqrt{D_k} 
        &\leq \sqrt{ \Exp_{\Pi_m} \left| (O_{-\phi_{k+1}} - \id) \psi^{(k)} \right|^2 }   \\
        &= \sqrt{ \Exp_{\Pi_m} \left| (e^{-i\phi_{k+1}}-1) \Pi_m \psi^{(k)} \right|^2 } \\
        &\leq 2\sqrt{ \Exp_{\Pi_m} \braket{\psi^{(k)} | \Pi_m | \psi^{(k)}} }   ,
    \end{aligned}
\end{equation}
where in the last line the inequality is because $\left| e^{-i\phi_{k+1}}-1 \right| \leq 2$. Now, the average over projectors appears again. We plug in Eq.~\eqref{eq:avg_proj} to get
\begin{equation}
    \sqrt{D_{k+1}} - \sqrt{D_k} \leq 2\sqrt{ \frac{d_m-p_m}{D-1} + \frac{Dp_m-d_m}{D-1} |\braket{\psi^{(k)} | \psi }|^2 }   .
\end{equation}
There are again two cases as in the discussion for $F_k$. If $p_m\geq d_m/D$, 
\begin{equation}
    \sqrt{D_{k+1}} - \sqrt{D_k} \leq 2\sqrt{ p_m }   ;
\end{equation}
if $p_m<d_m/D$,
\begin{equation}
    \sqrt{D_{k+1}} - \sqrt{D_k} \leq 2\sqrt{ \frac{d_m-p_m}{D-1} }   .
\end{equation}
For the latter case, assuming $D\gg 1$ to omit $1$ and $p_m$ on the right-hand side, we approximately have
\begin{equation}
    \sqrt{D_{k+1}} - \sqrt{D_k} \leq 2\sqrt{ d_m/D }   .
\end{equation}
This proves the second step. Combining both steps yields the lower bound stated in the theorem.
\end{proof}

From Eq.~\eqref{eq:Dk+1}, we observe that the maximal increment of $\sqrt{D_k}$ occurs for $\phi_k = \pi$, reducing precisely to Grover's algorithm. Other phase sequences slow the convergence by a constant factor but can provide beneficial properties, such as fixed-pointness.

Finally, we discuss the implications of the sign of $p_m-d_m/D$. It is not hard to notice that the case $p_m>d_m/D$ is similar to the situation of Grover's algorithm, where one can take advantage of the large overlap between $\ket{\psi}$ and $\Pi_m$ to quickly approach the target space by using $\exp(i \pi \Pi_m)$ and $\exp(i \pi \ket{\psi}\bra{\psi})$. While the bound $\Omega((d_m/D)^{-1/2})$ in the case $p_m<d_m/D$ is seemingly smaller than the Grover's bound, it can be understood as follows. One can replace the state $\ket{\psi}$ by a Haar random state $\ket{\psi'}$, which satisfies $ \braket{\psi' | \Pi_m | \psi'} \geq d_m/D$ with high probability. Then we do Grover's algorithm between $\ket{\psi'}$ and $\Pi_m$ and thus the query complexity is reduced to $\Omega((d_m/D)^{-1/2})$. Notably, if we restrict the initial state to be $\ket{\psi}$ and the intermittent unitaries to be $U_j = \exp (i\phi_j \ket{\psi}\bra{\psi})$, as in Grover's or QSVT algorithm, then we must have $\left| \braket{\psi^{(k)} | \psi} \right|^2 = 1$ and thus the query complexity becomes $\Omega(1/\sqrt{p_m})$ regardless of the sign of $p_m-d_m/D$.

\section{Infidelity and polynomial approximation error}
\label{app:polyerror}

In this Appendix, we provide a detailed analysis of the required accuracy of a polynomial approximation to achieve a desired fidelity for the task of mixed-state post-selection discussed in Sec.~\ref{sec:mixed}, and separately for the task of teleportation decoding discussed in Sec.~\ref{sec:teleportation}. In particular, we explain why, for both tasks, it is necessary to bound the multiplicative error rather than the additive error.

\begin{theorem} \label{thm:infidelity_mix}
    Assume the notations in Sec.~\ref{sec:mixed}. If, for the range of $p_{am}$, the QSVT algorithm realizes the transformation
    \begin{equation}    \label{eq:LAA_w_error}
        f\left( x \right) = \frac{x}{\sqrt{p^*}} \left( 1+\delta(x) \right)    ,
    \end{equation}
    where $|\delta(x)|\leq \delta <1$, then $F_{\mathrm{QSVT}} \geq 1-\delta^2$. 
\end{theorem}

\begin{proof}
    Our starting point is the expression for $F_{\mathrm{QSVT}}$ derived in Eq.~\eqref{eq:FQSVT_mix}, which we transcribe here for convenience:
    \begin{equation}
        F_{\mathrm{QSVT}} = \dfrac{ \left( \sum_a \sqrt{p_{am}}f(\sqrt{p_{am}}) \right)^2 }{ d_R p_m \sum_a f(\sqrt{p_{am}})^2 }    .
    \end{equation}
    In order to compare multiplicative and additive errors while proving the theorem, we rewrite the transformation in an additive error form:
    \begin{equation}
        f(x) = \frac{x}{\sqrt{p^*}} + \delta'(x).
    \end{equation}
    With some straightforward expansion, we find
    \begin{equation}    \label{eq:FQSVT_mix_expansion}
        F_{\mathrm{QSVT}} = \dfrac{ 1 + 2\Delta_1 + \Delta_1^2 }{ 1 + 2\Delta_1 + \Delta_2 }   ,
    \end{equation}
    where     
    \begin{equation}
        \Delta_1 = \frac{\sqrt{p^*}}{d_R p_m} \sum_a \sqrt{p_{am}} \delta'(\sqrt{p_{am}}),
    \end{equation}
    \begin{equation}
        \Delta_2 = \frac{p^*}{d_Rp_m} \sum_a \delta'(\sqrt{p_{am}})^2   .
    \end{equation}
    Now, we use the relation $\delta'(x) = \frac{x}{\sqrt{p^*}} \delta(x)$ to obtain the following bounds for $\Delta_1$ and $\Delta_2$:
    \begin{equation}
        |\Delta_1| \leq \frac{\sqrt{p^*}}{d_R p_m} \sum_a \frac{p_{am}}{\sqrt{p^*}} \delta = \delta,
    \end{equation}
    \begin{equation}
        0 \leq \Delta_2 \leq \frac{p^*}{d_Rp_m} \sum_a \frac{p_{am}}{p^*} \delta^2 = \delta^2   .
    \end{equation}

    Viewing $F_{\mathrm{QSVT}}$ as a function of $\Delta_1$ and $\Delta_2$, we take the partial derivative with respect to $\Delta_1$ and find that $F_{\mathrm{QSVT}}$ is minimized when $\Delta_1 = -\Delta_2$. Plugging this into Eq.~\eqref{eq:FQSVT_mix_expansion} yields
    \begin{equation}
        F_{\mathrm{QSVT}} \geq \frac{1-2\Delta_2  +\Delta_2^2}{1-2\Delta_2+\Delta_2} = 1-\Delta_2 \geq 1-\delta^2   ,
    \end{equation}
    as asserted in the theorem. 

    Finally, we demonstrate why bounding the additive error alone is insufficient. Without employing the relation $\delta'(x) = \frac{x}{\sqrt{p^*}} \delta(x)$ and instead assuming $\delta'(x)\leq \delta'$, we would only obtain the following bounds on $\Delta_1$ and $\Delta_2$:
    \begin{equation}
        |\Delta_1| \leq \frac{\sqrt{p^*}}{d_Rp_m} \sum_a \sqrt{p_{am}} \delta' \leq \frac{p_{\max}}{p_m} \delta'   ,
    \end{equation}
    \begin{equation}
        0 \leq \Delta_2 \leq \frac{p^*}{p_m} \delta'^2   .
    \end{equation}
    In the above, we used the fact that $p^* \geq p_{\max} = \max_a p_{am}$, since all $p_{am}$ lie in the range from $0$ to $p^*$. Consequently, instead of having $|\Delta_2| \leq \delta'$, the bound acquires a spectrum-dependent factor of $p^*/p_m>1$. Although this factor may be modest for many distributions of $p_{am}$, it can become large if the distribution has long tails. Thus, we have shown that it is necessary to bound the multiplicative error in order to lower-bound the fidelity for an arbitrary input state.
\end{proof}

\begin{theorem} \label{thm:infidelity_tel}
    Assume the notations in Sec.~\ref{sec:teleportation}. If, over the range of $p_{am}$, the QSVT algorithm realizes the transformation
    \begin{equation}    \label{eq:inverse_w_error}
        f\left( x \right) = \frac{\sqrt{p^*}}{x} \left( 1+\delta(x) \right)    ,
    \end{equation}
    where $|\delta(x)|\leq \delta <1$, then $F_{\mathrm{QSVT}} \geq 1-\delta^2$.
\end{theorem}

\begin{proof}
    We begin with the expression for $F_{\mathrm{QSVT}}$ in the teleportation decoding task, which is essentially the same as in the previous proof except for the expression for $p_{\mathrm{QSVT}}$:
    \begin{equation}
        F_{\mathrm{QSVT}} = \dfrac{ \left( \sum_a \sqrt{p_{am}}f(\sqrt{p_{am}}) \right)^2 }{ d_R \sum_a p_{am} f(\sqrt{p_{am}})^2 }    .
    \end{equation}
    By inserting the expression for $f(x)$ and expanding the fidelity, we obtain
    \begin{equation}
        F_{\mathrm{QSVT}} = \dfrac{ 1 + 2\Delta_1 + \Delta_1^2 }{ 1 + 2\Delta_1 + \Delta_2 }   ,
    \end{equation}
    where
    \begin{equation}
        \Delta_1 = \frac{1}{d_R} \sum_a \delta(\sqrt{p_{am}}),
    \end{equation}
    \begin{equation}
        \Delta_2 = \frac{1}{d_R} \sum_a \delta(\sqrt{p_{am}})^2   .
    \end{equation}
    It is easy to see that $|\Delta_1| \leq \delta$ and $0\leq \Delta_2 \leq \delta^2$. By applying the same reasoning as in the proof of Theorem \ref{thm:infidelity_mix}, we conclude that $F_{\mathrm{QSVT}} \geq 1-\delta^2$.

    Finally, we consider the scenario where only the additive error is bounded, i.e., 
    \begin{equation}
        f(x) = \frac{\sqrt{p^*}}{x} + \delta'(x),\quad \text{with } |\delta'(x)|\leq\delta'.
    \end{equation}
    In this case, the bounds on $\Delta_1$ and $\Delta_2$ become
    \begin{equation}
        |\Delta_1| \leq \frac{p_m}{p^*} \delta'    ,
    \end{equation}
    \begin{equation}
        0 \leq \Delta_2 \leq \frac{p_m}{p^*} \delta'^2  .
    \end{equation}
    To ensure the inversion works for all $p_{am}$, we must have $p^* \leq p_{\min} = \min_a p_{am}$. Consequently, the term $\Delta_2$ acquires a spectrum-dependent factor of $p_m/p^*>1$, which could be large if the distribution of $p_{am}$ has long tails. Therefore, it is necessary to bound the multiplicative error.
\end{proof}

\section{Uhlmann bounds on fidelity}
\label{app:Uhlmann}

In the main text, we have seen that $\omega_R = \id_R/d_R$, dubbed the decoupling condition, appears as the condition for the existence of a deterministic algorithm for both the task of mixed-state post-selection and the task of teleportation decoding. In this Appendix, we prove the corresponding robust statements. 

To prove these theorems, our central tool is Uhlmann's theorem, the proof of which can be found in textbooks such as Refs.~\cite{Nielsen_Chuang_2010,preskill_chp_2}.

\begin{lemma}[Uhlmann \cite{Uhlmann1976273}]
    Let $\rho$ and $\sigma$ be two density matrices on $A$. Let $\ket{\rho}_{AB}$ and $\ket{\sigma}_{AC}$ be purifications of $\rho$ and $\sigma$, respectively, i.e., 
    \begin{equation}
        \rho = \tr_B \ket{\rho}\bra{\rho}_{AB}, \quad \sigma = \tr_C \ket{\sigma}\bra{\sigma}_{AC}   .
    \end{equation}
    Then
    \begin{equation}
        F(\rho, \sigma) = \max_{U_{B\to C}} \left| \braket{\sigma |_{AC} U_{B\to C}| \rho }_{AB} \right|^2  ,
    \end{equation}
    where $U_{B\to C}$ can be any isometry from $B$ to $C$. In other words, the fidelity between two states equals the maximum overlap between their purifications.
\end{lemma}

We start with the theorem concerning decodability. We begin with a tripartite pure state $\ket{\omega}_{RED}$, where $R$ represents a reference system containing encoded quantum information, and $E$ denotes an inaccessible environment. A decoder attempts to apply a decoding channel on $D$ to recover the state in $R$.

The full form of decoupling condition asserts that a decoding channel exists if $\omega_{RE} = \frac{\id_R}{d_R} \otimes \omega_E$. In our study of both mixed-state post-selection and decoding teleportation, no subsystem is erased. Instead, the state $\omega_{RD|m}$ (see Fig.~\ref{fig:tel}(a)) is prepared via unitary operations and projective measurements. Here, we may consider $E$ as nonexistent, reducing the decoupling condition to $\omega_{R|m} = \id_R/d_R$.

The ``$\Leftarrow$'' part of Theorem \ref{thm:decoupling} also appears in Ref.~\cite{hayden2008decoupling}. Below is the theorem. 



\begin{theorem} \label{thm:decoupling}
    Let $\ket{\omega}_{RED}$ be a pure state. There exists a decoding channel $\mathcal{D}_{D\to R'}$ satisfying
    \begin{equation}
        F( \mathcal{D}_{D\to R'}(\omega_{RD}), \EPR_{RR'} ) \geq 1-\varepsilon
    \end{equation}
    if and only if there exists a (possibly mixed) state $\tau_E$ on $E$ such that
    \begin{equation}
        F( \omega_{RE}, \frac{\id_R}{d_R} \otimes \tau_E ) \geq 1-\varepsilon   . 
    \end{equation}
\end{theorem}

\begin{proof}
    ``$\Leftarrow$'': We start from $F( \omega_{RE}, \frac{\id_R}{d_R} \otimes \tau_E ) \geq 1-\varepsilon$. To use Uhlmann's theorem, we need to introduce purifications for $\omega_{RE}$ and $\frac{\id_R}{d_R} \otimes \tau_E$. For $\omega_{RE}$, $\ket{\omega}_{RED}$ is a purification. For $\frac{\id_R}{d_R} \otimes \tau_E$, the purification is chosen to be $\EPR_{RR'} \otimes \ket{\tau}_{EE'}$, where  $R'$ and $E'$ are identical quantum registers as $R$ and $E$, respectively, and $\ket{\tau}_{EE'}$ is a purification of $\tau_E$. So there exists an isometry $U_{D\to R'E'}$ that satisfies
    \begin{equation}
        F( U_{D\to R'E'}\ket{\omega}_{RED}, \EPR_{RR'} \otimes \ket{\tau}_{EE'} ) \geq 1-\varepsilon  .
    \end{equation}
    In fact, $U_{D\to R'E'}$ is a Stinespring dilation of the channel we are looking for. In other words, let
    \begin{equation}
        \mathcal{D}_{D\to R'}(\cdot) = \tr_{E'} \left( U_{D\to R'E'} (\cdot) U_{D\to R'E'}^\dagger \right)    .
    \end{equation}
    Tracing out $EE'$ for both states does not decrease the fidelity (which can be derived from Uhlmann's theorem). Therefore, we get
    \begin{equation}
        \begin{aligned}
            &\phantom{\geq} F(\mathcal{D}_{D\to R'}(\omega_{RD}) ,  \EPR_{RR'}) \\
            &\geq F( U_{D\to R'E'}\ket{\omega}_{RED}, \EPR_{RR'} \otimes \ket{\tau}_{EE'} ) \\
            &\geq 1-\varepsilon   .
        \end{aligned}
    \end{equation}

    ``$\Rightarrow$'': Suppose there exists a decoding channel $\mathcal{D}_{D\to R'}$ such that $F( \mathcal{D}_{D\to R'}(\omega_{RD}), \EPR_{RR'} ) \geq 1-\varepsilon$. Similar to the proof of ``$\Leftarrow$'', we introduce systems $E$ and $E'$ and a Stinespring dilation, such that $U_{D\to R'E'} \ket{\omega}_{RED}$ is a purification of $\mathcal{D}_{D\to R'}(\omega_{RD})$. For $\EPR_{RR'}$, since it is already pure, the ``purification'' of it on $RR'EE'$ must be of the form $\EPR_{RR'} \otimes \ket{\tau}_{EE'}$. Therefore, there exists a $\ket{\tau}_{EE'}$ that satisfies
    \begin{equation}
        \begin{aligned}
            &\phantom{=} F( U_{D\to R'E'} \ket{\omega}_{RED}, \EPR_{RR'} \otimes \ket{\tau}_{EE'} ) \\
            &= F(\mathcal{D}_{D\to R'}(\omega_{RD}), \EPR_{RR'}) .
        \end{aligned}
    \end{equation}
    For these two pure states, we trace out $R'E'$ and the fidelity does not decrease, i.e. (note that by doing partial trace $U_{D\to R'E'}$ is canceled), 
    \begin{equation}
        \begin{aligned}
            &\phantom{=} F( \omega_{RE}, \frac{\id_R}{d_R} \otimes \tau_E ) \\
            &\geq F( U_{D\to R'E'} \ket{\omega}_{RED}, \EPR_{RR'} \otimes \ket{\tau}_{EE'} ) \\
            &\geq 1-\varepsilon ,
        \end{aligned}
    \end{equation}
    where $\tau_E = \tr_{E'} \ket{\tau}\bra{\tau}_{EE'}$. 
\end{proof}

Regarding simulation of the post-selected mixed state, we establish the following theorem. 

\begin{theorem} \label{thm:decoupling_sim}
    Suppose $\ket{\omega}_{RS}$ and $\ket{\omega'}_{RS}$ are pure states. There exists a channel on $S$, denoted as $\mathcal{D}_S$, such that
    \begin{equation}
        F \left( \mathcal{D}_S (\omega_{RS}), \omega'_{RS} \right) \geq 1-\varepsilon   
    \end{equation} 
    if and only if
    \begin{equation}
        F(\omega_R, \omega'_R) \geq 1-\varepsilon   .
    \end{equation}   
\end{theorem}

\begin{proof}
    ``$\Leftarrow$'': This is simply by one step of Uhlmann's theorem: there exists a unitary $U_{S}$ satisfying
    \begin{equation}
        F \left( U_S\ket{\omega}_{RS}, \ket{\omega'}_{RS} \right) 
        = F(\omega_R, \omega'_R)    .
    \end{equation}
    In fact, not only have we shown the existence of a channel on $S$, it can even be a unitary operator. 

    ``$\Rightarrow$'': Here, we start from $F \left( \mathcal{D}_S (\omega_{RS}), \omega'_{RS} \right)$. We introduce purifications for the two states involved. Let $U_{S\to SE}$ be a Stinespring dilation of $\mathcal{D}_S$. Then a purification for $\mathcal{D}_S (\omega_{RS})$ is $U_{S\to SE} \ket{\omega}_{RS}$, and for $\omega'_{RS}$ we use $\ket{\omega'}_{RS}\otimes\ket{\tau}_E$. Then there exists a $\ket{\tau}_E$ such that
    \begin{equation}
        F \left( U_{S\to SE} \ket{\omega}_{RS}, \ket{\omega'}_{RS} \otimes \ket{\tau}_E \right) \geq 1-\varepsilon  .
    \end{equation}
    Finally, we trace out $SE$ to get that $F(\omega_R, \omega'_R) \geq 1-\varepsilon$.
\end{proof}

To relate this theorem to our setup in Sec.~\ref{sec:mixed}, we consider $\ket{\omega}_{RS}$ to be the pre-measurement state $\ket{\Psi}_{RS} = \frac{1}{d_R} \sum_{a=1}^{d_R} \ket{a}_R \ket{\psi_a}_S$, and $\ket{\omega'}_{RS}$ to be the post-measurement state $\ket{\Psi_m}_{RS} = \frac{1}{\sqrt{d_R p_m}} \sum_{a=1}^{d_R} \sqrt{p_{am}} \ket{a}_R\ket{\psi_{am}}_S$. Additionally, by construction $\omega_R = \id_R/d_R$. It follows that the condition for the existence of a deterministic mixed state post-selection simulator is $\omega'_R = \id_R/d_R$, which is the same as the decoupling condition. But as we have seen, the underlying information-theoretic task and derivation are different.

\section{Lower bounds on QSVT success probability}
\label{app:pQSVTbound}

In this Appendix, we derive lower bounds for the success probability of the QSVT algorithm in the mixed-state post-selection task and the teleportation decoding task. 

\begin{theorem} \label{thm:tailbound_mix}
    Consider the mixed-state post-selection task in Sec.~\ref{sec:mixed} and assume its notations. Suppose the LAA by QSVT algorithm achieves fidelity $F_{\mathrm{QSVT}} \geq 1-\varepsilon$. Then the success probability of this algorithm is greater than $( 1-\varepsilon ) \alpha$, where $\alpha$ satisfies
    \begin{equation}
        \frac{1}{d_Rp_m} \sum_{p_{am}>p_m/\alpha} p_{am} \leq \varepsilon   .
    \end{equation}
\end{theorem}

\begin{proof}
    Here, we model the QSVT function as
    \begin{equation}    \label{eq:LAA_func_new}
        f(x) = \begin{cases}
            x/\sqrt{p^*} & |x| \leq \sqrt{p^*}   \\
            0 & \sqrt{p^*} < |x| \leq 1
        \end{cases}   .
    \end{equation}
    Similar to the discussion in Sec.~\ref{sec:mixed}, we have suppressed the polynomial approximation error to zero since it is computationally cheap. It is easy to see that Eq.~\eqref{eq:LAA_func_new} results in worse fidelity than the modeling $f(x)$ as in Eq.~\eqref{eq:LAAfunc}, since the latter gives larger overlap with the ideal state $\ket{\Psi_m}_{RS}$. But using Eq.~\eqref{eq:LAA_func_new} will greatly simplify the following analysis. 

    Inserting Eq.~\eqref{eq:LAA_func_new} into Eqs.~\eqref{eq:FQSVT_mix} and \eqref{eq:pQSVT}, we have
    \begin{equation}
        F_{\mathrm{QSVT}} 
        = \frac{1}{d_Rp_m} \sum_{p_{am}\leq p^*} p_{am} ,
    \end{equation}
    and 
    \begin{equation}
        p_{\mathrm{QSVT}} = \frac{p_m}{p^*} F_{\mathrm{QSVT}}   .
    \end{equation}

    Starting from $F_{\mathrm{QSVT}} \geq 1-\varepsilon$, we solve for allowed values of $p^*$. Since $\sum_{a} p_{am} = d_R p_m$, $p^*$ can be any value that satisfies
    \begin{equation}
        \frac{1}{d_Rp_m} \sum_{p_{am}> p^*} p_{am} \leq \varepsilon    .
    \end{equation}
    On the other hand, $p_{\mathrm{QSVT}} = \frac{p_m}{p^*} F_{\mathrm{QSVT}} \geq \frac{p_m}{p^*} (1-\varepsilon)$. Finally, we introduce $\alpha = p_m/p^*$ and this inequality becomes the lower bound as claimed in the theorem. 
\end{proof}

\begin{theorem} \label{thm:tailbound_tel}
    Consider the decoding task in Sec.~\ref{sec:teleportation} and assume the notations within. Suppose the pseudoinverse decoder achieves fidelity $F_{\mathrm{QSVT}} \geq 1-\varepsilon$. Then the success probability of this decoder is greater than $( 1-\varepsilon ) \alpha$, where $\alpha$ satisfies
    \begin{equation}    \label{eq:tailbound_tel_condition}
        \Pr\left( p_{am} < \alpha p_m \right) 
        \equiv \frac{1}{d_R} \sum_{p_{am}<\alpha p^*} 1
        \leq \varepsilon   .
    \end{equation}
\end{theorem}

\begin{proof}
    Let the QSVT function be modeled by 
    \begin{equation}    \label{eq:inv_func_new}
        f(x) = \begin{cases}
            \sqrt{p^*}/x & \sqrt{p^*} \leq |x| \leq 1   \\
            0 & |x| < \sqrt{p^*}
        \end{cases}   
    \end{equation}
    for simplicity. Using this will invariably result in a lower fidelity than modeling $f(x)$ as Eq.~\eqref{eq:inv_func} in the main text. 

    Inserting Eq.~\eqref{eq:inv_func_new} into Eqs.~\eqref{eq:FQSVT_tel} and \eqref{eq:pQSVT_tel}, we get
    \begin{equation}
        F_{\mathrm{QSVT}} 
        = \frac{1}{d_R} \sum_{p_{am}\geq p^*} 1
        = \Pr \left( p_{am} \geq p^* \right)    ,
    \end{equation}
    and
    \begin{equation}
        p_{\mathrm{QSVT}} = \frac{p^*}{p_m} F_{\mathrm{QSVT}}   .
    \end{equation}

    Given $F_{\mathrm{QSVT}} \geq 1-\varepsilon$, we first need to solve for $p^*$. $F_{\mathrm{QSVT}}$ turns out to be exactly the complementary cumulative distribution function. Hence, $p^*$ can be any value that satisfies
    \begin{equation}
        \Pr \left( p_{am} < p^* \right) \leq \varepsilon    .
    \end{equation}
    Furthermore, the success probability is greater than $\frac{p^*}{p_m} (1-\varepsilon)$. Substituting $p^* = \alpha p_m$ proves the claim in the theorem. 
\end{proof}

From the above two theorems, we see that whether the QSVT algorithm can succeed with $O(1)$ probability crucially depends on whether there is an $O(1)$ solution of $\alpha$, which is a model-dependent problem. As a crude picture, for the mixed-state post-selection task, is it desirable that there is no $p_{am}$'s that appear at much larger value than the average $p_m$. For the teleportation decoding task, the desirable property is that only a vanishing portion of $p_{am}$'s concentrate at values much smaller than $p_m$.

The following is a corollary from Theorem~\ref{thm:tailbound_tel}, which shows that $p_{\mathrm{QSVT}}$ can be $O(1)$ when applied to approximate quantum codes. 

\begin{corollary}   \label{cor:AQEC_bound_pQSVT}
    Suppose the pseudoinverse decoder satisfies $F_{\mathrm{QSVT}}\geq 1-\varepsilon$. If the probabilities $p_{am}$ satisfy
    \begin{equation}
        \frac{1}{2d_R} \sum_{a=1}^{d_R} \left| \frac{p_{am}}{p_m} - 1 \right| \leq \varepsilon' ,
    \end{equation}
    then
    \begin{equation}
        p_{\mathrm{QSVT}} \geq (1-\varepsilon'/\varepsilon)(1-\varepsilon)   .
    \end{equation}
\end{corollary}

\begin{proof}
    Using Markov's inequality, we can derive
    \begin{equation}
        \Pr(p_{am}<\alpha p^*) \leq \frac{\varepsilon'}{1-\alpha}    .
    \end{equation}
    We now use this to bound $\alpha$ in the statement of Theorem~\ref{thm:tailbound_tel}. $\alpha$ satisfying the following equality also satisfies Eq.~\eqref{eq:tailbound_tel_condition}:
    \begin{equation}
        \frac{\varepsilon'}{1-\alpha} = \varepsilon.
    \end{equation}
    It follows that $\alpha$ can be at least $1-\varepsilon'/\varepsilon$.
\end{proof}



\section{Comparison with other decoders}
\label{app:decoder}

This Appendix expands the discussion in Sec.~\ref{sec:teleportation} on the quantum teleportation and decoding protocols. Here, we provide a self-contained introduction to the generalized Yoshida--Kitaev (YK) decoder and the Petz-like decoder proposed in Ref.~\cite{utsumi2024explicit}. We then adapt these two types of decoders to the teleportation decoding task. This facilitates the comparison between the FPAA (Petz-like) decoder and the pseudoinverse decoder.

\subsection{Decoding after decoherence}

Ref.~\cite{utsumi2024explicit} focuses on decoding after a noise channel, which differs from the focus of this work: decoding after a forced measurement.
A slightly simplified setup of the decoding problem considered in Ref.~\cite{utsumi2024explicit} is as follows. Alice holds a secret state in $A$, and she also holds a reference system $R$ that is maximally entangled with $A$. Alice's state is isometrically encoded into a larger Hilbert space by initializing another system $B$ in the state $\ket{0}_B$ and applying a unitary $U$ to $AB$. Up to this point, the state of the system is denoted as
\begin{equation}    \label{eq:omegaRED}
    \omega_{RED} = 
    \vcenter{\hbox{\includegraphics{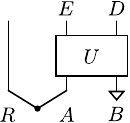}}},
\end{equation}
where $ED$ is a different bipartition of $AB$. After encoding, $E$ is erased (cf. Stinespring dilation of quantum channels). The decoder Bob attempts to recover Alice's initial state by applying quantum operations on $D$. Bob is assumed to know both $U$ and the initial state of $B$.

Although several types of decoders exist and may appear different, they can all be understood through the following common strategy. Bob prepares the same state $\omega_{R'E'D'}$ on his own quantum registers. Upon receiving $\omega_D$, he compares it with his own $\omega_{D'}$. If $\omega_{D'} \approx \omega_{D}$, then ideally $\omega_{R'}$ should be close to $\omega_R$. While this is not always guaranteed, the strategy works when the decoupling condition approximately holds, as we derive below.

A verbatim implementation of the above idea yields the generalized YK decoder. It is illustrated in the following circuit diagram:
\begin{equation}    \label{eq:YKdecoder}
    \vcenter{\hbox{\includegraphics{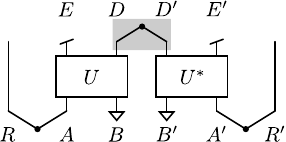}}}, 
\end{equation}
where the gray box denotes a measurement and post-selection on $DD'$ being in the EPR state. When the decoding condition is satisfied, the resulting state on $RR'$ is approximately the EPR state. The quality of this decoding protocol is characterized by the fidelity between the state in Eq.~\eqref{eq:YKdecoder} and the EPR state between $R$ and $R'$. This state must also be normalized by the post-selection success probability. Therefore, the decoding fidelity (i.e., the fidelity with $\EPR_{RR'}$) is given by
\begin{equation}
    F = \tilde{F}/p_{\text{succ}}
    = \frac{1}{d_R} e^{S^{(2)}\left( \omega_{D} \right) - S^{(2)}\left( \omega_{RD} \right)}   ,
\end{equation}
where
\begin{equation}
    \begin{aligned}
        \tilde{F} &= 
    \vcenter{\hbox{\includegraphics{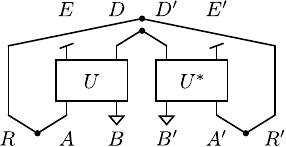}}}  \\
        &= \frac{1}{d_R d_D} P^{(2)} \left( \omega_{RD} \right)   ,
    \end{aligned}
\end{equation}
and 
\begin{equation}
    \begin{aligned}
        p_{\text{succ}} &= 
    \vcenter{\hbox{\includegraphics{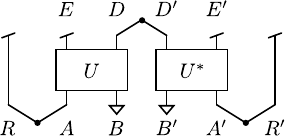}}}  \\
        &= \frac{1}{d_D} P^{(2)} \left( \omega_{D} \right)    .
    \end{aligned}
\end{equation}
Since $\omega_{RED}$ is a pure state, $F$ can alternatively be expressed as $F = \frac{1}{d_R} e^{S^{(2)}(\omega_{RE}) - S^{(2)}(\omega_{E})}$. If the decoupling condition holds, i.e., $\omega_{RE} \approx \frac{\id_R}{d_R}\otimes \omega_E$, then $S^{(2)}(\omega_{RE}) \approx \log d_R + S^{(2)}(\omega_{E})$, so $F \approx 1$.

The decoder in Eq.~\eqref{eq:YKdecoder} succeeds with probability $p_{\text{succ}}$, which is generally low. One can instead replace the post-selection on $\EPR_{DD'}$ with an FPAA algorithm. The complexity of FPAA is related to the minimum eigenvalue of $\omega_{RE}$ and the desired accuracy of the polynomial approximation, as discussed in detail in Ref.~\cite{utsumi2024explicit}.

The Petz-like decoder also enforces $\omega_D = \omega_{D'}$, subsequently anticipating that $\omega_R$ and $\omega_{R'}$ are identical; however, it does so by back-evolving $\omega_D$ and requiring that $B$ returns to its initial state $\ket{0}_B$. The Petz-like decoder is shown in the following circuit diagram:
\begin{equation}    \label{eq:Petzdecoder}
    \vcenter{\hbox{\includegraphics{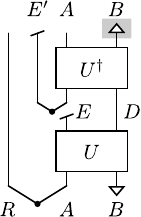}}}
\end{equation}
Conditioned on $B$ being measured in the $\ket{0}_B$ state, $A$ and $R$ are approximately in the EPR state when the decoupling condition holds. The decoding fidelity and success probability can be calculated similarly to the YK case:
\begin{equation}
    F = \tilde{F}/p_{\text{succ}}
    = \frac{1}{d_R} e^{S^{(2)}\left( \omega_{D} \right) - S^{(2)}\left( \omega_{RD} \right)}   ,
\end{equation}
where 
\begin{equation}
    \begin{aligned}
        \tilde{F} &= 
        \vcenter{\hbox{\includegraphics{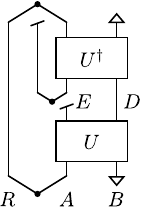}}}    \\
        &= \frac{1}{d_E} P^{(2)} \left( \omega_{RD} \right)   ,
    \end{aligned}
\end{equation}
and 
\begin{equation}
    \begin{aligned}
        p_{\text{succ}} &= 
        \vcenter{\hbox{\includegraphics{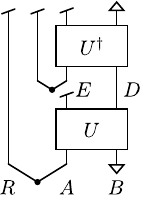}}}    \\
        &= \frac{d_R}{d_E} P^{(2)} \left( \omega_{D} \right)    .
    \end{aligned}
\end{equation}
The post-selection on $\ket{0}_B$ in the Petz-like decoder can also be replaced by FPAA. Comparing the YK and Petz decoders, we find that they yield the same fidelity. When $d_E > d_R d_D$, the YK decoder achieves a higher success probability and requires less quantum memory to implement; the opposite is true when $d_E < d_R d_D$.

\subsection{Decoding for teleportation}

The setup we study is the same as described at the beginning of Sec.~\ref{sec:teleportation}. After the encoding stage in Eq.~\eqref{eq:omegaRED}, we now consider that $E$ is projected onto a pure state $\ket{m}_E$ due to a projective measurement. We denote the post-measurement state as 
\begin{equation}    \label{eq:omegaRDm}
    \omega_{RD|m} = \frac{1}{\sqrt{p_m}}
    \vcenter{\hbox{\includegraphics{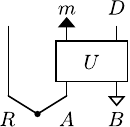}}}    ,
\end{equation}
where the black triangle represents $\bra{m}_E$, and $p_m$ is the probability of obtaining outcome $m$, introduced to normalize the state. Bob (the decoder) attempts to recover Alice's initial state by applying $m$-dependent quantum operations on $D$. This setup is structurally similar to the decoherence-decoding protocol studied in the previous subsection, the only difference being that the partial trace over $E$ is now replaced by a projective measurement. For teleportation decoding, we propose two types of decoders that are analogous to the generalized YK decoder and the Petz-like decoder in Ref.~\cite{utsumi2024explicit}. 


In the spirit of the YK decoder, we construct the following decoder:
\begin{equation}
    \vcenter{\hbox{\includegraphics{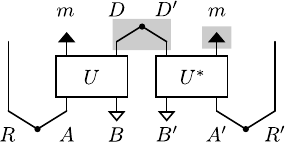}}}
\end{equation}
In words, Bob's decoding proceeds as follows. He prepares the same state $\omega_{R'D'|m}^*$ on his own quantum registers by post-selecting on $\ket{m}_{E'}$. Upon receiving $\omega_{D|m}$, he measures and post-selects $DD'$ onto the EPR state. Conditioned on the success of these post-selections (the gray boxes), $RR'$ is expected to be in the EPR state, an we will derive the condition for this to hold. The decoding fidelity and success probability are given by
\begin{equation}
    \begin{aligned}
        F &= \vcenter{\hbox{\includegraphics{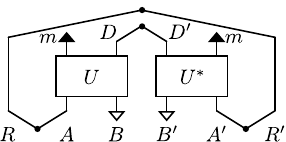}}}   \\
        &\phantom{=} \Big{/} \vcenter{\hbox{\includegraphics{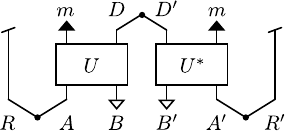}}} \\
        &= \frac{1}{d_R} e^{S^{(2)}\left( \omega_{D|m} \right) - S^{(2)}\left( \omega_{RD|m} \right)}   .
    \end{aligned}
\end{equation}
This can be further simplified because $\omega_{RD|m}$ is a pure state:
\begin{equation}    \label{eq:F_YK}
    F = \frac{1}{d_R} e^{S^{(2)}\left( \omega_{D|m} \right)} 
    = \frac{1}{d_R} e^{S^{(2)}\left( \omega_{R|m} \right)}   .
\end{equation}
It follows that $F = 1$ only when the entropy of $\omega_{R|m}$ is maximal, i.e., when $\omega_{R|m} = \frac{\id_R}{d_R}$. 
The success probability of the decoder is the probability of measuring $\ket{m}_E \EPR_{DD'} \ket{m}_{E'}$ \emph{conditioned on} $E$ being in the state $\ket{m}_E$. This is given by
\begin{equation}
    \begin{aligned}
        p_{\text{succ}} &= \frac{1}{p_m} \vcenter{\hbox{\includegraphics{YK_post_denominator.pdf}}}   \\
        &= \frac{p_m}{d_D} e^{-S^{(2)}\left( \omega_{D|m} \right)}    .
    \end{aligned}
\end{equation}

On the other hand, the teleportation Petz-like decoder is given by
\begin{equation}
    \vcenter{\hbox{\includegraphics{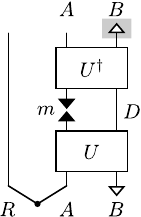}}}
\end{equation}
Again, the post-selection on $\ket{0}_B$ can be replaced by FPAA, resulting in a deterministic decoder. For this Petz decoder, the decoding fidelity shares the same expression as in Eq.~\eqref{eq:F_YK}, while the success probability is $d_R p_m e^{-S^{(2)}\left( \omega_{R|m} \right)}$.

As teleportation decoders, the Petz decoder invariably uses fewer ancilla qubits than the YK decoder. At the end of Sec.~\ref{sec:teleportation}, we compared our pseudoinverse decoder with an ``FPAA decoder''. One can concretely interpret the FPAA decoder as the Petz-like decoder constructed in this section.

\subsection{Decoders as QSVT}

In the language of QSVT, both the YK and Petz decoders implement the same functional transformation---an approximation of the sign function. Their only difference lies in the fact that they correspond to different block encodings of the same matrix, namely,
\begin{equation}
    M^\dagger = \sum_{a=1}^{d_R} \ket{a}_R \bra{\psi_{am}}_D    ,
\end{equation}
where $M$ is defined in Eq.~\eqref{eq:mix_block_enc}. Concretely, the YK decoder uses the following block encoding of $M^\dagger$, rescaled by a factor of $1/\sqrt{d_D d_R}$:
\begin{equation}
    \vcenter{\hbox{\includegraphics{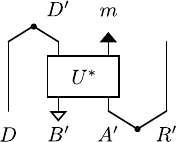}}}
\end{equation}
whereas in the Petz decoder, the block encoding of $M^\dagger$ is given by (noting that $A$ is isometric to $R$):
\begin{equation}
    \vcenter{\hbox{\includegraphics{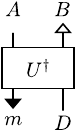}}}
\end{equation}
As a final remark, our pseudoinverse decoder in Sec.~\ref{sec:teleportation} uses the same block encoding as the Petz decoder.

\bibliography{references}

\begin{thebibliography}{69}%
\makeatletter
\providecommand \@ifxundefined [1]{%
 \@ifx{#1\undefined}
}%
\providecommand \@ifnum [1]{%
 \ifnum #1\expandafter \@firstoftwo
 \else \expandafter \@secondoftwo
 \fi
}%
\providecommand \@ifx [1]{%
 \ifx #1\expandafter \@firstoftwo
 \else \expandafter \@secondoftwo
 \fi
}%
\providecommand \natexlab [1]{#1}%
\providecommand \enquote  [1]{``#1''}%
\providecommand \bibnamefont  [1]{#1}%
\providecommand \bibfnamefont [1]{#1}%
\providecommand \citenamefont [1]{#1}%
\providecommand \href@noop [0]{\@secondoftwo}%
\providecommand \href [0]{\begingroup \@sanitize@url \@href}%
\providecommand \@href[1]{\@@startlink{#1}\@@href}%
\providecommand \@@href[1]{\endgroup#1\@@endlink}%
\providecommand \@sanitize@url [0]{\catcode `\\12\catcode `\$12\catcode `\&12\catcode `\#12\catcode `\^12\catcode `\_12\catcode `\%12\relax}%
\providecommand \@@startlink[1]{}%
\providecommand \@@endlink[0]{}%
\providecommand \url  [0]{\begingroup\@sanitize@url \@url }%
\providecommand \@url [1]{\endgroup\@href {#1}{\urlprefix }}%
\providecommand \urlprefix  [0]{URL }%
\providecommand \Eprint [0]{\href }%
\providecommand \doibase [0]{https://doi.org/}%
\providecommand \selectlanguage [0]{\@gobble}%
\providecommand \bibinfo  [0]{\@secondoftwo}%
\providecommand \bibfield  [0]{\@secondoftwo}%
\providecommand \translation [1]{[#1]}%
\providecommand \BibitemOpen [0]{}%
\providecommand \bibitemStop [0]{}%
\providecommand \bibitemNoStop [0]{.\EOS\space}%
\providecommand \EOS [0]{\spacefactor3000\relax}%
\providecommand \BibitemShut  [1]{\csname bibitem#1\endcsname}%
\let\auto@bib@innerbib\@empty
\bibitem [{\citenamefont {Brydges}\ \emph {et~al.}(2019)\citenamefont {Brydges}, \citenamefont {Elben}, \citenamefont {Jurcevic}, \citenamefont {Vermersch}, \citenamefont {Maier}, \citenamefont {Lanyon}, \citenamefont {Zoller}, \citenamefont {Blatt},\ and\ \citenamefont {Roos}}]{Brydges2019probing}%
  \BibitemOpen
  \bibfield  {author} {\bibinfo {author} {\bibfnamefont {T.}~\bibnamefont {Brydges}}, \bibinfo {author} {\bibfnamefont {A.}~\bibnamefont {Elben}}, \bibinfo {author} {\bibfnamefont {P.}~\bibnamefont {Jurcevic}}, \bibinfo {author} {\bibfnamefont {B.}~\bibnamefont {Vermersch}}, \bibinfo {author} {\bibfnamefont {C.}~\bibnamefont {Maier}}, \bibinfo {author} {\bibfnamefont {B.~P.}\ \bibnamefont {Lanyon}}, \bibinfo {author} {\bibfnamefont {P.}~\bibnamefont {Zoller}}, \bibinfo {author} {\bibfnamefont {R.}~\bibnamefont {Blatt}},\ and\ \bibinfo {author} {\bibfnamefont {C.~F.}\ \bibnamefont {Roos}},\ }\bibfield  {title} {\bibinfo {title} {Probing rényi entanglement entropy via randomized measurements},\ }\href {https://doi.org/10.1126/science.aau4963} {\bibfield  {journal} {\bibinfo  {journal} {Science}\ }\textbf {\bibinfo {volume} {364}},\ \bibinfo {pages} {260} (\bibinfo {year} {2019})},\ \Eprint {https://arxiv.org/abs/https://www.science.org/doi/pdf/10.1126/science.aau4963} {https://www.science.org/doi/pdf/10.1126/science.aau4963} \BibitemShut {NoStop}%
\bibitem [{\citenamefont {Elben}\ \emph {et~al.}(2020)\citenamefont {Elben}, \citenamefont {Vermersch}, \citenamefont {van Bijnen}, \citenamefont {Kokail}, \citenamefont {Brydges}, \citenamefont {Maier}, \citenamefont {Joshi}, \citenamefont {Blatt}, \citenamefont {Roos},\ and\ \citenamefont {Zoller}}]{PhysRevLett.124.010504}%
  \BibitemOpen
  \bibfield  {author} {\bibinfo {author} {\bibfnamefont {A.}~\bibnamefont {Elben}}, \bibinfo {author} {\bibfnamefont {B.}~\bibnamefont {Vermersch}}, \bibinfo {author} {\bibfnamefont {R.}~\bibnamefont {van Bijnen}}, \bibinfo {author} {\bibfnamefont {C.}~\bibnamefont {Kokail}}, \bibinfo {author} {\bibfnamefont {T.}~\bibnamefont {Brydges}}, \bibinfo {author} {\bibfnamefont {C.}~\bibnamefont {Maier}}, \bibinfo {author} {\bibfnamefont {M.~K.}\ \bibnamefont {Joshi}}, \bibinfo {author} {\bibfnamefont {R.}~\bibnamefont {Blatt}}, \bibinfo {author} {\bibfnamefont {C.~F.}\ \bibnamefont {Roos}},\ and\ \bibinfo {author} {\bibfnamefont {P.}~\bibnamefont {Zoller}},\ }\bibfield  {title} {\bibinfo {title} {Cross-platform verification of intermediate scale quantum devices},\ }\href {https://doi.org/10.1103/PhysRevLett.124.010504} {\bibfield  {journal} {\bibinfo  {journal} {Phys. Rev. Lett.}\ }\textbf {\bibinfo {volume} {124}},\ \bibinfo {pages} {010504} (\bibinfo {year} {2020})}\BibitemShut {NoStop}%
\bibitem [{\citenamefont {Joshi}\ \emph {et~al.}(2022)\citenamefont {Joshi}, \citenamefont {Elben}, \citenamefont {Vikram}, \citenamefont {Vermersch}, \citenamefont {Galitski},\ and\ \citenamefont {Zoller}}]{PhysRevX.12.011018}%
  \BibitemOpen
  \bibfield  {author} {\bibinfo {author} {\bibfnamefont {L.~K.}\ \bibnamefont {Joshi}}, \bibinfo {author} {\bibfnamefont {A.}~\bibnamefont {Elben}}, \bibinfo {author} {\bibfnamefont {A.}~\bibnamefont {Vikram}}, \bibinfo {author} {\bibfnamefont {B.}~\bibnamefont {Vermersch}}, \bibinfo {author} {\bibfnamefont {V.}~\bibnamefont {Galitski}},\ and\ \bibinfo {author} {\bibfnamefont {P.}~\bibnamefont {Zoller}},\ }\bibfield  {title} {\bibinfo {title} {Probing many-body quantum chaos with quantum simulators},\ }\href {https://doi.org/10.1103/PhysRevX.12.011018} {\bibfield  {journal} {\bibinfo  {journal} {Phys. Rev. X}\ }\textbf {\bibinfo {volume} {12}},\ \bibinfo {pages} {011018} (\bibinfo {year} {2022})}\BibitemShut {NoStop}%
\bibitem [{\citenamefont {Bluvstein}\ \emph {et~al.}(2024)\citenamefont {Bluvstein}, \citenamefont {Evered}, \citenamefont {Geim}, \citenamefont {Li}, \citenamefont {Zhou}, \citenamefont {Manovitz}, \citenamefont {Ebadi}, \citenamefont {Cain}, \citenamefont {Kalinowski}, \citenamefont {Hangleiter}, \citenamefont {Bonilla~Ataides}, \citenamefont {Maskara}, \citenamefont {Cong}, \citenamefont {Gao}, \citenamefont {Sales~Rodriguez}, \citenamefont {Karolyshyn}, \citenamefont {Semeghini}, \citenamefont {Gullans}, \citenamefont {Greiner}, \citenamefont {Vuleti{\'{c}}},\ and\ \citenamefont {Lukin}}]{Bluvstein2024}%
  \BibitemOpen
  \bibfield  {author} {\bibinfo {author} {\bibfnamefont {D.}~\bibnamefont {Bluvstein}}, \bibinfo {author} {\bibfnamefont {S.~J.}\ \bibnamefont {Evered}}, \bibinfo {author} {\bibfnamefont {A.~A.}\ \bibnamefont {Geim}}, \bibinfo {author} {\bibfnamefont {S.~H.}\ \bibnamefont {Li}}, \bibinfo {author} {\bibfnamefont {H.}~\bibnamefont {Zhou}}, \bibinfo {author} {\bibfnamefont {T.}~\bibnamefont {Manovitz}}, \bibinfo {author} {\bibfnamefont {S.}~\bibnamefont {Ebadi}}, \bibinfo {author} {\bibfnamefont {M.}~\bibnamefont {Cain}}, \bibinfo {author} {\bibfnamefont {M.}~\bibnamefont {Kalinowski}}, \bibinfo {author} {\bibfnamefont {D.}~\bibnamefont {Hangleiter}}, \bibinfo {author} {\bibfnamefont {J.~P.}\ \bibnamefont {Bonilla~Ataides}}, \bibinfo {author} {\bibfnamefont {N.}~\bibnamefont {Maskara}}, \bibinfo {author} {\bibfnamefont {I.}~\bibnamefont {Cong}}, \bibinfo {author} {\bibfnamefont {X.}~\bibnamefont {Gao}}, \bibinfo {author} {\bibfnamefont {P.}~\bibnamefont {Sales~Rodriguez}}, \bibinfo {author} {\bibfnamefont {T.}~\bibnamefont {Karolyshyn}}, \bibinfo {author} {\bibfnamefont {G.}~\bibnamefont {Semeghini}}, \bibinfo {author} {\bibfnamefont {M.~J.}\ \bibnamefont {Gullans}}, \bibinfo {author} {\bibfnamefont {M.}~\bibnamefont {Greiner}}, \bibinfo {author} {\bibfnamefont {V.}~\bibnamefont {Vuleti{\'{c}}}},\ and\ \bibinfo {author} {\bibfnamefont {M.~D.}\ \bibnamefont {Lukin}},\ }\bibfield  {title} {\bibinfo {title} {Logical quantum processor based on reconfigurable atom arrays},\ }\href {https://doi.org/10.1038/s41586-023-06927-3} {\bibfield  {journal} {\bibinfo  {journal} {Nature}\ }\textbf {\bibinfo {volume} {626}},\ \bibinfo {pages} {58} (\bibinfo {year} {2024})}\BibitemShut {NoStop}%
\bibitem [{\citenamefont {Brand\~ao}\ \emph {et~al.}(2021)\citenamefont {Brand\~ao}, \citenamefont {Chemissany}, \citenamefont {Hunter-Jones}, \citenamefont {Kueng},\ and\ \citenamefont {Preskill}}]{PRXQuantum.2.030316}%
  \BibitemOpen
  \bibfield  {author} {\bibinfo {author} {\bibfnamefont {F.~G.}\ \bibnamefont {Brand\~ao}}, \bibinfo {author} {\bibfnamefont {W.}~\bibnamefont {Chemissany}}, \bibinfo {author} {\bibfnamefont {N.}~\bibnamefont {Hunter-Jones}}, \bibinfo {author} {\bibfnamefont {R.}~\bibnamefont {Kueng}},\ and\ \bibinfo {author} {\bibfnamefont {J.}~\bibnamefont {Preskill}},\ }\bibfield  {title} {\bibinfo {title} {Models of quantum complexity growth},\ }\href {https://doi.org/10.1103/PRXQuantum.2.030316} {\bibfield  {journal} {\bibinfo  {journal} {PRX Quantum}\ }\textbf {\bibinfo {volume} {2}},\ \bibinfo {pages} {030316} (\bibinfo {year} {2021})}\BibitemShut {NoStop}%
\bibitem [{\citenamefont {Hayden}\ and\ \citenamefont {Preskill}(2007)}]{HaydenPreskill2007}%
  \BibitemOpen
  \bibfield  {author} {\bibinfo {author} {\bibfnamefont {P.}~\bibnamefont {Hayden}}\ and\ \bibinfo {author} {\bibfnamefont {J.}~\bibnamefont {Preskill}},\ }\bibfield  {title} {\bibinfo {title} {Black holes as mirrors: quantum information in random subsystems},\ }\href {https://doi.org/10.1088/1126-6708/2007/09/120} {\bibfield  {journal} {\bibinfo  {journal} {Journal of High Energy Physics}\ }\textbf {\bibinfo {volume} {2007}},\ \bibinfo {pages} {120} (\bibinfo {year} {2007})}\BibitemShut {NoStop}%
\bibitem [{\citenamefont {Stanford}\ and\ \citenamefont {Susskind}(2014)}]{PhysRevD.90.126007}%
  \BibitemOpen
  \bibfield  {author} {\bibinfo {author} {\bibfnamefont {D.}~\bibnamefont {Stanford}}\ and\ \bibinfo {author} {\bibfnamefont {L.}~\bibnamefont {Susskind}},\ }\bibfield  {title} {\bibinfo {title} {Complexity and shock wave geometries},\ }\href {https://doi.org/10.1103/PhysRevD.90.126007} {\bibfield  {journal} {\bibinfo  {journal} {Phys. Rev. D}\ }\textbf {\bibinfo {volume} {90}},\ \bibinfo {pages} {126007} (\bibinfo {year} {2014})}\BibitemShut {NoStop}%
\bibitem [{\citenamefont {Arute}\ \emph {et~al.}(2019)\citenamefont {Arute}, \citenamefont {Arya}, \citenamefont {Babbush}, \citenamefont {Bacon}, \citenamefont {Bardin}, \citenamefont {Barends}, \citenamefont {Biswas}, \citenamefont {Boixo}, \citenamefont {Brandao}, \citenamefont {Buell}, \citenamefont {Burkett}, \citenamefont {Chen}, \citenamefont {Chen}, \citenamefont {Chiaro}, \citenamefont {Collins}, \citenamefont {Courtney}, \citenamefont {Dunsworth}, \citenamefont {Farhi}, \citenamefont {Foxen}, \citenamefont {Fowler}, \citenamefont {Gidney}, \citenamefont {Giustina}, \citenamefont {Graff}, \citenamefont {Guerin}, \citenamefont {Habegger}, \citenamefont {Harrigan}, \citenamefont {Hartmann}, \citenamefont {Ho}, \citenamefont {Hoffmann}, \citenamefont {Huang}, \citenamefont {Humble}, \citenamefont {Isakov}, \citenamefont {Jeffrey}, \citenamefont {Jiang}, \citenamefont {Kafri}, \citenamefont {Kechedzhi}, \citenamefont {Kelly}, \citenamefont {Klimov}, \citenamefont {Knysh}, \citenamefont {Korotkov}, \citenamefont {Kostritsa}, \citenamefont {Landhuis}, \citenamefont {Lindmark}, \citenamefont {Lucero}, \citenamefont {Lyakh}, \citenamefont {Mandr{\`a}}, \citenamefont {McClean}, \citenamefont {McEwen}, \citenamefont {Megrant}, \citenamefont {Mi}, \citenamefont {Michielsen}, \citenamefont {Mohseni}, \citenamefont {Mutus}, \citenamefont {Naaman}, \citenamefont {Neeley}, \citenamefont {Neill}, \citenamefont {Niu}, \citenamefont {Ostby}, \citenamefont {Petukhov}, \citenamefont {Platt}, \citenamefont {Quintana}, \citenamefont {Rieffel}, \citenamefont {Roushan}, \citenamefont {Rubin}, \citenamefont {Sank}, \citenamefont {Satzinger}, \citenamefont {Smelyanskiy}, \citenamefont {Sung}, \citenamefont {Trevithick}, \citenamefont {Vainsencher}, \citenamefont {Villalonga}, \citenamefont {White}, \citenamefont {Yao}, \citenamefont {Yeh}, \citenamefont {Zalcman}, \citenamefont {Neven},\ and\ \citenamefont {Martinis}}]{Arute2019}%
  \BibitemOpen
  \bibfield  {author} {\bibinfo {author} {\bibfnamefont {F.}~\bibnamefont {Arute}}, \bibinfo {author} {\bibfnamefont {K.}~\bibnamefont {Arya}}, \bibinfo {author} {\bibfnamefont {R.}~\bibnamefont {Babbush}}, \bibinfo {author} {\bibfnamefont {D.}~\bibnamefont {Bacon}}, \bibinfo {author} {\bibfnamefont {J.~C.}\ \bibnamefont {Bardin}}, \bibinfo {author} {\bibfnamefont {R.}~\bibnamefont {Barends}}, \bibinfo {author} {\bibfnamefont {R.}~\bibnamefont {Biswas}}, \bibinfo {author} {\bibfnamefont {S.}~\bibnamefont {Boixo}}, \bibinfo {author} {\bibfnamefont {F.~G. S.~L.}\ \bibnamefont {Brandao}}, \bibinfo {author} {\bibfnamefont {D.~A.}\ \bibnamefont {Buell}}, \bibinfo {author} {\bibfnamefont {B.}~\bibnamefont {Burkett}}, \bibinfo {author} {\bibfnamefont {Y.}~\bibnamefont {Chen}}, \bibinfo {author} {\bibfnamefont {Z.}~\bibnamefont {Chen}}, \bibinfo {author} {\bibfnamefont {B.}~\bibnamefont {Chiaro}}, \bibinfo {author} {\bibfnamefont {R.}~\bibnamefont {Collins}}, \bibinfo {author} {\bibfnamefont {W.}~\bibnamefont {Courtney}}, \bibinfo {author} {\bibfnamefont {A.}~\bibnamefont {Dunsworth}}, \bibinfo {author} {\bibfnamefont {E.}~\bibnamefont {Farhi}}, \bibinfo {author} {\bibfnamefont {B.}~\bibnamefont {Foxen}}, \bibinfo {author} {\bibfnamefont {A.}~\bibnamefont {Fowler}}, \bibinfo {author} {\bibfnamefont {C.}~\bibnamefont {Gidney}}, \bibinfo {author} {\bibfnamefont {M.}~\bibnamefont {Giustina}}, \bibinfo {author} {\bibfnamefont {R.}~\bibnamefont {Graff}}, \bibinfo {author} {\bibfnamefont {K.}~\bibnamefont {Guerin}}, \bibinfo {author} {\bibfnamefont {S.}~\bibnamefont {Habegger}}, \bibinfo {author} {\bibfnamefont {M.~P.}\ \bibnamefont {Harrigan}}, \bibinfo {author} {\bibfnamefont {M.~J.}\ \bibnamefont {Hartmann}}, \bibinfo {author} {\bibfnamefont {A.}~\bibnamefont {Ho}}, \bibinfo {author} {\bibfnamefont {M.}~\bibnamefont {Hoffmann}}, \bibinfo {author} {\bibfnamefont {T.}~\bibnamefont {Huang}}, \bibinfo {author} {\bibfnamefont {T.~S.}\ \bibnamefont {Humble}}, \bibinfo {author} {\bibfnamefont {S.~V.}\ \bibnamefont {Isakov}}, \bibinfo {author} {\bibfnamefont {E.}~\bibnamefont {Jeffrey}}, \bibinfo {author} {\bibfnamefont {Z.}~\bibnamefont {Jiang}}, \bibinfo {author} {\bibfnamefont {D.}~\bibnamefont {Kafri}}, \bibinfo {author} {\bibfnamefont {K.}~\bibnamefont {Kechedzhi}}, \bibinfo {author} {\bibfnamefont {J.}~\bibnamefont {Kelly}}, \bibinfo {author} {\bibfnamefont {P.~V.}\ \bibnamefont {Klimov}}, \bibinfo {author} {\bibfnamefont {S.}~\bibnamefont {Knysh}}, \bibinfo {author} {\bibfnamefont {A.}~\bibnamefont {Korotkov}}, \bibinfo {author} {\bibfnamefont {F.}~\bibnamefont {Kostritsa}}, \bibinfo {author} {\bibfnamefont {D.}~\bibnamefont {Landhuis}}, \bibinfo {author} {\bibfnamefont {M.}~\bibnamefont {Lindmark}}, \bibinfo {author} {\bibfnamefont {E.}~\bibnamefont {Lucero}}, \bibinfo {author} {\bibfnamefont {D.}~\bibnamefont {Lyakh}}, \bibinfo {author} {\bibfnamefont {S.}~\bibnamefont {Mandr{\`a}}}, \bibinfo {author} {\bibfnamefont {J.~R.}\ \bibnamefont {McClean}}, \bibinfo {author} {\bibfnamefont {M.}~\bibnamefont {McEwen}}, \bibinfo {author} {\bibfnamefont {A.}~\bibnamefont {Megrant}}, \bibinfo {author} {\bibfnamefont {X.}~\bibnamefont {Mi}}, \bibinfo {author} {\bibfnamefont {K.}~\bibnamefont {Michielsen}}, \bibinfo {author} {\bibfnamefont {M.}~\bibnamefont {Mohseni}}, \bibinfo {author} {\bibfnamefont {J.}~\bibnamefont {Mutus}}, \bibinfo {author} {\bibfnamefont {O.}~\bibnamefont {Naaman}}, \bibinfo {author} {\bibfnamefont {M.}~\bibnamefont {Neeley}}, \bibinfo {author} {\bibfnamefont {C.}~\bibnamefont {Neill}}, \bibinfo {author} {\bibfnamefont {M.~Y.}\ \bibnamefont {Niu}}, \bibinfo {author} {\bibfnamefont {E.}~\bibnamefont {Ostby}}, \bibinfo {author} {\bibfnamefont {A.}~\bibnamefont {Petukhov}}, \bibinfo {author} {\bibfnamefont {J.~C.}\ \bibnamefont {Platt}}, \bibinfo {author} {\bibfnamefont {C.}~\bibnamefont {Quintana}}, \bibinfo {author} {\bibfnamefont {E.~G.}\ \bibnamefont {Rieffel}}, \bibinfo {author} {\bibfnamefont {P.}~\bibnamefont {Roushan}}, \bibinfo {author} {\bibfnamefont {N.~C.}\ \bibnamefont {Rubin}}, \bibinfo {author} {\bibfnamefont {D.}~\bibnamefont {Sank}}, \bibinfo {author} {\bibfnamefont {K.~J.}\ \bibnamefont {Satzinger}}, \bibinfo {author} {\bibfnamefont {V.}~\bibnamefont {Smelyanskiy}}, \bibinfo {author} {\bibfnamefont {K.~J.}\ \bibnamefont {Sung}}, \bibinfo {author} {\bibfnamefont {M.~D.}\ \bibnamefont {Trevithick}}, \bibinfo {author} {\bibfnamefont {A.}~\bibnamefont {Vainsencher}}, \bibinfo {author} {\bibfnamefont {B.}~\bibnamefont {Villalonga}}, \bibinfo {author} {\bibfnamefont {T.}~\bibnamefont {White}}, \bibinfo {author} {\bibfnamefont {Z.~J.}\ \bibnamefont {Yao}}, \bibinfo {author} {\bibfnamefont {P.}~\bibnamefont {Yeh}}, \bibinfo {author} {\bibfnamefont {A.}~\bibnamefont {Zalcman}}, \bibinfo {author} {\bibfnamefont {H.}~\bibnamefont {Neven}},\ and\ \bibinfo {author} {\bibfnamefont {J.~M.}\ \bibnamefont {Martinis}},\ }\bibfield  {title} {\bibinfo {title} {Quantum supremacy using a programmable superconducting processor},\ }\href {https://doi.org/10.1038/s41586-019-1666-5} {\bibfield  {journal} {\bibinfo  {journal} {Nature}\ }\textbf {\bibinfo {volume} {574}},\ \bibinfo {pages} {505} (\bibinfo {year} {2019})}\BibitemShut {NoStop}%
\bibitem [{\citenamefont {Wu}\ \emph {et~al.}(2021)\citenamefont {Wu}, \citenamefont {Bao}, \citenamefont {Cao}, \citenamefont {Chen}, \citenamefont {Chen}, \citenamefont {Chen}, \citenamefont {Chung}, \citenamefont {Deng}, \citenamefont {Du}, \citenamefont {Fan}, \citenamefont {Gong}, \citenamefont {Guo}, \citenamefont {Guo}, \citenamefont {Guo}, \citenamefont {Han}, \citenamefont {Hong}, \citenamefont {Huang}, \citenamefont {Huo}, \citenamefont {Li}, \citenamefont {Li}, \citenamefont {Li}, \citenamefont {Li}, \citenamefont {Liang}, \citenamefont {Lin}, \citenamefont {Lin}, \citenamefont {Qian}, \citenamefont {Qiao}, \citenamefont {Rong}, \citenamefont {Su}, \citenamefont {Sun}, \citenamefont {Wang}, \citenamefont {Wang}, \citenamefont {Wu}, \citenamefont {Xu}, \citenamefont {Yan}, \citenamefont {Yang}, \citenamefont {Yang}, \citenamefont {Ye}, \citenamefont {Yin}, \citenamefont {Ying}, \citenamefont {Yu}, \citenamefont {Zha}, \citenamefont {Zhang}, \citenamefont {Zhang}, \citenamefont {Zhang}, \citenamefont {Zhang}, \citenamefont {Zhao}, \citenamefont {Zhao}, \citenamefont {Zhou}, \citenamefont {Zhu}, \citenamefont {Lu}, \citenamefont {Peng}, \citenamefont {Zhu},\ and\ \citenamefont {Pan}}]{PhysRevLett.127.180501}%
  \BibitemOpen
  \bibfield  {author} {\bibinfo {author} {\bibfnamefont {Y.}~\bibnamefont {Wu}}, \bibinfo {author} {\bibfnamefont {W.-S.}\ \bibnamefont {Bao}}, \bibinfo {author} {\bibfnamefont {S.}~\bibnamefont {Cao}}, \bibinfo {author} {\bibfnamefont {F.}~\bibnamefont {Chen}}, \bibinfo {author} {\bibfnamefont {M.-C.}\ \bibnamefont {Chen}}, \bibinfo {author} {\bibfnamefont {X.}~\bibnamefont {Chen}}, \bibinfo {author} {\bibfnamefont {T.-H.}\ \bibnamefont {Chung}}, \bibinfo {author} {\bibfnamefont {H.}~\bibnamefont {Deng}}, \bibinfo {author} {\bibfnamefont {Y.}~\bibnamefont {Du}}, \bibinfo {author} {\bibfnamefont {D.}~\bibnamefont {Fan}}, \bibinfo {author} {\bibfnamefont {M.}~\bibnamefont {Gong}}, \bibinfo {author} {\bibfnamefont {C.}~\bibnamefont {Guo}}, \bibinfo {author} {\bibfnamefont {C.}~\bibnamefont {Guo}}, \bibinfo {author} {\bibfnamefont {S.}~\bibnamefont {Guo}}, \bibinfo {author} {\bibfnamefont {L.}~\bibnamefont {Han}}, \bibinfo {author} {\bibfnamefont {L.}~\bibnamefont {Hong}}, \bibinfo {author} {\bibfnamefont {H.-L.}\ \bibnamefont {Huang}}, \bibinfo {author} {\bibfnamefont {Y.-H.}\ \bibnamefont {Huo}}, \bibinfo {author} {\bibfnamefont {L.}~\bibnamefont {Li}}, \bibinfo {author} {\bibfnamefont {N.}~\bibnamefont {Li}}, \bibinfo {author} {\bibfnamefont {S.}~\bibnamefont {Li}}, \bibinfo {author} {\bibfnamefont {Y.}~\bibnamefont {Li}}, \bibinfo {author} {\bibfnamefont {F.}~\bibnamefont {Liang}}, \bibinfo {author} {\bibfnamefont {C.}~\bibnamefont {Lin}}, \bibinfo {author} {\bibfnamefont {J.}~\bibnamefont {Lin}}, \bibinfo {author} {\bibfnamefont {H.}~\bibnamefont {Qian}}, \bibinfo {author} {\bibfnamefont {D.}~\bibnamefont {Qiao}}, \bibinfo {author} {\bibfnamefont {H.}~\bibnamefont {Rong}}, \bibinfo {author} {\bibfnamefont {H.}~\bibnamefont {Su}}, \bibinfo {author} {\bibfnamefont {L.}~\bibnamefont {Sun}}, \bibinfo {author} {\bibfnamefont {L.}~\bibnamefont {Wang}}, \bibinfo {author} {\bibfnamefont {S.}~\bibnamefont {Wang}}, \bibinfo {author} {\bibfnamefont {D.}~\bibnamefont {Wu}}, \bibinfo {author} {\bibfnamefont {Y.}~\bibnamefont {Xu}}, \bibinfo {author} {\bibfnamefont {K.}~\bibnamefont {Yan}}, \bibinfo {author} {\bibfnamefont {W.}~\bibnamefont {Yang}}, \bibinfo {author} {\bibfnamefont {Y.}~\bibnamefont {Yang}}, \bibinfo {author} {\bibfnamefont {Y.}~\bibnamefont {Ye}}, \bibinfo {author} {\bibfnamefont {J.}~\bibnamefont {Yin}}, \bibinfo {author} {\bibfnamefont {C.}~\bibnamefont {Ying}}, \bibinfo {author} {\bibfnamefont {J.}~\bibnamefont {Yu}}, \bibinfo {author} {\bibfnamefont {C.}~\bibnamefont {Zha}}, \bibinfo {author} {\bibfnamefont {C.}~\bibnamefont {Zhang}}, \bibinfo {author} {\bibfnamefont {H.}~\bibnamefont {Zhang}}, \bibinfo {author} {\bibfnamefont {K.}~\bibnamefont {Zhang}}, \bibinfo {author} {\bibfnamefont {Y.}~\bibnamefont {Zhang}}, \bibinfo {author} {\bibfnamefont {H.}~\bibnamefont {Zhao}}, \bibinfo {author} {\bibfnamefont {Y.}~\bibnamefont {Zhao}}, \bibinfo {author} {\bibfnamefont {L.}~\bibnamefont {Zhou}}, \bibinfo {author} {\bibfnamefont {Q.}~\bibnamefont {Zhu}}, \bibinfo {author} {\bibfnamefont {C.-Y.}\ \bibnamefont {Lu}}, \bibinfo {author} {\bibfnamefont {C.-Z.}\ \bibnamefont {Peng}}, \bibinfo {author} {\bibfnamefont {X.}~\bibnamefont {Zhu}},\ and\ \bibinfo {author} {\bibfnamefont {J.-W.}\ \bibnamefont {Pan}},\ }\bibfield  {title} {\bibinfo {title} {Strong quantum computational advantage using a superconducting quantum processor},\ }\href {https://doi.org/10.1103/PhysRevLett.127.180501} {\bibfield  {journal} {\bibinfo  {journal} {Phys. Rev. Lett.}\ }\textbf {\bibinfo {volume} {127}},\ \bibinfo {pages} {180501} (\bibinfo {year} {2021})}\BibitemShut {NoStop}%
\bibitem [{\citenamefont {Choi}\ \emph {et~al.}(2023)\citenamefont {Choi}, \citenamefont {Shaw}, \citenamefont {Madjarov}, \citenamefont {Xie}, \citenamefont {Finkelstein}, \citenamefont {Covey}, \citenamefont {Cotler}, \citenamefont {Mark}, \citenamefont {Huang}, \citenamefont {Kale}, \citenamefont {Pichler}, \citenamefont {Brand{\~a}o}, \citenamefont {Choi},\ and\ \citenamefont {Endres}}]{Choi2023}%
  \BibitemOpen
  \bibfield  {author} {\bibinfo {author} {\bibfnamefont {J.}~\bibnamefont {Choi}}, \bibinfo {author} {\bibfnamefont {A.~L.}\ \bibnamefont {Shaw}}, \bibinfo {author} {\bibfnamefont {I.~S.}\ \bibnamefont {Madjarov}}, \bibinfo {author} {\bibfnamefont {X.}~\bibnamefont {Xie}}, \bibinfo {author} {\bibfnamefont {R.}~\bibnamefont {Finkelstein}}, \bibinfo {author} {\bibfnamefont {J.~P.}\ \bibnamefont {Covey}}, \bibinfo {author} {\bibfnamefont {J.~S.}\ \bibnamefont {Cotler}}, \bibinfo {author} {\bibfnamefont {D.~K.}\ \bibnamefont {Mark}}, \bibinfo {author} {\bibfnamefont {H.-Y.}\ \bibnamefont {Huang}}, \bibinfo {author} {\bibfnamefont {A.}~\bibnamefont {Kale}}, \bibinfo {author} {\bibfnamefont {H.}~\bibnamefont {Pichler}}, \bibinfo {author} {\bibfnamefont {F.~G. S.~L.}\ \bibnamefont {Brand{\~a}o}}, \bibinfo {author} {\bibfnamefont {S.}~\bibnamefont {Choi}},\ and\ \bibinfo {author} {\bibfnamefont {M.}~\bibnamefont {Endres}},\ }\bibfield  {title} {\bibinfo {title} {Preparing random states and benchmarking with many-body quantum chaos},\ }\href {https://doi.org/10.1038/s41586-022-05442-1} {\bibfield  {journal} {\bibinfo  {journal} {Nature}\ }\textbf {\bibinfo {volume} {613}},\ \bibinfo {pages} {468} (\bibinfo {year} {2023})}\BibitemShut {NoStop}%
\bibitem [{\citenamefont {Cotler}\ \emph {et~al.}(2023)\citenamefont {Cotler}, \citenamefont {Mark}, \citenamefont {Huang}, \citenamefont {Hern\'andez}, \citenamefont {Choi}, \citenamefont {Shaw}, \citenamefont {Endres},\ and\ \citenamefont {Choi}}]{PRXQuantum.4.010311}%
  \BibitemOpen
  \bibfield  {author} {\bibinfo {author} {\bibfnamefont {J.~S.}\ \bibnamefont {Cotler}}, \bibinfo {author} {\bibfnamefont {D.~K.}\ \bibnamefont {Mark}}, \bibinfo {author} {\bibfnamefont {H.-Y.}\ \bibnamefont {Huang}}, \bibinfo {author} {\bibfnamefont {F.}~\bibnamefont {Hern\'andez}}, \bibinfo {author} {\bibfnamefont {J.}~\bibnamefont {Choi}}, \bibinfo {author} {\bibfnamefont {A.~L.}\ \bibnamefont {Shaw}}, \bibinfo {author} {\bibfnamefont {M.}~\bibnamefont {Endres}},\ and\ \bibinfo {author} {\bibfnamefont {S.}~\bibnamefont {Choi}},\ }\bibfield  {title} {\bibinfo {title} {Emergent quantum state designs from individual many-body wave functions},\ }\href {https://doi.org/10.1103/PRXQuantum.4.010311} {\bibfield  {journal} {\bibinfo  {journal} {PRX Quantum}\ }\textbf {\bibinfo {volume} {4}},\ \bibinfo {pages} {010311} (\bibinfo {year} {2023})}\BibitemShut {NoStop}%
\bibitem [{\citenamefont {Ho}\ and\ \citenamefont {Choi}(2022)}]{PhysRevLett.128.060601}%
  \BibitemOpen
  \bibfield  {author} {\bibinfo {author} {\bibfnamefont {W.~W.}\ \bibnamefont {Ho}}\ and\ \bibinfo {author} {\bibfnamefont {S.}~\bibnamefont {Choi}},\ }\bibfield  {title} {\bibinfo {title} {Exact emergent quantum state designs from quantum chaotic dynamics},\ }\href {https://doi.org/10.1103/PhysRevLett.128.060601} {\bibfield  {journal} {\bibinfo  {journal} {Phys. Rev. Lett.}\ }\textbf {\bibinfo {volume} {128}},\ \bibinfo {pages} {060601} (\bibinfo {year} {2022})}\BibitemShut {NoStop}%
\bibitem [{\citenamefont {Ippoliti}\ and\ \citenamefont {Ho}(2022)}]{Ippoliti2022solvablemodelofdeep}%
  \BibitemOpen
  \bibfield  {author} {\bibinfo {author} {\bibfnamefont {M.}~\bibnamefont {Ippoliti}}\ and\ \bibinfo {author} {\bibfnamefont {W.~W.}\ \bibnamefont {Ho}},\ }\bibfield  {title} {\bibinfo {title} {Solvable model of deep thermalization with distinct design times},\ }\href {https://doi.org/10.22331/q-2022-12-29-886} {\bibfield  {journal} {\bibinfo  {journal} {{Quantum}}\ }\textbf {\bibinfo {volume} {6}},\ \bibinfo {pages} {886} (\bibinfo {year} {2022})}\BibitemShut {NoStop}%
\bibitem [{\citenamefont {Skinner}\ \emph {et~al.}(2019)\citenamefont {Skinner}, \citenamefont {Ruhman},\ and\ \citenamefont {Nahum}}]{PhysRevX.9.031009}%
  \BibitemOpen
  \bibfield  {author} {\bibinfo {author} {\bibfnamefont {B.}~\bibnamefont {Skinner}}, \bibinfo {author} {\bibfnamefont {J.}~\bibnamefont {Ruhman}},\ and\ \bibinfo {author} {\bibfnamefont {A.}~\bibnamefont {Nahum}},\ }\bibfield  {title} {\bibinfo {title} {Measurement-induced phase transitions in the dynamics of entanglement},\ }\href {https://doi.org/10.1103/PhysRevX.9.031009} {\bibfield  {journal} {\bibinfo  {journal} {Phys. Rev. X}\ }\textbf {\bibinfo {volume} {9}},\ \bibinfo {pages} {031009} (\bibinfo {year} {2019})}\BibitemShut {NoStop}%
\bibitem [{\citenamefont {Chan}\ \emph {et~al.}(2019)\citenamefont {Chan}, \citenamefont {Nandkishore}, \citenamefont {Pretko},\ and\ \citenamefont {Smith}}]{PhysRevB.99.224307}%
  \BibitemOpen
  \bibfield  {author} {\bibinfo {author} {\bibfnamefont {A.}~\bibnamefont {Chan}}, \bibinfo {author} {\bibfnamefont {R.~M.}\ \bibnamefont {Nandkishore}}, \bibinfo {author} {\bibfnamefont {M.}~\bibnamefont {Pretko}},\ and\ \bibinfo {author} {\bibfnamefont {G.}~\bibnamefont {Smith}},\ }\bibfield  {title} {\bibinfo {title} {Unitary-projective entanglement dynamics},\ }\href {https://doi.org/10.1103/PhysRevB.99.224307} {\bibfield  {journal} {\bibinfo  {journal} {Phys. Rev. B}\ }\textbf {\bibinfo {volume} {99}},\ \bibinfo {pages} {224307} (\bibinfo {year} {2019})}\BibitemShut {NoStop}%
\bibitem [{\citenamefont {Choi}\ \emph {et~al.}(2020)\citenamefont {Choi}, \citenamefont {Bao}, \citenamefont {Qi},\ and\ \citenamefont {Altman}}]{PhysRevLett.125.030505}%
  \BibitemOpen
  \bibfield  {author} {\bibinfo {author} {\bibfnamefont {S.}~\bibnamefont {Choi}}, \bibinfo {author} {\bibfnamefont {Y.}~\bibnamefont {Bao}}, \bibinfo {author} {\bibfnamefont {X.-L.}\ \bibnamefont {Qi}},\ and\ \bibinfo {author} {\bibfnamefont {E.}~\bibnamefont {Altman}},\ }\bibfield  {title} {\bibinfo {title} {Quantum error correction in scrambling dynamics and measurement-induced phase transition},\ }\href {https://doi.org/10.1103/PhysRevLett.125.030505} {\bibfield  {journal} {\bibinfo  {journal} {Phys. Rev. Lett.}\ }\textbf {\bibinfo {volume} {125}},\ \bibinfo {pages} {030505} (\bibinfo {year} {2020})}\BibitemShut {NoStop}%
\bibitem [{\citenamefont {Gullans}\ and\ \citenamefont {Huse}(2020{\natexlab{a}})}]{PhysRevX.10.041020}%
  \BibitemOpen
  \bibfield  {author} {\bibinfo {author} {\bibfnamefont {M.~J.}\ \bibnamefont {Gullans}}\ and\ \bibinfo {author} {\bibfnamefont {D.~A.}\ \bibnamefont {Huse}},\ }\bibfield  {title} {\bibinfo {title} {Dynamical purification phase transition induced by quantum measurements},\ }\href {https://doi.org/10.1103/PhysRevX.10.041020} {\bibfield  {journal} {\bibinfo  {journal} {Phys. Rev. X}\ }\textbf {\bibinfo {volume} {10}},\ \bibinfo {pages} {041020} (\bibinfo {year} {2020}{\natexlab{a}})}\BibitemShut {NoStop}%
\bibitem [{\citenamefont {Jian}\ \emph {et~al.}(2020)\citenamefont {Jian}, \citenamefont {You}, \citenamefont {Vasseur},\ and\ \citenamefont {Ludwig}}]{PhysRevB.101.104302}%
  \BibitemOpen
  \bibfield  {author} {\bibinfo {author} {\bibfnamefont {C.-M.}\ \bibnamefont {Jian}}, \bibinfo {author} {\bibfnamefont {Y.-Z.}\ \bibnamefont {You}}, \bibinfo {author} {\bibfnamefont {R.}~\bibnamefont {Vasseur}},\ and\ \bibinfo {author} {\bibfnamefont {A.~W.~W.}\ \bibnamefont {Ludwig}},\ }\bibfield  {title} {\bibinfo {title} {Measurement-induced criticality in random quantum circuits},\ }\href {https://doi.org/10.1103/PhysRevB.101.104302} {\bibfield  {journal} {\bibinfo  {journal} {Phys. Rev. B}\ }\textbf {\bibinfo {volume} {101}},\ \bibinfo {pages} {104302} (\bibinfo {year} {2020})}\BibitemShut {NoStop}%
\bibitem [{\citenamefont {Bao}\ \emph {et~al.}(2020)\citenamefont {Bao}, \citenamefont {Choi},\ and\ \citenamefont {Altman}}]{PhysRevB.101.104301}%
  \BibitemOpen
  \bibfield  {author} {\bibinfo {author} {\bibfnamefont {Y.}~\bibnamefont {Bao}}, \bibinfo {author} {\bibfnamefont {S.}~\bibnamefont {Choi}},\ and\ \bibinfo {author} {\bibfnamefont {E.}~\bibnamefont {Altman}},\ }\bibfield  {title} {\bibinfo {title} {Theory of the phase transition in random unitary circuits with measurements},\ }\href {https://doi.org/10.1103/PhysRevB.101.104301} {\bibfield  {journal} {\bibinfo  {journal} {Phys. Rev. B}\ }\textbf {\bibinfo {volume} {101}},\ \bibinfo {pages} {104301} (\bibinfo {year} {2020})}\BibitemShut {NoStop}%
\bibitem [{\citenamefont {Gullans}\ and\ \citenamefont {Huse}(2020{\natexlab{b}})}]{PhysRevLett.125.070606}%
  \BibitemOpen
  \bibfield  {author} {\bibinfo {author} {\bibfnamefont {M.~J.}\ \bibnamefont {Gullans}}\ and\ \bibinfo {author} {\bibfnamefont {D.~A.}\ \bibnamefont {Huse}},\ }\bibfield  {title} {\bibinfo {title} {Scalable probes of measurement-induced criticality},\ }\href {https://doi.org/10.1103/PhysRevLett.125.070606} {\bibfield  {journal} {\bibinfo  {journal} {Phys. Rev. Lett.}\ }\textbf {\bibinfo {volume} {125}},\ \bibinfo {pages} {070606} (\bibinfo {year} {2020}{\natexlab{b}})}\BibitemShut {NoStop}%
\bibitem [{\citenamefont {Li}\ \emph {et~al.}(2023)\citenamefont {Li}, \citenamefont {Zou}, \citenamefont {Glorioso}, \citenamefont {Altman},\ and\ \citenamefont {Fisher}}]{PhysRevLett.130.220404}%
  \BibitemOpen
  \bibfield  {author} {\bibinfo {author} {\bibfnamefont {Y.}~\bibnamefont {Li}}, \bibinfo {author} {\bibfnamefont {Y.}~\bibnamefont {Zou}}, \bibinfo {author} {\bibfnamefont {P.}~\bibnamefont {Glorioso}}, \bibinfo {author} {\bibfnamefont {E.}~\bibnamefont {Altman}},\ and\ \bibinfo {author} {\bibfnamefont {M.~P.~A.}\ \bibnamefont {Fisher}},\ }\bibfield  {title} {\bibinfo {title} {Cross entropy benchmark for measurement-induced phase transitions},\ }\href {https://doi.org/10.1103/PhysRevLett.130.220404} {\bibfield  {journal} {\bibinfo  {journal} {Phys. Rev. Lett.}\ }\textbf {\bibinfo {volume} {130}},\ \bibinfo {pages} {220404} (\bibinfo {year} {2023})}\BibitemShut {NoStop}%
\bibitem [{\citenamefont {Garratt}\ and\ \citenamefont {Altman}(2024)}]{PRXQuantum.5.030311}%
  \BibitemOpen
  \bibfield  {author} {\bibinfo {author} {\bibfnamefont {S.~J.}\ \bibnamefont {Garratt}}\ and\ \bibinfo {author} {\bibfnamefont {E.}~\bibnamefont {Altman}},\ }\bibfield  {title} {\bibinfo {title} {Probing postmeasurement entanglement without postselection},\ }\href {https://doi.org/10.1103/PRXQuantum.5.030311} {\bibfield  {journal} {\bibinfo  {journal} {PRX Quantum}\ }\textbf {\bibinfo {volume} {5}},\ \bibinfo {pages} {030311} (\bibinfo {year} {2024})}\BibitemShut {NoStop}%
\bibitem [{\citenamefont {Noel}\ \emph {et~al.}(2022)\citenamefont {Noel}, \citenamefont {Niroula}, \citenamefont {Zhu}, \citenamefont {Risinger}, \citenamefont {Egan}, \citenamefont {Biswas}, \citenamefont {Cetina}, \citenamefont {Gorshkov}, \citenamefont {Gullans}, \citenamefont {Huse},\ and\ \citenamefont {Monroe}}]{Noel2022}%
  \BibitemOpen
  \bibfield  {author} {\bibinfo {author} {\bibfnamefont {C.}~\bibnamefont {Noel}}, \bibinfo {author} {\bibfnamefont {P.}~\bibnamefont {Niroula}}, \bibinfo {author} {\bibfnamefont {D.}~\bibnamefont {Zhu}}, \bibinfo {author} {\bibfnamefont {A.}~\bibnamefont {Risinger}}, \bibinfo {author} {\bibfnamefont {L.}~\bibnamefont {Egan}}, \bibinfo {author} {\bibfnamefont {D.}~\bibnamefont {Biswas}}, \bibinfo {author} {\bibfnamefont {M.}~\bibnamefont {Cetina}}, \bibinfo {author} {\bibfnamefont {A.~V.}\ \bibnamefont {Gorshkov}}, \bibinfo {author} {\bibfnamefont {M.~J.}\ \bibnamefont {Gullans}}, \bibinfo {author} {\bibfnamefont {D.~A.}\ \bibnamefont {Huse}},\ and\ \bibinfo {author} {\bibfnamefont {C.}~\bibnamefont {Monroe}},\ }\bibfield  {title} {\bibinfo {title} {Measurement-induced quantum phases realized in a trapped-ion quantum computer},\ }\href {https://doi.org/10.1038/s41567-022-01619-7} {\bibfield  {journal} {\bibinfo  {journal} {Nature Physics}\ }\textbf {\bibinfo {volume} {18}},\ \bibinfo {pages} {760} (\bibinfo {year} {2022})}\BibitemShut {NoStop}%
\bibitem [{\citenamefont {Hoke}\ \emph {et~al.}(2023)\citenamefont {Hoke}, \citenamefont {Ippoliti}, \citenamefont {Rosenberg}, \citenamefont {Abanin}, \citenamefont {Acharya}, \citenamefont {Andersen}, \citenamefont {Ansmann}, \citenamefont {Arute}, \citenamefont {Arya}, \citenamefont {Asfaw}, \citenamefont {Atalaya}, \citenamefont {Bardin}, \citenamefont {Bengtsson}, \citenamefont {Bortoli}, \citenamefont {Bourassa}, \citenamefont {Bovaird}, \citenamefont {Brill}, \citenamefont {Broughton}, \citenamefont {Buckley}, \citenamefont {Buell}, \citenamefont {Burger}, \citenamefont {Burkett}, \citenamefont {Bushnell}, \citenamefont {Chen}, \citenamefont {Chiaro}, \citenamefont {Chik}, \citenamefont {Cogan}, \citenamefont {Collins}, \citenamefont {Conner}, \citenamefont {Courtney}, \citenamefont {Crook}, \citenamefont {Curtin}, \citenamefont {Dau}, \citenamefont {Debroy}, \citenamefont {Del Toro~Barba}, \citenamefont {Demura}, \citenamefont {Di~Paolo}, \citenamefont {Drozdov}, \citenamefont {Dunsworth}, \citenamefont {Eppens}, \citenamefont {Erickson}, \citenamefont {Farhi}, \citenamefont {Fatemi}, \citenamefont {Ferreira}, \citenamefont {Burgos}, \citenamefont {Forati}, \citenamefont {Fowler}, \citenamefont {Foxen}, \citenamefont {Giang}, \citenamefont {Gidney}, \citenamefont {Gilboa}, \citenamefont {Giustina}, \citenamefont {Gosula}, \citenamefont {Gross}, \citenamefont {Habegger}, \citenamefont {Hamilton}, \citenamefont {Hansen}, \citenamefont {Harrigan}, \citenamefont {Harrington}, \citenamefont {Heu}, \citenamefont {Hoffmann}, \citenamefont {Hong}, \citenamefont {Huang}, \citenamefont {Huff}, \citenamefont {Huggins}, \citenamefont {Isakov}, \citenamefont {Iveland}, \citenamefont {Jeffrey}, \citenamefont {Jiang}, \citenamefont {Jones}, \citenamefont {Juhas}, \citenamefont {Kafri}, \citenamefont {Kechedzhi}, \citenamefont {Khattar}, \citenamefont {Khezri}, \citenamefont {Kieferov{\'a}}, \citenamefont {Kim}, \citenamefont {Kitaev}, \citenamefont {Klimov}, \citenamefont {Klots}, \citenamefont {Korotkov}, \citenamefont {Kostritsa}, \citenamefont {Kreikebaum}, \citenamefont {Landhuis}, \citenamefont {Laptev}, \citenamefont {Lau}, \citenamefont {Laws}, \citenamefont {Lee}, \citenamefont {Lee}, \citenamefont {Lensky}, \citenamefont {Lester}, \citenamefont {Lill}, \citenamefont {Liu}, \citenamefont {Locharla}, \citenamefont {Martin}, \citenamefont {McClean}, \citenamefont {McEwen}, \citenamefont {Miao}, \citenamefont {Mieszala}, \citenamefont {Montazeri}, \citenamefont {Morvan}, \citenamefont {Movassagh}, \citenamefont {Mruczkiewicz}, \citenamefont {Neeley}, \citenamefont {Neill}, \citenamefont {Nersisyan}, \citenamefont {Newman}, \citenamefont {Ng}, \citenamefont {Nguyen}, \citenamefont {Nguyen}, \citenamefont {Niu}, \citenamefont {O'Brien}, \citenamefont {Omonije}, \citenamefont {Opremcak}, \citenamefont {Petukhov}, \citenamefont {Potter}, \citenamefont {Pryadko}, \citenamefont {Quintana}, \citenamefont {Rocque}, \citenamefont {Rubin}, \citenamefont {Saei}, \citenamefont {Sank}, \citenamefont {Sankaragomathi}, \citenamefont {Satzinger}, \citenamefont {Schurkus}, \citenamefont {Schuster}, \citenamefont {Shearn}, \citenamefont {Shorter}, \citenamefont {Shutty}, \citenamefont {Shvarts}, \citenamefont {Skruzny}, \citenamefont {Smith}, \citenamefont {Somma}, \citenamefont {Sterling}, \citenamefont {Strain}, \citenamefont {Szalay}, \citenamefont {Torres}, \citenamefont {Vidal}, \citenamefont {Villalonga}, \citenamefont {Heidweiller}, \citenamefont {White}, \citenamefont {Woo}, \citenamefont {Xing}, \citenamefont {Yao}, \citenamefont {Yeh}, \citenamefont {Yoo}, \citenamefont {Young}, \citenamefont {Zalcman}, \citenamefont {Zhang}, \citenamefont {Zhu}, \citenamefont {Zobrist}, \citenamefont {Neven}, \citenamefont {Babbush}, \citenamefont {Bacon}, \citenamefont {Boixo}, \citenamefont {Hilton}, \citenamefont {Lucero}, \citenamefont {Megrant}, \citenamefont {Kelly}, \citenamefont {Chen}, \citenamefont {Smelyanskiy}, \citenamefont {Mi}, \citenamefont {Khemani}, \citenamefont {Roushan}, \citenamefont {AI},\ and\ \citenamefont {{Collaborators}}}]{Hoke2023}%
  \BibitemOpen
  \bibfield  {author} {\bibinfo {author} {\bibfnamefont {J.~C.}\ \bibnamefont {Hoke}}, \bibinfo {author} {\bibfnamefont {M.}~\bibnamefont {Ippoliti}}, \bibinfo {author} {\bibfnamefont {E.}~\bibnamefont {Rosenberg}}, \bibinfo {author} {\bibfnamefont {D.}~\bibnamefont {Abanin}}, \bibinfo {author} {\bibfnamefont {R.}~\bibnamefont {Acharya}}, \bibinfo {author} {\bibfnamefont {T.~I.}\ \bibnamefont {Andersen}}, \bibinfo {author} {\bibfnamefont {M.}~\bibnamefont {Ansmann}}, \bibinfo {author} {\bibfnamefont {F.}~\bibnamefont {Arute}}, \bibinfo {author} {\bibfnamefont {K.}~\bibnamefont {Arya}}, \bibinfo {author} {\bibfnamefont {A.}~\bibnamefont {Asfaw}}, \bibinfo {author} {\bibfnamefont {J.}~\bibnamefont {Atalaya}}, \bibinfo {author} {\bibfnamefont {J.~C.}\ \bibnamefont {Bardin}}, \bibinfo {author} {\bibfnamefont {A.}~\bibnamefont {Bengtsson}}, \bibinfo {author} {\bibfnamefont {G.}~\bibnamefont {Bortoli}}, \bibinfo {author} {\bibfnamefont {A.}~\bibnamefont {Bourassa}}, \bibinfo {author} {\bibfnamefont {J.}~\bibnamefont {Bovaird}}, \bibinfo {author} {\bibfnamefont {L.}~\bibnamefont {Brill}}, \bibinfo {author} {\bibfnamefont {M.}~\bibnamefont {Broughton}}, \bibinfo {author} {\bibfnamefont {B.~B.}\ \bibnamefont {Buckley}}, \bibinfo {author} {\bibfnamefont {D.~A.}\ \bibnamefont {Buell}}, \bibinfo {author} {\bibfnamefont {T.}~\bibnamefont {Burger}}, \bibinfo {author} {\bibfnamefont {B.}~\bibnamefont {Burkett}}, \bibinfo {author} {\bibfnamefont {N.}~\bibnamefont {Bushnell}}, \bibinfo {author} {\bibfnamefont {Z.}~\bibnamefont {Chen}}, \bibinfo {author} {\bibfnamefont {B.}~\bibnamefont {Chiaro}}, \bibinfo {author} {\bibfnamefont {D.}~\bibnamefont {Chik}}, \bibinfo {author} {\bibfnamefont {J.}~\bibnamefont {Cogan}}, \bibinfo {author} {\bibfnamefont {R.}~\bibnamefont {Collins}}, \bibinfo {author} {\bibfnamefont {P.}~\bibnamefont {Conner}}, \bibinfo {author} {\bibfnamefont {W.}~\bibnamefont {Courtney}}, \bibinfo {author} {\bibfnamefont {A.~L.}\ \bibnamefont {Crook}}, \bibinfo {author} {\bibfnamefont {B.}~\bibnamefont {Curtin}}, \bibinfo {author} {\bibfnamefont {A.~G.}\ \bibnamefont {Dau}}, \bibinfo {author} {\bibfnamefont {D.~M.}\ \bibnamefont {Debroy}}, \bibinfo {author} {\bibfnamefont {A.}~\bibnamefont {Del Toro~Barba}}, \bibinfo {author} {\bibfnamefont {S.}~\bibnamefont {Demura}}, \bibinfo {author} {\bibfnamefont {A.}~\bibnamefont {Di~Paolo}}, \bibinfo {author} {\bibfnamefont {I.~K.}\ \bibnamefont {Drozdov}}, \bibinfo {author} {\bibfnamefont {A.}~\bibnamefont {Dunsworth}}, \bibinfo {author} {\bibfnamefont {D.}~\bibnamefont {Eppens}}, \bibinfo {author} {\bibfnamefont {C.}~\bibnamefont {Erickson}}, \bibinfo {author} {\bibfnamefont {E.}~\bibnamefont {Farhi}}, \bibinfo {author} {\bibfnamefont {R.}~\bibnamefont {Fatemi}}, \bibinfo {author} {\bibfnamefont {V.~S.}\ \bibnamefont {Ferreira}}, \bibinfo {author} {\bibfnamefont {L.~F.}\ \bibnamefont {Burgos}}, \bibinfo {author} {\bibfnamefont {E.}~\bibnamefont {Forati}}, \bibinfo {author} {\bibfnamefont {A.~G.}\ \bibnamefont {Fowler}}, \bibinfo {author} {\bibfnamefont {B.}~\bibnamefont {Foxen}}, \bibinfo {author} {\bibfnamefont {W.}~\bibnamefont {Giang}}, \bibinfo {author} {\bibfnamefont {C.}~\bibnamefont {Gidney}}, \bibinfo {author} {\bibfnamefont {D.}~\bibnamefont {Gilboa}}, \bibinfo {author} {\bibfnamefont {M.}~\bibnamefont {Giustina}}, \bibinfo {author} {\bibfnamefont {R.}~\bibnamefont {Gosula}}, \bibinfo {author} {\bibfnamefont {J.~A.}\ \bibnamefont {Gross}}, \bibinfo {author} {\bibfnamefont {S.}~\bibnamefont {Habegger}}, \bibinfo {author} {\bibfnamefont {M.~C.}\ \bibnamefont {Hamilton}}, \bibinfo {author} {\bibfnamefont {M.}~\bibnamefont {Hansen}}, \bibinfo {author} {\bibfnamefont {M.~P.}\ \bibnamefont {Harrigan}}, \bibinfo {author} {\bibfnamefont {S.~D.}\ \bibnamefont {Harrington}}, \bibinfo {author} {\bibfnamefont {P.}~\bibnamefont {Heu}}, \bibinfo {author} {\bibfnamefont {M.~R.}\ \bibnamefont {Hoffmann}}, \bibinfo {author} {\bibfnamefont {S.}~\bibnamefont {Hong}}, \bibinfo {author} {\bibfnamefont {T.}~\bibnamefont {Huang}}, \bibinfo {author} {\bibfnamefont {A.}~\bibnamefont {Huff}}, \bibinfo {author} {\bibfnamefont {W.~J.}\ \bibnamefont {Huggins}}, \bibinfo {author} {\bibfnamefont {S.~V.}\ \bibnamefont {Isakov}}, \bibinfo {author} {\bibfnamefont {J.}~\bibnamefont {Iveland}}, \bibinfo {author} {\bibfnamefont {E.}~\bibnamefont {Jeffrey}}, \bibinfo {author} {\bibfnamefont {Z.}~\bibnamefont {Jiang}}, \bibinfo {author} {\bibfnamefont {C.}~\bibnamefont {Jones}}, \bibinfo {author} {\bibfnamefont {P.}~\bibnamefont {Juhas}}, \bibinfo {author} {\bibfnamefont {D.}~\bibnamefont {Kafri}}, \bibinfo {author} {\bibfnamefont {K.}~\bibnamefont {Kechedzhi}}, \bibinfo {author} {\bibfnamefont {T.}~\bibnamefont {Khattar}}, \bibinfo {author} {\bibfnamefont {M.}~\bibnamefont {Khezri}}, \bibinfo {author} {\bibfnamefont {M.}~\bibnamefont {Kieferov{\'a}}}, \bibinfo {author} {\bibfnamefont {S.}~\bibnamefont {Kim}}, \bibinfo {author} {\bibfnamefont {A.}~\bibnamefont {Kitaev}}, \bibinfo {author} {\bibfnamefont {P.~V.}\ \bibnamefont {Klimov}}, \bibinfo {author} {\bibfnamefont {A.~R.}\ \bibnamefont {Klots}}, \bibinfo {author} {\bibfnamefont {A.~N.}\ \bibnamefont {Korotkov}}, \bibinfo {author} {\bibfnamefont {F.}~\bibnamefont {Kostritsa}}, \bibinfo {author} {\bibfnamefont {J.~M.}\ \bibnamefont {Kreikebaum}}, \bibinfo {author} {\bibfnamefont {D.}~\bibnamefont {Landhuis}}, \bibinfo {author} {\bibfnamefont {P.}~\bibnamefont {Laptev}}, \bibinfo {author} {\bibfnamefont {K.-M.}\ \bibnamefont {Lau}}, \bibinfo {author} {\bibfnamefont {L.}~\bibnamefont {Laws}}, \bibinfo {author} {\bibfnamefont {J.}~\bibnamefont {Lee}}, \bibinfo {author} {\bibfnamefont {K.~W.}\ \bibnamefont {Lee}}, \bibinfo {author} {\bibfnamefont {Y.~D.}\ \bibnamefont {Lensky}}, \bibinfo {author} {\bibfnamefont {B.~J.}\ \bibnamefont {Lester}}, \bibinfo {author} {\bibfnamefont {A.~T.}\ \bibnamefont {Lill}}, \bibinfo {author} {\bibfnamefont {W.}~\bibnamefont {Liu}}, \bibinfo {author} {\bibfnamefont {A.}~\bibnamefont {Locharla}}, \bibinfo {author} {\bibfnamefont {O.}~\bibnamefont {Martin}}, \bibinfo {author} {\bibfnamefont {J.~R.}\ \bibnamefont {McClean}}, \bibinfo {author} {\bibfnamefont {M.}~\bibnamefont {McEwen}}, \bibinfo {author} {\bibfnamefont {K.~C.}\ \bibnamefont {Miao}}, \bibinfo {author} {\bibfnamefont {A.}~\bibnamefont {Mieszala}}, \bibinfo {author} {\bibfnamefont {S.}~\bibnamefont {Montazeri}}, \bibinfo {author} {\bibfnamefont {A.}~\bibnamefont {Morvan}}, \bibinfo {author} {\bibfnamefont {R.}~\bibnamefont {Movassagh}}, \bibinfo {author} {\bibfnamefont {W.}~\bibnamefont {Mruczkiewicz}}, \bibinfo {author} {\bibfnamefont {M.}~\bibnamefont {Neeley}}, \bibinfo {author} {\bibfnamefont {C.}~\bibnamefont {Neill}}, \bibinfo {author} {\bibfnamefont {A.}~\bibnamefont {Nersisyan}}, \bibinfo {author} {\bibfnamefont {M.}~\bibnamefont {Newman}}, \bibinfo {author} {\bibfnamefont {J.~H.}\ \bibnamefont {Ng}}, \bibinfo {author} {\bibfnamefont {A.}~\bibnamefont {Nguyen}}, \bibinfo {author} {\bibfnamefont {M.}~\bibnamefont {Nguyen}}, \bibinfo {author} {\bibfnamefont {M.~Y.}\ \bibnamefont {Niu}}, \bibinfo {author} {\bibfnamefont {T.~E.}\ \bibnamefont {O'Brien}}, \bibinfo {author} {\bibfnamefont {S.}~\bibnamefont {Omonije}}, \bibinfo {author} {\bibfnamefont {A.}~\bibnamefont {Opremcak}}, \bibinfo {author} {\bibfnamefont {A.}~\bibnamefont {Petukhov}}, \bibinfo {author} {\bibfnamefont {R.}~\bibnamefont {Potter}}, \bibinfo {author} {\bibfnamefont {L.~P.}\ \bibnamefont {Pryadko}}, \bibinfo {author} {\bibfnamefont {C.}~\bibnamefont {Quintana}}, \bibinfo {author} {\bibfnamefont {C.}~\bibnamefont {Rocque}}, \bibinfo {author} {\bibfnamefont {N.~C.}\ \bibnamefont {Rubin}}, \bibinfo {author} {\bibfnamefont {N.}~\bibnamefont {Saei}}, \bibinfo {author} {\bibfnamefont {D.}~\bibnamefont {Sank}}, \bibinfo {author} {\bibfnamefont {K.}~\bibnamefont {Sankaragomathi}}, \bibinfo {author} {\bibfnamefont {K.~J.}\ \bibnamefont {Satzinger}}, \bibinfo {author} {\bibfnamefont {H.~F.}\ \bibnamefont {Schurkus}}, \bibinfo {author} {\bibfnamefont {C.}~\bibnamefont {Schuster}}, \bibinfo {author} {\bibfnamefont {M.~J.}\ \bibnamefont {Shearn}}, \bibinfo {author} {\bibfnamefont {A.}~\bibnamefont {Shorter}}, \bibinfo {author} {\bibfnamefont {N.}~\bibnamefont {Shutty}}, \bibinfo {author} {\bibfnamefont {V.}~\bibnamefont {Shvarts}}, \bibinfo {author} {\bibfnamefont {J.}~\bibnamefont {Skruzny}}, \bibinfo {author} {\bibfnamefont {W.~C.}\ \bibnamefont {Smith}}, \bibinfo {author} {\bibfnamefont {R.}~\bibnamefont {Somma}}, \bibinfo {author} {\bibfnamefont {G.}~\bibnamefont {Sterling}}, \bibinfo {author} {\bibfnamefont {D.}~\bibnamefont {Strain}}, \bibinfo {author} {\bibfnamefont {M.}~\bibnamefont {Szalay}}, \bibinfo {author} {\bibfnamefont {A.}~\bibnamefont {Torres}}, \bibinfo {author} {\bibfnamefont {G.}~\bibnamefont {Vidal}}, \bibinfo {author} {\bibfnamefont {B.}~\bibnamefont {Villalonga}}, \bibinfo {author} {\bibfnamefont {C.~V.}\ \bibnamefont {Heidweiller}}, \bibinfo {author} {\bibfnamefont {T.}~\bibnamefont {White}}, \bibinfo {author} {\bibfnamefont {B.~W.~K.}\ \bibnamefont {Woo}}, \bibinfo {author} {\bibfnamefont {C.}~\bibnamefont {Xing}}, \bibinfo {author} {\bibfnamefont {Z.~J.}\ \bibnamefont {Yao}}, \bibinfo {author} {\bibfnamefont {P.}~\bibnamefont {Yeh}}, \bibinfo {author} {\bibfnamefont {J.}~\bibnamefont {Yoo}}, \bibinfo {author} {\bibfnamefont {G.}~\bibnamefont {Young}}, \bibinfo {author} {\bibfnamefont {A.}~\bibnamefont {Zalcman}}, \bibinfo {author} {\bibfnamefont {Y.}~\bibnamefont {Zhang}}, \bibinfo {author} {\bibfnamefont {N.}~\bibnamefont {Zhu}}, \bibinfo {author} {\bibfnamefont {N.}~\bibnamefont {Zobrist}}, \bibinfo {author} {\bibfnamefont {H.}~\bibnamefont {Neven}}, \bibinfo {author} {\bibfnamefont {R.}~\bibnamefont {Babbush}}, \bibinfo {author} {\bibfnamefont {D.}~\bibnamefont {Bacon}}, \bibinfo {author} {\bibfnamefont {S.}~\bibnamefont {Boixo}}, \bibinfo {author} {\bibfnamefont {J.}~\bibnamefont {Hilton}}, \bibinfo {author} {\bibfnamefont {E.}~\bibnamefont {Lucero}}, \bibinfo {author}
  {\bibfnamefont {A.}~\bibnamefont {Megrant}}, \bibinfo {author} {\bibfnamefont {J.}~\bibnamefont {Kelly}}, \bibinfo {author} {\bibfnamefont {Y.}~\bibnamefont {Chen}}, \bibinfo {author} {\bibfnamefont {V.}~\bibnamefont {Smelyanskiy}}, \bibinfo {author} {\bibfnamefont {X.}~\bibnamefont {Mi}}, \bibinfo {author} {\bibfnamefont {V.}~\bibnamefont {Khemani}}, \bibinfo {author} {\bibfnamefont {P.}~\bibnamefont {Roushan}}, \bibinfo {author} {\bibfnamefont {G.~Q.}\ \bibnamefont {AI}},\ and\ \bibinfo {author} {\bibnamefont {{Collaborators}}},\ }\bibfield  {title} {\bibinfo {title} {Measurement-induced entanglement and teleportation on a noisy quantum processor},\ }\href {https://doi.org/10.1038/s41586-023-06505-7} {\bibfield  {journal} {\bibinfo  {journal} {Nature}\ }\textbf {\bibinfo {volume} {622}},\ \bibinfo {pages} {481} (\bibinfo {year} {2023})}\BibitemShut {NoStop}%
\bibitem [{\citenamefont {Kamakari}\ \emph {et~al.}(2024)\citenamefont {Kamakari}, \citenamefont {Sun}, \citenamefont {Li}, \citenamefont {Thio}, \citenamefont {Gujarati}, \citenamefont {Fisher}, \citenamefont {Motta},\ and\ \citenamefont {Minnich}}]{kamakari2024scalablecrossentropy}%
  \BibitemOpen
  \bibfield  {author} {\bibinfo {author} {\bibfnamefont {H.}~\bibnamefont {Kamakari}}, \bibinfo {author} {\bibfnamefont {J.}~\bibnamefont {Sun}}, \bibinfo {author} {\bibfnamefont {Y.}~\bibnamefont {Li}}, \bibinfo {author} {\bibfnamefont {J.~J.}\ \bibnamefont {Thio}}, \bibinfo {author} {\bibfnamefont {T.~P.}\ \bibnamefont {Gujarati}}, \bibinfo {author} {\bibfnamefont {M.~P.~A.}\ \bibnamefont {Fisher}}, \bibinfo {author} {\bibfnamefont {M.}~\bibnamefont {Motta}},\ and\ \bibinfo {author} {\bibfnamefont {A.~J.}\ \bibnamefont {Minnich}},\ }\href {https://arxiv.org/abs/2403.00938} {\bibinfo {title} {Experimental demonstration of scalable cross-entropy benchmarking to detect measurement-induced phase transitions on a superconducting quantum processor}} (\bibinfo {year} {2024}),\ \Eprint {https://arxiv.org/abs/2403.00938} {arXiv:2403.00938 [quant-ph]} \BibitemShut {NoStop}%
\bibitem [{\citenamefont {Grover}(1996)}]{grover1996fast}%
  \BibitemOpen
  \bibfield  {author} {\bibinfo {author} {\bibfnamefont {L.~K.}\ \bibnamefont {Grover}},\ }\href@noop {} {\bibinfo {title} {A fast quantum mechanical algorithm for database search}} (\bibinfo {year} {1996}),\ \Eprint {https://arxiv.org/abs/quant-ph/9605043} {arXiv:quant-ph/9605043 [quant-ph]} \BibitemShut {NoStop}%
\bibitem [{\citenamefont {Yoder}\ \emph {et~al.}(2014)\citenamefont {Yoder}, \citenamefont {Low},\ and\ \citenamefont {Chuang}}]{PhysRevLett.113.210501}%
  \BibitemOpen
  \bibfield  {author} {\bibinfo {author} {\bibfnamefont {T.~J.}\ \bibnamefont {Yoder}}, \bibinfo {author} {\bibfnamefont {G.~H.}\ \bibnamefont {Low}},\ and\ \bibinfo {author} {\bibfnamefont {I.~L.}\ \bibnamefont {Chuang}},\ }\bibfield  {title} {\bibinfo {title} {Fixed-point quantum search with an optimal number of queries},\ }\href {https://doi.org/10.1103/PhysRevLett.113.210501} {\bibfield  {journal} {\bibinfo  {journal} {Phys. Rev. Lett.}\ }\textbf {\bibinfo {volume} {113}},\ \bibinfo {pages} {210501} (\bibinfo {year} {2014})}\BibitemShut {NoStop}%
\bibitem [{\citenamefont {Gily\'{e}n}\ \emph {et~al.}(2019)\citenamefont {Gily\'{e}n}, \citenamefont {Su}, \citenamefont {Low},\ and\ \citenamefont {Wiebe}}]{GilyenQSVT}%
  \BibitemOpen
  \bibfield  {author} {\bibinfo {author} {\bibfnamefont {A.}~\bibnamefont {Gily\'{e}n}}, \bibinfo {author} {\bibfnamefont {Y.}~\bibnamefont {Su}}, \bibinfo {author} {\bibfnamefont {G.~H.}\ \bibnamefont {Low}},\ and\ \bibinfo {author} {\bibfnamefont {N.}~\bibnamefont {Wiebe}},\ }\bibfield  {title} {\bibinfo {title} {Quantum singular value transformation and beyond: exponential improvements for quantum matrix arithmetics},\ }in\ \href {https://doi.org/10.1145/3313276.3316366} {\emph {\bibinfo {booktitle} {Proceedings of the 51st Annual ACM SIGACT Symposium on Theory of Computing}}},\ \bibinfo {series and number} {STOC 2019}\ (\bibinfo  {publisher} {Association for Computing Machinery},\ \bibinfo {address} {New York, NY, USA},\ \bibinfo {year} {2019})\ p.\ \bibinfo {pages} {193–204}\BibitemShut {NoStop}%
\bibitem [{\citenamefont {Martyn}\ \emph {et~al.}(2021)\citenamefont {Martyn}, \citenamefont {Rossi}, \citenamefont {Tan},\ and\ \citenamefont {Chuang}}]{PRXQuantum.2.040203}%
  \BibitemOpen
  \bibfield  {author} {\bibinfo {author} {\bibfnamefont {J.~M.}\ \bibnamefont {Martyn}}, \bibinfo {author} {\bibfnamefont {Z.~M.}\ \bibnamefont {Rossi}}, \bibinfo {author} {\bibfnamefont {A.~K.}\ \bibnamefont {Tan}},\ and\ \bibinfo {author} {\bibfnamefont {I.~L.}\ \bibnamefont {Chuang}},\ }\bibfield  {title} {\bibinfo {title} {Grand unification of quantum algorithms},\ }\href {https://doi.org/10.1103/PRXQuantum.2.040203} {\bibfield  {journal} {\bibinfo  {journal} {PRX Quantum}\ }\textbf {\bibinfo {volume} {2}},\ \bibinfo {pages} {040203} (\bibinfo {year} {2021})}\BibitemShut {NoStop}%
\bibitem [{\citenamefont {Bennett}\ \emph {et~al.}(1997)\citenamefont {Bennett}, \citenamefont {Bernstein}, \citenamefont {Brassard},\ and\ \citenamefont {Vazirani}}]{BBBV1997}%
  \BibitemOpen
  \bibfield  {author} {\bibinfo {author} {\bibfnamefont {C.~H.}\ \bibnamefont {Bennett}}, \bibinfo {author} {\bibfnamefont {E.}~\bibnamefont {Bernstein}}, \bibinfo {author} {\bibfnamefont {G.}~\bibnamefont {Brassard}},\ and\ \bibinfo {author} {\bibfnamefont {U.}~\bibnamefont {Vazirani}},\ }\bibfield  {title} {\bibinfo {title} {Strengths and weaknesses of quantum computing},\ }\href {https://doi.org/10.1137/s0097539796300933} {\bibfield  {journal} {\bibinfo  {journal} {SIAM Journal on Computing}\ }\textbf {\bibinfo {volume} {26}},\ \bibinfo {pages} {1510–1523} (\bibinfo {year} {1997})}\BibitemShut {NoStop}%
\bibitem [{\citenamefont {Bennett}\ \emph {et~al.}(1993)\citenamefont {Bennett}, \citenamefont {Brassard}, \citenamefont {Cr\'epeau}, \citenamefont {Jozsa}, \citenamefont {Peres},\ and\ \citenamefont {Wootters}}]{PhysRevLett.70.1895}%
  \BibitemOpen
  \bibfield  {author} {\bibinfo {author} {\bibfnamefont {C.~H.}\ \bibnamefont {Bennett}}, \bibinfo {author} {\bibfnamefont {G.}~\bibnamefont {Brassard}}, \bibinfo {author} {\bibfnamefont {C.}~\bibnamefont {Cr\'epeau}}, \bibinfo {author} {\bibfnamefont {R.}~\bibnamefont {Jozsa}}, \bibinfo {author} {\bibfnamefont {A.}~\bibnamefont {Peres}},\ and\ \bibinfo {author} {\bibfnamefont {W.~K.}\ \bibnamefont {Wootters}},\ }\bibfield  {title} {\bibinfo {title} {Teleporting an unknown quantum state via dual classical and einstein-podolsky-rosen channels},\ }\href {https://doi.org/10.1103/PhysRevLett.70.1895} {\bibfield  {journal} {\bibinfo  {journal} {Phys. Rev. Lett.}\ }\textbf {\bibinfo {volume} {70}},\ \bibinfo {pages} {1895} (\bibinfo {year} {1993})}\BibitemShut {NoStop}%
\bibitem [{\citenamefont {Bao}\ \emph {et~al.}(2024)\citenamefont {Bao}, \citenamefont {Block},\ and\ \citenamefont {Altman}}]{PhysRevLett.132.030401}%
  \BibitemOpen
  \bibfield  {author} {\bibinfo {author} {\bibfnamefont {Y.}~\bibnamefont {Bao}}, \bibinfo {author} {\bibfnamefont {M.}~\bibnamefont {Block}},\ and\ \bibinfo {author} {\bibfnamefont {E.}~\bibnamefont {Altman}},\ }\bibfield  {title} {\bibinfo {title} {Finite-time teleportation phase transition in random quantum circuits},\ }\href {https://doi.org/10.1103/PhysRevLett.132.030401} {\bibfield  {journal} {\bibinfo  {journal} {Phys. Rev. Lett.}\ }\textbf {\bibinfo {volume} {132}},\ \bibinfo {pages} {030401} (\bibinfo {year} {2024})}\BibitemShut {NoStop}%
\bibitem [{\citenamefont {Antonini}\ \emph {et~al.}(2022)\citenamefont {Antonini}, \citenamefont {Bentsen}, \citenamefont {Cao}, \citenamefont {Harper}, \citenamefont {Jian},\ and\ \citenamefont {Swingle}}]{Antonini2022}%
  \BibitemOpen
  \bibfield  {author} {\bibinfo {author} {\bibfnamefont {S.}~\bibnamefont {Antonini}}, \bibinfo {author} {\bibfnamefont {G.}~\bibnamefont {Bentsen}}, \bibinfo {author} {\bibfnamefont {C.}~\bibnamefont {Cao}}, \bibinfo {author} {\bibfnamefont {J.}~\bibnamefont {Harper}}, \bibinfo {author} {\bibfnamefont {S.-K.}\ \bibnamefont {Jian}},\ and\ \bibinfo {author} {\bibfnamefont {B.}~\bibnamefont {Swingle}},\ }\bibfield  {title} {\bibinfo {title} {Holographic measurement and bulk teleportation},\ }\href {https://doi.org/10.1007/JHEP12(2022)124} {\bibfield  {journal} {\bibinfo  {journal} {Journal of High Energy Physics}\ }\textbf {\bibinfo {volume} {2022}},\ \bibinfo {pages} {124} (\bibinfo {year} {2022})}\BibitemShut {NoStop}%
\bibitem [{\citenamefont {Hayden}\ \emph {et~al.}(2008)\citenamefont {Hayden}, \citenamefont {Horodecki}, \citenamefont {Winter},\ and\ \citenamefont {Yard}}]{hayden2008decoupling}%
  \BibitemOpen
  \bibfield  {author} {\bibinfo {author} {\bibfnamefont {P.}~\bibnamefont {Hayden}}, \bibinfo {author} {\bibfnamefont {M.}~\bibnamefont {Horodecki}}, \bibinfo {author} {\bibfnamefont {A.}~\bibnamefont {Winter}},\ and\ \bibinfo {author} {\bibfnamefont {J.}~\bibnamefont {Yard}},\ }\bibfield  {title} {\bibinfo {title} {A decoupling approach to the quantum capacity},\ }\href {https://doi.org/10.1142/S1230161208000043} {\bibfield  {journal} {\bibinfo  {journal} {Open Systems \& Information Dynamics}\ }\textbf {\bibinfo {volume} {15}},\ \bibinfo {pages} {7} (\bibinfo {year} {2008})},\ \Eprint {https://arxiv.org/abs/https://doi.org/10.1142/S1230161208000043} {https://doi.org/10.1142/S1230161208000043} \BibitemShut {NoStop}%
\bibitem [{\citenamefont {Boixo}\ \emph {et~al.}(2018)\citenamefont {Boixo}, \citenamefont {Isakov}, \citenamefont {Smelyanskiy}, \citenamefont {Babbush}, \citenamefont {Ding}, \citenamefont {Jiang}, \citenamefont {Bremner}, \citenamefont {Martinis},\ and\ \citenamefont {Neven}}]{Boixo2018}%
  \BibitemOpen
  \bibfield  {author} {\bibinfo {author} {\bibfnamefont {S.}~\bibnamefont {Boixo}}, \bibinfo {author} {\bibfnamefont {S.~V.}\ \bibnamefont {Isakov}}, \bibinfo {author} {\bibfnamefont {V.~N.}\ \bibnamefont {Smelyanskiy}}, \bibinfo {author} {\bibfnamefont {R.}~\bibnamefont {Babbush}}, \bibinfo {author} {\bibfnamefont {N.}~\bibnamefont {Ding}}, \bibinfo {author} {\bibfnamefont {Z.}~\bibnamefont {Jiang}}, \bibinfo {author} {\bibfnamefont {M.~J.}\ \bibnamefont {Bremner}}, \bibinfo {author} {\bibfnamefont {J.~M.}\ \bibnamefont {Martinis}},\ and\ \bibinfo {author} {\bibfnamefont {H.}~\bibnamefont {Neven}},\ }\bibfield  {title} {\bibinfo {title} {Characterizing quantum supremacy in near-term devices},\ }\href {https://doi.org/10.1038/s41567-018-0124-x} {\bibfield  {journal} {\bibinfo  {journal} {Nature Physics}\ }\textbf {\bibinfo {volume} {14}},\ \bibinfo {pages} {595} (\bibinfo {year} {2018})}\BibitemShut {NoStop}%
\bibitem [{\citenamefont {Hangleiter}\ and\ \citenamefont {Eisert}(2023)}]{RevModPhys.95.035001}%
  \BibitemOpen
  \bibfield  {author} {\bibinfo {author} {\bibfnamefont {D.}~\bibnamefont {Hangleiter}}\ and\ \bibinfo {author} {\bibfnamefont {J.}~\bibnamefont {Eisert}},\ }\bibfield  {title} {\bibinfo {title} {Computational advantage of quantum random sampling},\ }\href {https://doi.org/10.1103/RevModPhys.95.035001} {\bibfield  {journal} {\bibinfo  {journal} {Rev. Mod. Phys.}\ }\textbf {\bibinfo {volume} {95}},\ \bibinfo {pages} {035001} (\bibinfo {year} {2023})}\BibitemShut {NoStop}%
\bibitem [{\citenamefont {\textit{et al.}}(2025)}]{Kim2025talk}%
  \BibitemOpen
  \bibfield  {author} {\bibinfo {author} {\bibfnamefont {E.-A.~K.}\ \bibnamefont {\textit{et al.}}},\ }\href@noop {} {\bibinfo {title} {To appear}} (\bibinfo {year} {2025}),\ \bibinfo {note} {slides available at \url{https://drive.google.com/file/d/1e8ylNadjgGqNETmtecw1Fqe6EY4OybE0/view}}\BibitemShut {NoStop}%
\bibitem [{\citenamefont {AI}\ and\ \citenamefont {collaborators}(2025)}]{You2025toappear}%
  \BibitemOpen
  \bibfield  {author} {\bibinfo {author} {\bibfnamefont {G.~Q.}\ \bibnamefont {AI}}\ and\ \bibinfo {author} {\bibnamefont {collaborators}},\ }\href@noop {} {\bibinfo {title} {To appear}} (\bibinfo {year} {2025})\BibitemShut {NoStop}%
\bibitem [{\citenamefont {Huang}\ \emph {et~al.}(2020)\citenamefont {Huang}, \citenamefont {Kueng},\ and\ \citenamefont {Preskill}}]{Huang2020}%
  \BibitemOpen
  \bibfield  {author} {\bibinfo {author} {\bibfnamefont {H.-Y.}\ \bibnamefont {Huang}}, \bibinfo {author} {\bibfnamefont {R.}~\bibnamefont {Kueng}},\ and\ \bibinfo {author} {\bibfnamefont {J.}~\bibnamefont {Preskill}},\ }\bibfield  {title} {\bibinfo {title} {Predicting many properties of a quantum system from very few measurements},\ }\href {https://doi.org/10.1038/s41567-020-0932-7} {\bibfield  {journal} {\bibinfo  {journal} {Nature Physics}\ }\textbf {\bibinfo {volume} {16}},\ \bibinfo {pages} {1050} (\bibinfo {year} {2020})}\BibitemShut {NoStop}%
\bibitem [{\citenamefont {Ippoliti}\ and\ \citenamefont {Khemani}(2021)}]{PhysRevLett.126.060501}%
  \BibitemOpen
  \bibfield  {author} {\bibinfo {author} {\bibfnamefont {M.}~\bibnamefont {Ippoliti}}\ and\ \bibinfo {author} {\bibfnamefont {V.}~\bibnamefont {Khemani}},\ }\bibfield  {title} {\bibinfo {title} {Postselection-free entanglement dynamics via spacetime duality},\ }\href {https://doi.org/10.1103/PhysRevLett.126.060501} {\bibfield  {journal} {\bibinfo  {journal} {Phys. Rev. Lett.}\ }\textbf {\bibinfo {volume} {126}},\ \bibinfo {pages} {060501} (\bibinfo {year} {2021})}\BibitemShut {NoStop}%
\bibitem [{\citenamefont {Brassard}(1997)}]{Brassard1997}%
  \BibitemOpen
  \bibfield  {author} {\bibinfo {author} {\bibfnamefont {G.}~\bibnamefont {Brassard}},\ }\bibfield  {title} {\bibinfo {title} {Searching a quantum phone book},\ }\href {https://doi.org/10.1126/science.275.5300.627} {\bibfield  {journal} {\bibinfo  {journal} {Science}\ }\textbf {\bibinfo {volume} {275}},\ \bibinfo {pages} {627} (\bibinfo {year} {1997})},\ \Eprint {https://arxiv.org/abs/https://www.science.org/doi/pdf/10.1126/science.275.5300.627} {https://www.science.org/doi/pdf/10.1126/science.275.5300.627} \BibitemShut {NoStop}%
\bibitem [{\citenamefont {Dong}\ \emph {et~al.}(2021)\citenamefont {Dong}, \citenamefont {Meng}, \citenamefont {Whaley},\ and\ \citenamefont {Lin}}]{Dong_2021}%
  \BibitemOpen
  \bibfield  {author} {\bibinfo {author} {\bibfnamefont {Y.}~\bibnamefont {Dong}}, \bibinfo {author} {\bibfnamefont {X.}~\bibnamefont {Meng}}, \bibinfo {author} {\bibfnamefont {K.~B.}\ \bibnamefont {Whaley}},\ and\ \bibinfo {author} {\bibfnamefont {L.}~\bibnamefont {Lin}},\ }\bibfield  {title} {\bibinfo {title} {Efficient phase-factor evaluation in quantum signal processing},\ }\bibfield  {journal} {\bibinfo  {journal} {Physical Review A}\ }\textbf {\bibinfo {volume} {103}},\ \href {https://doi.org/10.1103/physreva.103.042419} {10.1103/physreva.103.042419} (\bibinfo {year} {2021})\BibitemShut {NoStop}%
\bibitem [{Note1()}]{Note1}%
  \BibitemOpen
  \bibinfo {note} {As a remark, in Fig.~\ref {fig:post_to_AA}(c) and all subsequent QSVT implementations, we realize an odd degree-$d$ polynomial but use only $(d-1)$ applications of $(\Pi _0)_\phi $ and $(\Pi _m)_\phi $. This is because the final $(\Pi _m)_\phi $ acts like a global phase within the $\Pi _m$ subspace and can thus be omitted.}\BibitemShut {Stop}%
\bibitem [{\citenamefont {Low}(2017)}]{Low2017QSP}%
  \BibitemOpen
  \bibfield  {author} {\bibinfo {author} {\bibfnamefont {G.~H.}\ \bibnamefont {Low}},\ }\emph {\bibinfo {title} {Quantum Signal Processing by Single-Qubit Dynamics}},\ \href@noop {} {\bibinfo {type} {Ph.d. thesis}},\ \bibinfo  {school} {Massachusetts Institute of Technology, Department of Physics} (\bibinfo {year} {2017})\BibitemShut {NoStop}%
\bibitem [{\citenamefont {Buhrman}\ \emph {et~al.}(2001)\citenamefont {Buhrman}, \citenamefont {Cleve}, \citenamefont {Watrous},\ and\ \citenamefont {de~Wolf}}]{PhysRevLett.87.167902}%
  \BibitemOpen
  \bibfield  {author} {\bibinfo {author} {\bibfnamefont {H.}~\bibnamefont {Buhrman}}, \bibinfo {author} {\bibfnamefont {R.}~\bibnamefont {Cleve}}, \bibinfo {author} {\bibfnamefont {J.}~\bibnamefont {Watrous}},\ and\ \bibinfo {author} {\bibfnamefont {R.}~\bibnamefont {de~Wolf}},\ }\bibfield  {title} {\bibinfo {title} {Quantum fingerprinting},\ }\href {https://doi.org/10.1103/PhysRevLett.87.167902} {\bibfield  {journal} {\bibinfo  {journal} {Phys. Rev. Lett.}\ }\textbf {\bibinfo {volume} {87}},\ \bibinfo {pages} {167902} (\bibinfo {year} {2001})}\BibitemShut {NoStop}%
\bibitem [{\citenamefont {van Enk}\ and\ \citenamefont {Beenakker}(2012)}]{PhysRevLett.108.110503}%
  \BibitemOpen
  \bibfield  {author} {\bibinfo {author} {\bibfnamefont {S.~J.}\ \bibnamefont {van Enk}}\ and\ \bibinfo {author} {\bibfnamefont {C.~W.~J.}\ \bibnamefont {Beenakker}},\ }\bibfield  {title} {\bibinfo {title} {Measuring $\mathrm{Tr}{\ensuremath{\rho}}^{n}$ on single copies of $\ensuremath{\rho}$ using random measurements},\ }\href {https://doi.org/10.1103/PhysRevLett.108.110503} {\bibfield  {journal} {\bibinfo  {journal} {Phys. Rev. Lett.}\ }\textbf {\bibinfo {volume} {108}},\ \bibinfo {pages} {110503} (\bibinfo {year} {2012})}\BibitemShut {NoStop}%
\bibitem [{\citenamefont {Elben}\ \emph {et~al.}(2023)\citenamefont {Elben}, \citenamefont {Flammia}, \citenamefont {Huang}, \citenamefont {Kueng}, \citenamefont {Preskill}, \citenamefont {Vermersch},\ and\ \citenamefont {Zoller}}]{Elben2023toolbox}%
  \BibitemOpen
  \bibfield  {author} {\bibinfo {author} {\bibfnamefont {A.}~\bibnamefont {Elben}}, \bibinfo {author} {\bibfnamefont {S.~T.}\ \bibnamefont {Flammia}}, \bibinfo {author} {\bibfnamefont {H.-Y.}\ \bibnamefont {Huang}}, \bibinfo {author} {\bibfnamefont {R.}~\bibnamefont {Kueng}}, \bibinfo {author} {\bibfnamefont {J.}~\bibnamefont {Preskill}}, \bibinfo {author} {\bibfnamefont {B.}~\bibnamefont {Vermersch}},\ and\ \bibinfo {author} {\bibfnamefont {P.}~\bibnamefont {Zoller}},\ }\bibfield  {title} {\bibinfo {title} {The randomized measurement toolbox},\ }\href {https://doi.org/10.1038/s42254-022-00535-2} {\bibfield  {journal} {\bibinfo  {journal} {Nature Reviews Physics}\ }\textbf {\bibinfo {volume} {5}},\ \bibinfo {pages} {9} (\bibinfo {year} {2023})}\BibitemShut {NoStop}%
\bibitem [{Note2()}]{Note2}%
  \BibitemOpen
  \bibinfo {note} {In general, if the experimental protocol $\protect \hat {\protect \mathcal {O}}$ requires coherent measurements across all copies, then we must maintain $k$ copies in $k$ identical quantum registers. However, if no coherent measurement is required---as in the case of randomized or adaptive measurements---then one quantum register can be reused to prepare $\mathinner {|{\psi _m}\rangle }$ multiple times.}\BibitemShut {Stop}%
\bibitem [{\citenamefont {Low}\ and\ \citenamefont {Wiebe}(2019)}]{low2019interactionpicture}%
  \BibitemOpen
  \bibfield  {author} {\bibinfo {author} {\bibfnamefont {G.~H.}\ \bibnamefont {Low}}\ and\ \bibinfo {author} {\bibfnamefont {N.}~\bibnamefont {Wiebe}},\ }\href {https://arxiv.org/abs/1805.00675} {\bibinfo {title} {Hamiltonian simulation in the interaction picture}} (\bibinfo {year} {2019}),\ \Eprint {https://arxiv.org/abs/1805.00675} {arXiv:1805.00675 [quant-ph]} \BibitemShut {NoStop}%
\bibitem [{\citenamefont {Fang}\ \emph {et~al.}(2023)\citenamefont {Fang}, \citenamefont {Lin},\ and\ \citenamefont {Tong}}]{Fang_2023}%
  \BibitemOpen
  \bibfield  {author} {\bibinfo {author} {\bibfnamefont {D.}~\bibnamefont {Fang}}, \bibinfo {author} {\bibfnamefont {L.}~\bibnamefont {Lin}},\ and\ \bibinfo {author} {\bibfnamefont {Y.}~\bibnamefont {Tong}},\ }\bibfield  {title} {\bibinfo {title} {Time-marching based quantum solvers for time-dependent linear differential equations},\ }\href {https://doi.org/10.22331/q-2023-03-20-955} {\bibfield  {journal} {\bibinfo  {journal} {Quantum}\ }\textbf {\bibinfo {volume} {7}},\ \bibinfo {pages} {955} (\bibinfo {year} {2023})}\BibitemShut {NoStop}%
\bibitem [{\citenamefont {Low}\ and\ \citenamefont {Chuang}(2017)}]{low2017hamiltoniansimulationuniformspectral}%
  \BibitemOpen
  \bibfield  {author} {\bibinfo {author} {\bibfnamefont {G.~H.}\ \bibnamefont {Low}}\ and\ \bibinfo {author} {\bibfnamefont {I.~L.}\ \bibnamefont {Chuang}},\ }\href {https://arxiv.org/abs/1707.05391} {\bibinfo {title} {Hamiltonian simulation by uniform spectral amplification}} (\bibinfo {year} {2017}),\ \Eprint {https://arxiv.org/abs/1707.05391} {arXiv:1707.05391 [quant-ph]} \BibitemShut {NoStop}%
\bibitem [{\citenamefont {Utsumi}\ and\ \citenamefont {Nakata}(2024)}]{utsumi2024explicit}%
  \BibitemOpen
  \bibfield  {author} {\bibinfo {author} {\bibfnamefont {T.}~\bibnamefont {Utsumi}}\ and\ \bibinfo {author} {\bibfnamefont {Y.}~\bibnamefont {Nakata}},\ }\href {https://arxiv.org/abs/2405.06051} {\bibinfo {title} {Explicit decoders using fixed-point amplitude amplification based on qsvt}} (\bibinfo {year} {2024}),\ \Eprint {https://arxiv.org/abs/2405.06051} {arXiv:2405.06051 [quant-ph]} \BibitemShut {NoStop}%
\bibitem [{\citenamefont {Yoshida}\ and\ \citenamefont {Kitaev}(2017)}]{yoshida2017efficient}%
  \BibitemOpen
  \bibfield  {author} {\bibinfo {author} {\bibfnamefont {B.}~\bibnamefont {Yoshida}}\ and\ \bibinfo {author} {\bibfnamefont {A.}~\bibnamefont {Kitaev}},\ }\href@noop {} {\bibinfo {title} {Efficient decoding for the hayden-preskill protocol}} (\bibinfo {year} {2017}),\ \Eprint {https://arxiv.org/abs/1710.03363} {arXiv:1710.03363 [hep-th]} \BibitemShut {NoStop}%
\bibitem [{\citenamefont {Briegel}\ \emph {et~al.}(2009)\citenamefont {Briegel}, \citenamefont {Browne}, \citenamefont {D{\"u}r}, \citenamefont {Raussendorf},\ and\ \citenamefont {Van~den Nest}}]{Briegel2009}%
  \BibitemOpen
  \bibfield  {author} {\bibinfo {author} {\bibfnamefont {H.~J.}\ \bibnamefont {Briegel}}, \bibinfo {author} {\bibfnamefont {D.~E.}\ \bibnamefont {Browne}}, \bibinfo {author} {\bibfnamefont {W.}~\bibnamefont {D{\"u}r}}, \bibinfo {author} {\bibfnamefont {R.}~\bibnamefont {Raussendorf}},\ and\ \bibinfo {author} {\bibfnamefont {M.}~\bibnamefont {Van~den Nest}},\ }\bibfield  {title} {\bibinfo {title} {Measurement-based quantum computation},\ }\href {https://doi.org/10.1038/nphys1157} {\bibfield  {journal} {\bibinfo  {journal} {Nature Physics}\ }\textbf {\bibinfo {volume} {5}},\ \bibinfo {pages} {19} (\bibinfo {year} {2009})}\BibitemShut {NoStop}%
\bibitem [{\citenamefont {Wei}(2021)}]{Wei_2021}%
  \BibitemOpen
  \bibfield  {author} {\bibinfo {author} {\bibfnamefont {T.-C.}\ \bibnamefont {Wei}},\ }\href {https://doi.org/10.1093/acrefore/9780190871994.013.31} {\bibinfo {title} {Measurement-based quantum computation}} (\bibinfo {year} {2021})\BibitemShut {NoStop}%
\bibitem [{\citenamefont {Yoshida}(2022)}]{yoshida2022recovery}%
  \BibitemOpen
  \bibfield  {author} {\bibinfo {author} {\bibfnamefont {B.}~\bibnamefont {Yoshida}},\ }\href {https://arxiv.org/abs/2106.15628} {\bibinfo {title} {Recovery algorithms for clifford hayden-preskill problem}} (\bibinfo {year} {2022}),\ \Eprint {https://arxiv.org/abs/2106.15628} {arXiv:2106.15628 [quant-ph]} \BibitemShut {NoStop}%
\bibitem [{\citenamefont {Bostanci}\ \emph {et~al.}(2023)\citenamefont {Bostanci}, \citenamefont {Efron}, \citenamefont {Metger}, \citenamefont {Poremba}, \citenamefont {Qian},\ and\ \citenamefont {Yuen}}]{bostanci2023unitary}%
  \BibitemOpen
  \bibfield  {author} {\bibinfo {author} {\bibfnamefont {J.}~\bibnamefont {Bostanci}}, \bibinfo {author} {\bibfnamefont {Y.}~\bibnamefont {Efron}}, \bibinfo {author} {\bibfnamefont {T.}~\bibnamefont {Metger}}, \bibinfo {author} {\bibfnamefont {A.}~\bibnamefont {Poremba}}, \bibinfo {author} {\bibfnamefont {L.}~\bibnamefont {Qian}},\ and\ \bibinfo {author} {\bibfnamefont {H.}~\bibnamefont {Yuen}},\ }\href {https://arxiv.org/abs/2306.13073} {\bibinfo {title} {Unitary complexity and the uhlmann transformation problem}} (\bibinfo {year} {2023}),\ \Eprint {https://arxiv.org/abs/2306.13073} {arXiv:2306.13073 [quant-ph]} \BibitemShut {NoStop}%
\bibitem [{\citenamefont {Horowitz}\ and\ \citenamefont {Maldacena}(2004)}]{Horowitz_2004}%
  \BibitemOpen
  \bibfield  {author} {\bibinfo {author} {\bibfnamefont {G.~T.}\ \bibnamefont {Horowitz}}\ and\ \bibinfo {author} {\bibfnamefont {J.}~\bibnamefont {Maldacena}},\ }\bibfield  {title} {\bibinfo {title} {The black hole final state},\ }\href {https://doi.org/10.1088/1126-6708/2004/02/008} {\bibfield  {journal} {\bibinfo  {journal} {Journal of High Energy Physics}\ }\textbf {\bibinfo {volume} {2004}},\ \bibinfo {pages} {008–008} (\bibinfo {year} {2004})}\BibitemShut {NoStop}%
\bibitem [{\citenamefont {Brown}\ \emph {et~al.}(2019)\citenamefont {Brown}, \citenamefont {Gharibyan}, \citenamefont {Penington},\ and\ \citenamefont {Susskind}}]{brown2019pythonslunch}%
  \BibitemOpen
  \bibfield  {author} {\bibinfo {author} {\bibfnamefont {A.~R.}\ \bibnamefont {Brown}}, \bibinfo {author} {\bibfnamefont {H.}~\bibnamefont {Gharibyan}}, \bibinfo {author} {\bibfnamefont {G.}~\bibnamefont {Penington}},\ and\ \bibinfo {author} {\bibfnamefont {L.}~\bibnamefont {Susskind}},\ }\href {https://arxiv.org/abs/1912.00228} {\bibinfo {title} {The python's lunch: geometric obstructions to decoding hawking radiation}} (\bibinfo {year} {2019}),\ \Eprint {https://arxiv.org/abs/1912.00228} {arXiv:1912.00228 [hep-th]} \BibitemShut {NoStop}%
\bibitem [{\citenamefont {Preskill}(2018)}]{Preskill2018qc}%
  \BibitemOpen
  \bibfield  {author} {\bibinfo {author} {\bibfnamefont {J.}~\bibnamefont {Preskill}},\ }\bibfield  {title} {\bibinfo {title} {Quantum {C}omputing in the {NISQ} era and beyond},\ }\href {https://doi.org/10.22331/q-2018-08-06-79} {\bibfield  {journal} {\bibinfo  {journal} {{Quantum}}\ }\textbf {\bibinfo {volume} {2}},\ \bibinfo {pages} {79} (\bibinfo {year} {2018})}\BibitemShut {NoStop}%
\bibitem [{\citenamefont {Chen}\ \emph {et~al.}(2023)\citenamefont {Chen}, \citenamefont {Cotler}, \citenamefont {Huang},\ and\ \citenamefont {Li}}]{Chen2023}%
  \BibitemOpen
  \bibfield  {author} {\bibinfo {author} {\bibfnamefont {S.}~\bibnamefont {Chen}}, \bibinfo {author} {\bibfnamefont {J.}~\bibnamefont {Cotler}}, \bibinfo {author} {\bibfnamefont {H.-Y.}\ \bibnamefont {Huang}},\ and\ \bibinfo {author} {\bibfnamefont {J.}~\bibnamefont {Li}},\ }\bibfield  {title} {\bibinfo {title} {The complexity of nisq},\ }\href {https://doi.org/10.1038/s41467-023-41217-6} {\bibfield  {journal} {\bibinfo  {journal} {Nature Communications}\ }\textbf {\bibinfo {volume} {14}},\ \bibinfo {pages} {6001} (\bibinfo {year} {2023})}\BibitemShut {NoStop}%
\bibitem [{\citenamefont {Hangleiter}\ \emph {et~al.}(2018)\citenamefont {Hangleiter}, \citenamefont {Bermejo-Vega}, \citenamefont {Schwarz},\ and\ \citenamefont {Eisert}}]{Hangleiter2018anticoncentration}%
  \BibitemOpen
  \bibfield  {author} {\bibinfo {author} {\bibfnamefont {D.}~\bibnamefont {Hangleiter}}, \bibinfo {author} {\bibfnamefont {J.}~\bibnamefont {Bermejo-Vega}}, \bibinfo {author} {\bibfnamefont {M.}~\bibnamefont {Schwarz}},\ and\ \bibinfo {author} {\bibfnamefont {J.}~\bibnamefont {Eisert}},\ }\bibfield  {title} {\bibinfo {title} {Anticoncentration theorems for schemes showing a quantum speedup},\ }\href {https://doi.org/10.22331/q-2018-05-22-65} {\bibfield  {journal} {\bibinfo  {journal} {{Quantum}}\ }\textbf {\bibinfo {volume} {2}},\ \bibinfo {pages} {65} (\bibinfo {year} {2018})}\BibitemShut {NoStop}%
\bibitem [{\citenamefont {Dalzell}\ \emph {et~al.}(2022)\citenamefont {Dalzell}, \citenamefont {Hunter-Jones},\ and\ \citenamefont {Brand\~ao}}]{PRXQuantum.3.010333}%
  \BibitemOpen
  \bibfield  {author} {\bibinfo {author} {\bibfnamefont {A.~M.}\ \bibnamefont {Dalzell}}, \bibinfo {author} {\bibfnamefont {N.}~\bibnamefont {Hunter-Jones}},\ and\ \bibinfo {author} {\bibfnamefont {F.~G. S.~L.}\ \bibnamefont {Brand\~ao}},\ }\bibfield  {title} {\bibinfo {title} {Random quantum circuits anticoncentrate in log depth},\ }\href {https://doi.org/10.1103/PRXQuantum.3.010333} {\bibfield  {journal} {\bibinfo  {journal} {PRX Quantum}\ }\textbf {\bibinfo {volume} {3}},\ \bibinfo {pages} {010333} (\bibinfo {year} {2022})}\BibitemShut {NoStop}%
\bibitem [{\citenamefont {Harrow}\ \emph {et~al.}(2009)\citenamefont {Harrow}, \citenamefont {Hassidim},\ and\ \citenamefont {Lloyd}}]{HHL2009}%
  \BibitemOpen
  \bibfield  {author} {\bibinfo {author} {\bibfnamefont {A.~W.}\ \bibnamefont {Harrow}}, \bibinfo {author} {\bibfnamefont {A.}~\bibnamefont {Hassidim}},\ and\ \bibinfo {author} {\bibfnamefont {S.}~\bibnamefont {Lloyd}},\ }\bibfield  {title} {\bibinfo {title} {Quantum algorithm for linear systems of equations},\ }\bibfield  {journal} {\bibinfo  {journal} {Physical Review Letters}\ }\textbf {\bibinfo {volume} {103}},\ \href {https://doi.org/10.1103/physrevlett.103.150502} {10.1103/physrevlett.103.150502} (\bibinfo {year} {2009})\BibitemShut {NoStop}%
\bibitem [{\citenamefont {Costa}\ \emph {et~al.}(2021)\citenamefont {Costa}, \citenamefont {An}, \citenamefont {Sanders}, \citenamefont {Su}, \citenamefont {Babbush},\ and\ \citenamefont {Berry}}]{costa2021optimalscalingquantumlinear}%
  \BibitemOpen
  \bibfield  {author} {\bibinfo {author} {\bibfnamefont {P.~C.~S.}\ \bibnamefont {Costa}}, \bibinfo {author} {\bibfnamefont {D.}~\bibnamefont {An}}, \bibinfo {author} {\bibfnamefont {Y.~R.}\ \bibnamefont {Sanders}}, \bibinfo {author} {\bibfnamefont {Y.}~\bibnamefont {Su}}, \bibinfo {author} {\bibfnamefont {R.}~\bibnamefont {Babbush}},\ and\ \bibinfo {author} {\bibfnamefont {D.~W.}\ \bibnamefont {Berry}},\ }\href {https://arxiv.org/abs/2111.08152} {\bibinfo {title} {Optimal scaling quantum linear systems solver via discrete adiabatic theorem}} (\bibinfo {year} {2021}),\ \Eprint {https://arxiv.org/abs/2111.08152} {arXiv:2111.08152 [quant-ph]} \BibitemShut {NoStop}%
\bibitem [{\citenamefont {Low}\ and\ \citenamefont {Su}(2024)}]{low2024quantumlinear}%
  \BibitemOpen
  \bibfield  {author} {\bibinfo {author} {\bibfnamefont {G.~H.}\ \bibnamefont {Low}}\ and\ \bibinfo {author} {\bibfnamefont {Y.}~\bibnamefont {Su}},\ }\href {https://arxiv.org/abs/2410.18178} {\bibinfo {title} {Quantum linear system algorithm with optimal queries to initial state preparation}} (\bibinfo {year} {2024}),\ \Eprint {https://arxiv.org/abs/2410.18178} {arXiv:2410.18178 [quant-ph]} \BibitemShut {NoStop}%
\bibitem [{\citenamefont {Nielsen}\ and\ \citenamefont {Chuang}(2010)}]{Nielsen_Chuang_2010}%
  \BibitemOpen
  \bibfield  {author} {\bibinfo {author} {\bibfnamefont {M.~A.}\ \bibnamefont {Nielsen}}\ and\ \bibinfo {author} {\bibfnamefont {I.~L.}\ \bibnamefont {Chuang}},\ }\href@noop {} {\emph {\bibinfo {title} {Quantum Computation and Quantum Information: 10th Anniversary Edition}}}\ (\bibinfo  {publisher} {Cambridge University Press},\ \bibinfo {year} {2010})\BibitemShut {NoStop}%
\bibitem [{\citenamefont {Preskill}(2015)}]{preskill_chp_2}%
  \BibitemOpen
  \bibfield  {author} {\bibinfo {author} {\bibfnamefont {J.}~\bibnamefont {Preskill}},\ }\href {https://www.preskill.caltech.edu/ph219/chap2_15.pdf} {\bibinfo {title} {Lecture notes for ph219/cs219: Quantum computation, chapter 2}} (\bibinfo {year} {2015}),\ \bibinfo {note} {accessed: 2025-01-31}\BibitemShut {NoStop}%
\bibitem [{\citenamefont {Uhlmann}(1976)}]{Uhlmann1976273}%
  \BibitemOpen
  \bibfield  {author} {\bibinfo {author} {\bibfnamefont {A.}~\bibnamefont {Uhlmann}},\ }\bibfield  {title} {\bibinfo {title} {The “transition probability” in the state space of a *-algebra},\ }\href {https://doi.org/https://doi.org/10.1016/0034-4877(76)90060-4} {\bibfield  {journal} {\bibinfo  {journal} {Reports on Mathematical Physics}\ }\textbf {\bibinfo {volume} {9}},\ \bibinfo {pages} {273} (\bibinfo {year} {1976})}\BibitemShut {NoStop}%
\end{thebibliography}%

\end{document}